\documentclass[letter,12pt]{amsart}
\usepackage{wrapfig}
\usepackage{amsfonts}
\usepackage{amssymb}
\usepackage{amsmath,tabu}
\usepackage{amsthm}
\usepackage{comment}
\usepackage{diagbox}
\usepackage{graphics}
\usepackage{epstopdf}
\usepackage{pst-all,psfrag}
\usepackage[unicode]{hyperref}

\usepackage{indentfirst}
\usepackage{subcaption}

\usepackage{natbib}

\usepackage{verbatim}
\usepackage{eurosym}
\usepackage[utf8]{inputenc}
\usepackage[english]{babel}
\usepackage[margin=0.99in]{geometry}
\usepackage{chngcntr}
\usepackage{apptools}
\usepackage{tabularx}
\usepackage{setspace}

\usepackage[textsize=footnotesize]{todonotes}

\usepackage{accents}
\usepackage{bbm}

\newcommand{\dbtilde}[1]{\accentset{\approx}{#1}}

\newcommand{\rank}{\mathrm{rank}}

\makeatletter
\def\@secnumfont{\bfseries}
\def\section{\@startsection{section}{1}%
\z@{1.0\linespacing\@plus\linespacing}{0.5\linespacing}%
{\normalfont\bfseries\centering}}
\makeatletter

\AtAppendix{\counterwithin{theorem}{section}}

\newtheorem{theorem}{Theorem}
\newtheorem{proposition}[theorem]{Proposition}

\newtheorem{lemma}[theorem]{Lemma}
\newtheorem{conjecture}[theorem]{Conjecture}

\newtheorem{corollary}[theorem]{Corollary}
\newtheorem{remark}[theorem]{Remark}

\theoremstyle{definition}
\newtheorem{procedure}{Procedure}
\newtheorem{definition}[theorem]{Definition}

\newcommand{\eps}{\varepsilon}
\newcommand{\p}{\mathfrak p}
\newcommand{\q}{\mathfrak q}
\newcommand{\T}{\mathcal T}
\newcommand{\J}{\mathbf J}
\newcommand{\V}{\mathcal V}
\newcommand{\R}{\mathcal R}
\renewcommand{\aa}{\mathfrak a}
\newcommand{\1}{I}

\renewcommand{\d}{\mathrm d}

\newcommand{\UU}{\mathbf U}
\newcommand{\VV}{\mathbf V}
\newcommand{\w}{\mathbf w}
\renewcommand{\v}{\mathbf v}
\renewcommand{\u}{\mathbf u}
\newcommand{\ba}{\mathbf a}
\newcommand{\bb}{\mathbf b}

\begin{document}
\title[Asymptotics of cointegration tests for high-dimensional VAR($k$)]{Asymptotics of cointegration tests for high-dimensional VAR($\mathbf{k}$)}

\author{Anna Bykhovskaya}
\address[Anna Bykhovskaya]{Duke University}
\email{anna.bykhovskaya@duke.edu}

\author{Vadim Gorin}
\address[Vadim Gorin]{University of California at Berkeley}
\email{vadicgor@gmail.com}

\thanks{The authors would like to thank Bruce Hansen, Alexei Onatski, associate editor Zhipeng Liao, and three anonymous referees for valuable comments and suggestions. The authors are grateful to Victor Kleptsyn for his help with the proof of Lemma \ref{Lemma_match_spaces}. Finally, the authors would like to thank Eszter Kiss for excellent research assistance. Gorin's work was supported by NSF grants DMS-1664619, DMS-1949820, and DMS-2246449, and BSF grant 2018248.}
\date{\today}
\maketitle

\begin{abstract}
	The paper studies nonstationary high-dimensional vector autoregressions of order $k$, VAR($k$). Additional deterministic terms such as trend or seasonality are allowed. The number of time periods, $T$, and the number of coordinates, $N$, are assumed to be large and of the same order. Under this regime the first-order asymptotics of the Johansen likelihood ratio (LR), Pillai--Bartlett, and Hotelling--Lawley tests for cointegration are derived: the test statistics converge to nonrandom integrals. For more refined analysis, the paper proposes and analyzes a modification of the Johansen test. The new test for the absence of cointegration converges to the partial sum of the Airy$_1$ point process. Supporting Monte Carlo simulations indicate that the same behavior persists universally in many situations beyond those considered in our theorems.

The paper presents empirical implementations of the approach for the analysis of S$\&$P$100$ stocks and of cryptocurrencies. The latter example has a strong presence of multiple cointegrating relationships, while the results for the former are consistent with the null of no cointegration.
\end{abstract}


\newpage

\section{Introduction}

Starting with the pioneering work of \citet{sims}, vector autoregressions (VARs) became a workhorse model in macroeconomics and other fields. Many key time series in macroeconomics and finance (e.g., consumption and output) are nonstationary, and the properties of VARs can be very different depending on whether one is dealing with a stationary or nonstationary series. Moreover, there is a further subdivision to be accounted for in the case of nonstationary series: it is important to understand whether the data are cointegrated---that is, whether there exists a stationary nontrivial linear combination within the considered series (e.g., the log of consumption minus the log of output is stationary while the series themselves have unit roots).

Classical tools  for testing cointegration  (see, e.g., \citet{johansen_book}, \citet{maddala}, and \citet{juselius}) fail to achieve the desired finite sample performance when the number of time series, $N$, is large. Thus, they are not commonly used in such settings, and the design of proper tools to handle cointegration under a large $N$ remained an open problem for years (see, e.g., \citet[Sections 2.3.3, \,2.4]{choi15}). Recently  \citet{onatski_ecta, onatski_joe} and \citet{BG} have opened a new avenue based on the ``$T/N$ converging to a constant'' asymptotic regime. However, the testing procedures of these texts cover only  VAR($1$), while, in practice, researchers rarely confine themselves to VARs of order $1$, instead usually considering at least two lags. Indeed, as noted already in \citet{pagan}, ``most applications of Sims’ methodology have put the number of lags between four and ten.'' Since \citet{pagan} the lengths of available time series and computing power have only increased, thus allowing researchers to work with even more complex models. Hence,  it is important to generalize and extend the above papers to a VAR($k$) setting, which is the main topic of our text.

Our paper analyses a family of tests for the absence of cointegration for nonstationary VAR($k$), such as the Johansen likelihood ratio (LR) test  (\citet{johansen1988, johansen1991}) and related Hotelling--Lawley and Pillai--Bartlett tests (see, e.g., \citet{gonzalo_pitarakis1995} and references therein) as $N$ and $T$ jointly and proportionally go to infinity. The shared feature of these tests is that their statistics are based on the squared sample canonical correlations between certain transformations of current changes and past levels of the data. The main contribution of our paper is in the asymptotic analysis of these canonical correlations.
First, we show that for VAR($k$) with general $k$, under the null of no cointegration (and some additional technical conditions) the empirical distribution of the squared sample canonical correlations converges to the Wachter distribution. As a corollary, we deduce the first-order deterministic limits of the above test statistics. Second, we introduce a modification of the testing procedure and prove much more refined results in the modified setting. By computing the exact asymptotic behavior of the probability distributions of individual canonical correlations after proper recentering and rescaling, we are able to compute the critical values for the test of no cointegration with correct asymptotic size as $N$ and $T$ jointly and proportionally go to infinity.

We remark that there is a wide scope of literature devoted to the corrections of Johansen's LR test and its relatives (originally developed based on fixed $N$, large $T$ asymptotics) for large values of $N$ (see, e.g., \citet{Reinsel_Ahn}, \citet{johansen_correction}, \citet{Swensen}, \citet{cav_et_all}, and \citet{onatski_joe}). The distinguishing feature of our work is that we are not trying to correct the finite $N$ asymptotic statements, which stop working for large $N$, by introducing various empirical adjustments. Instead, we develop a theoretical framework for working with the large $N$ case directly. One advantage is that our approach explains the general phenomenology and predicts the asymptotic behavior. As a result, the empirical or simulational adjustments for particular values of the parameters of the model are no longer needed.

To achieve the above, in our proofs we use the VAR($1$) results of \citet{BG} as a cornerstone. The main technical work is devoted to producing recursive arguments, which reduce the VAR($k)$ behavior to that of VAR($k-1$) and eventually to VAR($1$). The central role is played by highly nontrivial projections from the group of orthogonal $T\times T$ matrices to the smaller subgroup of orthogonal $(T-N)\times (T-N)$ matrices. While such projections have previously been used in asymptotic representation theory, to the authors' knowledge, this is their first appearance in the econometrics or statistics context. Thus, many new properties of those projections need to be developed in our framework.

\medskip

The rest of the paper is organized as follows. Section \ref{Section_setting} describes our setting and provides the first asymptotic results. Section \ref{Section_modified_test} constructs our modified test and computes its asymptotics, while Section \ref{Section_MC} presents supporting Monte Carlo simulations. Section \ref{Section_empirical} illustrates our theoretical findings on S\&P100 data and on the prices of cryptocurrencies. Finally, Section \ref{Section_conclusion} concludes. All proofs are in Sections \ref{Section_Jacobi}--\ref{Section_small_rank}. The accompanying R package is available at the Github \url{https://github.com/eszter-kiss/Largevars}.

\section{First-order asymptotics of sample canonical correlations}
\label{Section_setting}

We consider an $N$-dimensional vector autoregressive process of order $k$, VAR($k$), based on a sequence of i.i.d.~mean zero Gaussian\footnote{We expect that all our results continue to hold for non-normally distributed errors as long as they have enough moments (cf.\ such distribution-independence results in other high-dimensional models, as in \cite{ErdosYau}, \cite{Tao_Vu}, \cite{HanPanYang}, and \cite{FanYang}).} errors $\{\eps_t\}$ with nondegenerate covariance matrix $\Lambda$. That is, written in the error correction form,
\begin{equation}\label{var_k}
\Delta X_t=\sum\limits_{i=1}^{k-1}\Gamma_i\Delta X_{t-i}+\Pi X_{t-k}+\Phi D_t+\eps_t,\qquad t=1,\ldots,T,
\end{equation}
where $\Delta X_t:=X_t-X_{t-1}$, $D_t$ is a $d_D$-dimensional vector of deterministic terms, such as a constant, a trend or seasonality (extra explanatory variables are also allowed as long as they are observed), and $\Gamma_1,\ldots,\Gamma_{k-1},\,\Pi,\,\Phi$ are unknown parameters. The process is initialized at fixed $X_{1-k},\ldots,X_0$. We do not impose any restrictions on $\Lambda$; thus, we allow for arbitrary correlations across coordinates of $X_t$.
In contrast, many previous approaches rely on specific properties of the covariance matrix $\Lambda$; see, e.g., \citet{Breitung_Pesarann_2008},
\citet[Section 7]{Bai_Ng_2008} and \citet{ZhangPanGao_2018}.

\begin{remark}
Alternatively, the error correction form can be written as
$$\Delta X_t=\Pi X_{t-1}+\sum\limits_{i=1}^{k-1}\tilde{\Gamma}_i\Delta X_{t-i}+\Phi D_t+\eps_t,\qquad t=1,\ldots,T,$$
so that $\tilde{\Gamma}_i=\Gamma_i-\Pi$. Whether we use the former (Eq.~\eqref{var_k}) or the latter form does not affect our results. The testing procedures of our interest are based on the residuals from regressing $X_{t-k}$ on $\Delta X_{t-1},\ldots,\Delta X_{t-k+1}$, which are the same as the residuals from regressing $X_{t-1}=X_{t-k}+\Delta X_{t-1}+\ldots+\Delta X_{t-k+1}$ on $\Delta X_{t-1},\ldots,\Delta X_{t-k+1}$.
\end{remark}

We are interested in the behavior of the squared sample canonical correlations between transformed past levels (lags) and changes (first differences) of the data $X_t$. As shown in \citep{johansen1988, johansen1991} (see also \citet{Anderson}), the correlations are related to whether the process is cointegrated. To be more specific, they appear in the likelihood ratio test for the presence and rank of the cointegration. Let us formally define these correlations.  Here and below $^*$ denotes matrix transposition.

\begin{procedure}[\citet{johansen1991}]\label{sscc_J}
Let $Z_{0t}=\Delta X_t$, $Z_{1t}=(\Delta X_{t-1}^{*},\ldots,\Delta X_{t-k+1}^{*},D_t^*)^{*}$, and $Z_{kt}=X_{t-k}$.  We regress lags $Z_{kt}$ and changes $Z_{0t}$ on regressors $Z_{1t}$ (lagged changes and deterministic terms) and define the residuals
\begin{equation}\label{res_J}
R_{it}=Z_{it}-\left(\sum\limits_{\tau=1}^{T} Z_{i\tau}Z_{1\tau}^{*}\right)\left(\sum\limits_{\tau=1}^{T} Z_{1\tau}Z_{1\tau}^{*}\right)^{-1}Z_{1t},\quad i=0,k.
\end{equation}
Define further $N\times N$ matrices $S_{ij}:=\sum\limits_{t=1}^{T} R_{it}R_{jt}^{*},\,i,j=0,k$ and finally set
$$
 \mathcal{C}=S_{kk}^{-1}S_{k0} S_{00}^{-1} S_{0k}.
$$
The $N$ eigenvalues $\lambda_1\ge \lambda_2\ge\dots\ge \lambda_N$ of $\mathcal C$ are squared sample canonical correlations of $R_0$ and $R_k$, where $R_i$ is $N\times T$ matrix composed of columns $R_{it},\,i=0,k$.
\end{procedure}

Johansen's LR statistic for testing the hypothesis $\rank(\Pi)\le r_1$ (at most $r_1$ cointegrating relationships) versus the alternative $\rank(\Pi)\in (r_1,r_2]$ (between $r_1$ and $r_2$ cointegrating relationships) with $r_2>r_1$ has the form\footnote{We omit a usual scaling factor of $T$ for the statistics \eqref{eq_LR_statistic}, \eqref{eq_PB_statistic}, and \eqref{eq_HW_statistic}.}
\begin{equation}
\label{eq_LR_statistic}
 \sum_{i=r_1+1}^{r_2} \ln(1-\lambda_i);
\end{equation}
the Pillai--Bartlett statistic  is
\begin{equation}
\label{eq_PB_statistic}
 \sum_{i=r_1+1}^{r_2} \lambda_i;
\end{equation}
and Hotelling--Lawley statistic is
\begin{equation}
\label{eq_HW_statistic}
 \sum_{i=r_1+1}^{r_2} \frac{\lambda_i}{1-\lambda_i}.
\end{equation}
See \citet{gonzalo_pitarakis1995} for a discussion and many references about these statistics.

In Theorem \ref{Theorem_empirical} we show that the empirical measure of eigenvalues of $\mathcal C$ converges (weakly in probability) to the Wachter distribution. The theorem generalizes the results of \citet{onatski_ecta} from the VAR($1$) to the VAR($k$) setting.\footnote{While \citet[Theorem 1]{onatski_ecta} allows the data to be VAR($k$) under restriction \eqref{eq_rank_restriction}, they construct the matrix $\mathcal C$ involved in the statistical testing procedures as if the data were VAR($1$). That is, the formal procedure is based on a misspecified VAR($1$) setting, while we use the true VAR($k$) procedure. As illustrated in Section \ref{section_MC_var_order}, using an underspecified VAR can lead to severe size distortions.}

\begin{definition}
The Wachter distribution is a probability distribution on $[0,1]$ that depends on two parameters $\p>1$ and $\q>1$ and has density
\begin{equation}
\label{eq_Jacobi_equilibrium_1}
 \mu_{\p,\q}(x) = \frac{\p+\q}{2\pi} \cdot \frac{\sqrt{(x-\lambda_-)(\lambda_+-x)}}{x (1-x)} \mathbf 1_{[\lambda_-,\lambda_+]}\,,
\end{equation}
where the support $[\lambda_-,\lambda_+]\subset (0,1)$ of the measure is defined via
\begin{equation}
 \lambda_\pm=\frac{1}{(\p+\q)^2}\left(\sqrt{\p(\p+\q-1)}\pm \sqrt{\q}  \right)^2.
\end{equation}
\end{definition}

\begin{theorem}\label{Theorem_empirical} Let $X_t$ follow Eq.\ \eqref{var_k}. Suppose that $k$ is fixed and, as $N\to\infty$,
\begin{equation}
\label{eq_limit_regime}
 \lim_{N\to\infty} \frac{T}{N}=\tau>(k+1) \qquad \text{ and}
\end{equation}
\begin{equation}
\label{eq_rank_restriction}
\lim_{N\to\infty} \frac{1}{N} \bigl( \rank(\Pi)+\rank(\Gamma_1)+\rank(\Gamma_2)+\dots+\rank(\Gamma_{k-1})+d_D\bigr)=0.
\end{equation}
Then, for each continuous function $f(x)$ on $x\in[0,1]$, we have
\begin{equation}
\label{eq_LLN_in_prob}
 \lim_{N\to\infty} \frac{1}{N} \sum_{i=1}^N f(\lambda_i)=\int_0^1 f(x) \mu_{2, \tau-k}(x)\, \d x, \quad \text { in probability.}
\end{equation}
Equivalently,  the empirical measure of eigenvalues  $\lambda_1\ge\dots\ge\lambda_N$ of $\mathcal C$ converges (weakly in probability) to the Wachter distribution of density $ \mu_{2, \tau-k}$.
\end{theorem}

Imposing assumption \eqref{eq_rank_restriction} can be viewed as a dimension reduction (cf.~sparsity assumption). Approximating data with low-rank matrices is a widely used and powerful technique in data science, in machine learning applications such as recommender systems (e.g., movie preference recognition), and in computational mathematics. We refer the reader to \citet{lowrk_highdim} for theoretical explanations of the suitability of low-rank models and many references to situations in which they are very efficient. In our particular context, the number of unknown parameters in the VAR model \eqref{var_k} is proportional to $N^2$, and we have access to $N T$ observations. Since $N^2$ and $NT$ are of the same order in the asymptotic regime \eqref{eq_limit_regime}, the model \eqref{var_k} can overfit the data. We view the rank restriction \eqref{eq_rank_restriction} as a natural way to avoid overfitting\footnote{An alternative way to introduce a low-rank assumption into the VAR model is, instead of using the error correction form in \eqref{var_k}, to rewrite the evolution as $X_t=\sum\limits_{i=1}^{k} A_i X_{t-i}+\Phi D_t+\eps_t$. Then, in the spirit of factor models, one can impose the low-rank assumption on $A_i,\,i=1,\ldots,k$. Notice that $A_i=\Gamma_i-\Gamma_{i-1}$ for $i=2,\ldots,k-1$, so that low-rank assumptions on higher-order lags in this and our setting are related. This alternative low-rank restriction complements ours via the rank of $\Pi$: in our setting, the number of cointegrating relationships grows sublinearly in $N$, while in the alternative setting, this number is close to $N$. While none of our theorems directly cover the factor setting, numeric simulations in Section \ref{section_power} indicate that the tests that we develop remain useful.} (cf.~the discussion in \citet{JASA_low_rank_VAR} and \citet{Wang_Tsay_low_rank_VAR}). Section \ref{Section_empirical} illustrates that the results obtained under this assumption are consistent with the behavior of large-dimensional financial datasets. Another setting in which we can expect \eqref{eq_rank_restriction} to be satisfied is when there are a few special coordinates in $X_t$, e.g., some macroeconomic indicators, that mostly drive the behavior of the entire vector $X_t$. This would correspond to the case where the columns of $\Gamma_i$ corresponding to those indicators are nonzero while the other columns are zero.

\smallskip

Figure \ref{small_rk_pic} illustrates Theorem \ref{Theorem_empirical} for independent standard normal errors and $k=2,\,N=150,\,T=1500$. The parameters are $\Phi D_t=1_N$, $\Gamma_1=0.95E_{12}$, $\Pi=-0.1E_{\cdot1}$, where $1_N$ is an $N\times 1$-column matrix of ones, $E_{12}$ is an $N\times N$ matrix with one at the intersection of the 1st row and the 2nd column and zeros everywhere else, and $E_{\cdot1}$ is an $N\times N$ matrix with ones in the first column and zeros everywhere else. Thus, all the matrices have rank one. The parameters of the Wachter distribution are $\p=2,\,\q=T/N-k=8$. The single separated (rightmost) eigenvalue corresponds to $\rank(\Pi)=1$, i.e.,~one cointegrating relationship. Generally, we expect that if there are $r$ separated eigenvalues and the value of $k$ in Procedure \ref{sscc_J} is correctly specified, then there are at least $r$ cointegrating relationships.

\begin{figure}[h]
{\scalebox{0.7}{\includegraphics{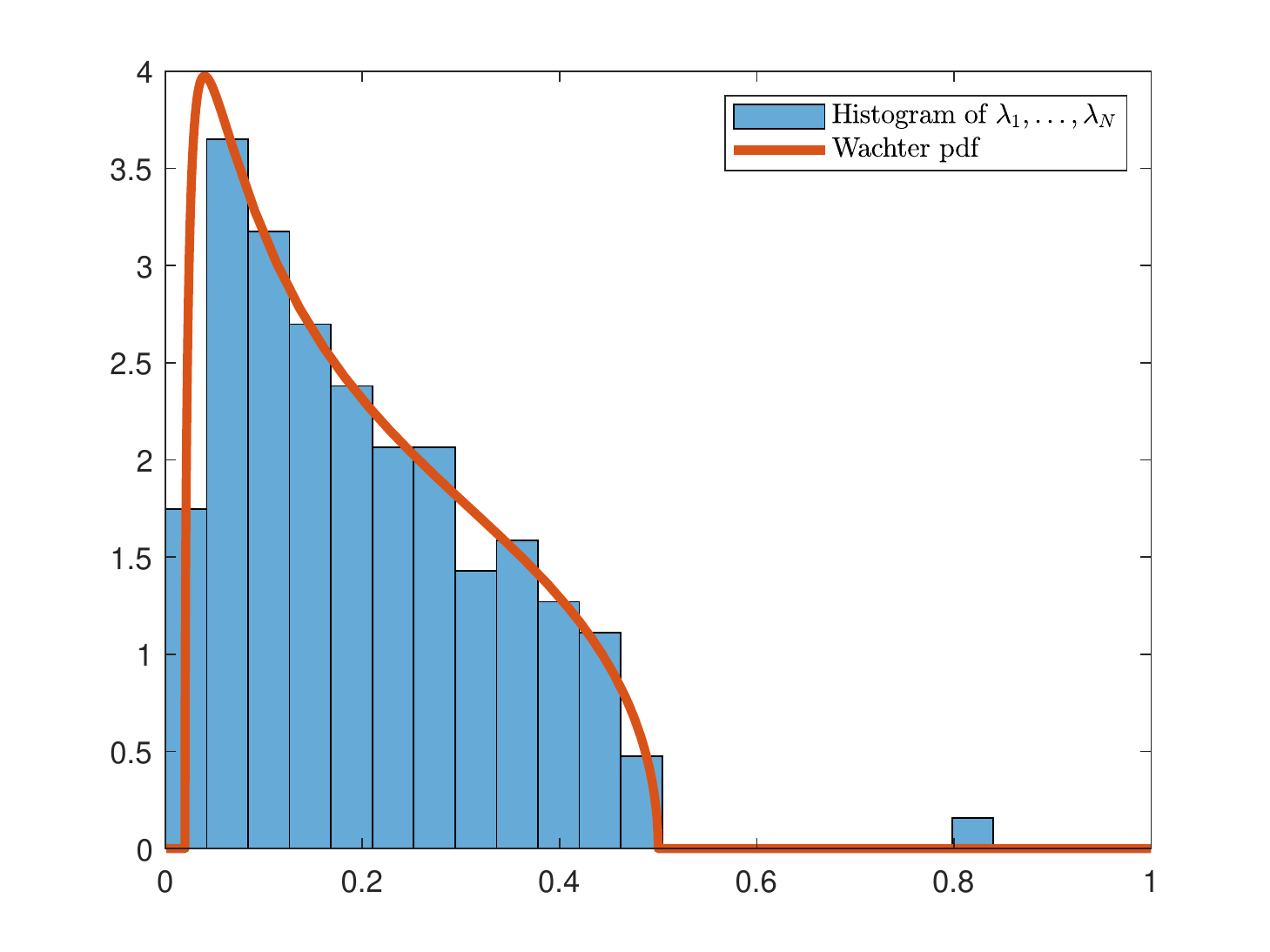}}}
\caption{Illustration of Theorem \ref{Theorem_empirical}: Eigenvalues and Wachter distribution.
Data generating process: $\Delta X_{t}=1_N+0.95E_{12}\Delta X_{t-1}-0.1E_{\cdot1}X_{t-2}+\eps_{t}$, \mbox{$T=1500$}, $N=150$, $\eps_{it}\thicksim$ i.i.d.~$\mathcal{N}(0,1)$.}
\label{small_rk_pic}
\end{figure}

The proof of Theorem \ref{Theorem_empirical} is based on treating the setting of \eqref{var_k} as a small-rank perturbation of a more restrictive setting analyzed in Section \ref{Section_modified_test}. The small-rank assumption in \eqref{eq_rank_restriction} is crucial for the validity of the theorem, and we expect that the asymptotic behavior changes in situations when \eqref{eq_rank_restriction} fails. This expectation is supported by the Monte Carlo experiment in Section \ref{section_mc_small_rank}. In contrast, the assumption that $T>(k+1)N$ can potentially be relaxed. When $T<(k+1)N$, the matrix $\mathcal C$ has deterministic eigenvalues equal to $1$, which should be taken into account in $N\to\infty$ asymptotics. This case can be also addressed by our methods, but we do not continue in this direction.

An important corollary of Theorem \ref{Theorem_empirical} is that it provides the asymptotic behavior of various tests constructed from eigenvalues of $\mathcal C$.

\begin{corollary} \label{Corollary_test_stat} Under the assumptions of Theorem \ref{Theorem_empirical}, suppose that the ranks $r_1=r_1(N)$ and $r_2=r_2(N)$ are such that
$$
 \lim_{N\to\infty} \frac{r_1}{N}=\rho_1,\quad \lim_{N\to\infty} \frac{r_2}{N}=\rho_2.
$$
Let $F(x)=\int_x^1  \mu_{2, \tau-k}(z)\, \d z$.
Then we have convergence in probability for the test statistics:
\begin{align*}
 &\lim_{N\to\infty} \frac{1}{N}
 \sum_{i=r_1+1}^{r_2} \ln(1-\lambda_i)=\int_{F^{-1}(\rho_2)}^{F^{-1}(\rho_1)} \ln(1-x) \mu_{2, \tau-k}(x)\, \d x,\ \quad \text{ if }\,\, \rho_1>0;
\\
 &\lim_{N\to\infty}  \frac{1}{N}\sum_{i=r_1+1}^{r_2} \lambda_i=\int_{F^{-1}(\rho_2)}^{F^{-1}(\rho_1)} x \mu_{2, \tau-k}(x)\, \d x;
\\
 &\lim_{N\to\infty}  \frac{1}{N}\sum_{i=r_1+1}^{r_2} \frac{\lambda_i}{1-\lambda_i}=\int_{F^{-1}(\rho_2)}^{F^{-1}(\rho_1)} \frac{x}{1-x} \mu_{2, \tau-k}(x)\, \d x,\ \quad \text{ if }\,\, \rho_1>0;
\end{align*}
and asymptotic inequalities: for each $\eps>0$,
\begin{align*}
&\lim_{N\to\infty}\mathrm{Prob}\left( \frac{1}{N}
 \sum_{i=r_1+1}^{r_2} \ln(1-\lambda_i)\le \int_{F^{-1}(\rho_2)}^{1} \ln(1-x) \mu_{2, \tau-k}(x)\, \d x+\eps\right)=1, \quad \text{ if }\,\, \rho_1=0;
\\
&\lim_{N\to\infty}\mathrm{Prob}\left(  \frac{1}{N}\sum_{i=r_1+1}^{r_2} \frac{\lambda_i}{1-\lambda_i}\ge \int_{F^{-1}(\rho_2)}^1 \frac{x}{1-x} \mu_{2, \tau-k}(x)\, \d x-\eps\right)=1,\ \quad \text{ if }\,\, \rho_1=0.
\end{align*}
\end{corollary}

Note that, when $\rho_1=0$, for the Johansen LR and Hotelling--Lawley (HW) statistics, we obtain inequalities rather than equalities. This is because of the singularity of $\ln(1-\lambda)$ and $\frac{\lambda}{1-\lambda}$ at $\lambda=1$: while Eq.~\eqref{eq_LLN_in_prob} controls average behavior, it does not control individual eigenvalues. Hence, the largest eigenvalue $\lambda_1$ can be arbitrarily close to $1$, so that the LR and HW statistics reach large negative and positive values, respectively. However, for similar statistics in the modified setting of the next section the inequalities turn into equalities (as can be proven by combining Theorem \ref{Theorem_vark_approximation} and Proposition \ref{Theorem_Jacobi_as}).

For commonly used tests, one often takes $r_2=N$ (i.e., $H_1:\rank(\Pi)\leq N$), in which case $F^{-1}(\rho_2)=0$.
We also remark that the integrals in Corollary \ref{Corollary_test_stat} can be explicitly computed in many situations. For instance, the one appearing in the asymptotic of the Pillai--Bartlett statistic for $\rho_1=0$, $\rho_2=1$ is
$$
 \int_{0}^{1} x\, \mu_{2, \tau-k}(x)\, \d x=\frac{2}{\tau+2-k}.
$$

There are several applications of Theorem \ref{Theorem_empirical} and Corollary \ref{Corollary_test_stat}:
\begin{itemize}
\item They can be used for validation of the applicability of model \eqref{var_k} to a given dataset. Namely, if a VAR($k$) model with low-rank matrices $\Gamma_i$ and $\Pi$ agrees with data, then irrespective of the true values of these parameters, we expect to see the Wachter distribution in the histogram of $\lambda_i$, $1\le i \le N$. In Section \ref{Section_empirical} we perform such a validation on S$\&$P$100$ and cryptocurrency data sets for VAR($k$) with $1\le k \le 4$ and observe a remarkable match. (VAR($1$) for S$\&$P$100$ is also reported in \cite[Figure 7]{BG}.)

\item They can be used as a screening device for preliminary conclusions about the rank of $\Pi$: If the rank is finite, then for any $r_1$ and $r_2$ we should be in the $\eps$-neighborhood of the limits in Corollary \ref{Corollary_test_stat}.

\item As explained in \citet{onatski_ecta}, such results can be used to explain overrejection in some of the widely used tests for the rank of $\Pi$.
\end{itemize}

To draw further economical and statistical conclusions and to develop precise statistical tests and their critical values, one needs to go beyond the first-order asymptotic results of Theorem \ref{Theorem_empirical} and Corollary \ref{Corollary_test_stat}. In the next section we introduce relevant modifications and develop appropriate second-order asymptotics.

\section{Cointegration test: Second-order asymptotics}
\label{Section_modified_test}

In the regime of $N$ and $T$ growing simultaneously and proportionally, the first-order asymptotics of tests based on the squared sample canonical correlations are given in Corollary \ref{Corollary_test_stat}. To perform testing and be able to reject at a given significance level, we need to be more precise and find a centered limit, which would be a random variable rather than a constant. To do this, we need to impose additional conditions on $\Gamma_i$, $D_t$, and $\Phi$ in Eq.~\eqref{var_k}. Let us first describe the modified procedure and then state the asymptotic results.

\subsection{Test}
\label{Section_modified_test_statistics}

We restrict our attention to the case $D_t=1$, i.e.,
\begin{equation}\label{var_k_restr}
\Delta X_t=\mu+\sum\limits_{i=1}^{k-1}\Gamma_i\Delta X_{t-i}+\Pi X_{t-k}+\eps_t,\qquad t=1,\ldots,T.
\end{equation}

The null hypothesis of no cointegration is $H_0:\:{\rm rank}(\Pi)=0$ or $\Pi\equiv0$. The complement to $H_0$ is ${\rm rank}(\Pi)>0$. However, to design our test we use an alternative hypothesis:
$$
 H(r):\quad {\rm rank}(\Pi)\in[1,r].
$$
As in \citet{BG}, our test is based on a modification of the Johansen LR test. The Johansen LR test for the original $H_0$ (i.e., $\Pi\equiv0$) versus $H(r)$ is
\begin{equation}
\label{eq_Joh_statistic_refined}
 \sum_{i=1}^{r} \ln(1-\lambda_i),
\end{equation}
where $\lambda_1,\lambda_2,\ldots$ are defined in Procedure \ref{sscc_J}. Let us describe how our modified test proceeds.

\begin{procedure}\label{sscc_BG}
\textbf{Step 1.} De-trend the data and define
\begin{equation}
\label{eq_detrending}
 \tilde X_t = X_{t-1} - \frac{t-1}{T} (X_T-X_0).
\end{equation}
Note that we do a time shift in line with the notation in \citet{BG}.

\textbf{Step 2.} Define regressors and dependent variables: For any $a\in\mathbb Z$, set
$$
 a\mid T= a+ k T,\quad \text{where } k\in\mathbb Z\text{ is such that } a+ k T\in \{1,2,\dots,T\}.
$$
Define
$$
 \tilde{Z}_{0t}=\Delta X_{t\mid T}\equiv\Delta X_{t},\quad\tilde{Z}_{kt}=\tilde X_{t-k+1\mid T},\quad\tilde{Z}_{1t}=(\Delta X_{t-1\mid T}^{*},\ldots,\Delta X_{t-k+1\mid T}^{*},1)^{*}.
$$
The main difference between $\tilde{Z}_{it}$ and $Z_{it}$ from Procedure \ref{sscc_J} is the usage of cyclic indices: values at $t=0,-1,\ldots$ are replaced by values at $t=T,T-1,\ldots$.

\textbf{Step 3.} Calculate the residuals from regressions $\tilde{Z}_{0t}$ on $\tilde{Z}_{1t}$ and $\tilde{Z}_{kt}$ on $\tilde{Z}_{1t}$:
\begin{equation}\label{res_BG}
\tilde{R}_{it}=\tilde{Z}_{it}-\left(\sum\limits_{\tau=1}^{T} \tilde{Z}_{i\tau}\tilde{Z}_{1\tau}^{*}\right)\left(\sum\limits_{\tau=1}^{T} \tilde{Z}_{1\tau}\tilde{Z}_{1\tau}^{*}\right)^{-1}\tilde{Z}_{1t},\quad i=0,k.
\end{equation}

\textbf{Step 4.} Calculate the squared sample canonical correlations between $\tilde{R}_0$ and $\tilde{R}_k$, where $\tilde{R}_i$ is an $N\times T$ matrix composed of columns $\tilde{R}_{it},\,i=0,k$. That is, define
\begin{equation}\begin{split}\label{S_matrices}
\tilde{S}_{ij}=\sum\limits_{t=1}^{T} \tilde{R}_{it} \tilde{R}^{\ast}_{jt},\quad i,j=0,k, \qquad \text{ and}
\end{split}\end{equation}
\begin{equation}\label{tilde_C}
\tilde{\mathcal C}=\tilde{S}_{k0}\tilde{S}^{-1}_{00}\tilde{S}_{0k}\tilde{S}^{-1}_{kk}.
\end{equation}
Then, calculate $N$ eigenvalues $\tilde{\lambda}_1\geq\ldots\geq\tilde{\lambda}_N$ of the matrix $\tilde{\mathcal C}$. The eigenvalues solve the equation
\begin{equation}
\label{eq_vark_eig}
 \det( \tilde S_{k0} \tilde S_{00}^{-1} \tilde S_{0k}-\tilde{\lambda} \tilde S_{kk})=0.
\end{equation}

\textbf{Step 5.} Form the test statistic
\begin{equation}\label{LR_NT}
LR_{N,T}(r)=\sum\limits_{i=1}^{r}\ln(1-\tilde{\lambda}_i).
\end{equation}
The subscript $N,T$ in \eqref{LR_NT} indicates that we modify the Johansen LR test to develop the large $N,T$ asymptotics. This statistic after centering and rescaling will be compared with appropriate critical values to decide whether one can reject $H_0$ (see Theorem \ref{Theorem_J_stat}). Visually, rejections correspond to the case when the largest eigenvalues are separated from the rest (as in Figure \ref{small_rk_pic}).

One can also consider other functions of largest eigenvalues $\tilde{\lambda}_1,\tilde{\lambda}_2,\ldots$ such as Pillai--Barlett or Hotelling--Lawley statistics. The asymptotic behavior in those cases can be derived in the same way as we treat statistic \eqref{LR_NT} in Theorem \ref{Theorem_J_stat}.

\end{procedure}

An alternative way to write residuals $\tilde{R}_{i},\,i=0,k$ is via an orthogonal projector: Let $\mathcal W$ be a linear subspace of dimension $N(k-1)+1$ in $T$-dimensional vector space, spanned  by vector $(1,1,\dots,1)$ and all rows of matrices $(\Delta X) (L_c^{i})^*$, $1\le i\le (k-1)$, where $L_c$ is a cyclic version of the conventional lag operator and $L_c^{i}$ is its $i$th power, that is, the cyclic lag applied $i$ times. The cyclic lag operator $L_c$ maps a vector $(x_1,x_2,\dots,x_T)$ to $(x_{T},x_1,x_2,\dots,x_{T-1})$. Let $P_{\bot \mathcal W}$ denote the projector on orthogonal complement to $\mathcal W$. Then,
\begin{equation}
  \tilde R_0=(\Delta X) P_{\bot \mathcal W}, \qquad \tilde R_k= \tilde X (L_c^{k-1})^* P_{\bot \mathcal W}.
\end{equation}

\subsection{Second-order asymptotics}\label{Section_asy_results}

In this section we show that, under additional restrictions, the eigenvalues $\tilde{\lambda}_i,\,{i=1,\ldots,N}$ are very close (up to $N^{-1+\epsilon}$ for arbitrary $\epsilon>0$) to a known random matrix distribution. From this result we deduce our main theorem (Theorem \ref{Theorem_J_stat}), which gives the large $N,T$ limit of the test statistic $LR_{N,T}(r)$ in Eq.~\eqref{LR_NT}. Before we formally state the results, let us define the relevant random matrix distributions.

\subsubsection{Definitions}

\begin{definition} \label{Definition_Jacobi}
The (real) Jacobi ensemble $\J(N;p,q)$ is a distribution on $N\times N$ real symmetric matrices $\mathcal M$ of density proportional to
 \begin{equation}
  \label{eq_Jacobi_def}
  \det(\mathcal M)^{p-1} \det(\1_N-\mathcal M)^{q-1}\, d\mathcal M,\qquad 0<\mathcal M< \1_N,
 \end{equation}
 with respect to the Lebesgue measure, where $p,q>0$ are two parameters, $\1_N$ is the $N\times N$ identity matrix, and $0< \mathcal M < \1_N$ means that both $\mathcal M$ and $\1_N-\mathcal M$ are positive definite.
\end{definition}

The Jacobi ensemble is a generalization of the Beta distribution to the space of square matrices (when $N=1$, we obtain the Beta distribution). It plays a prominent role in statistics; e.g., it appears in canonical correlation analysis for independent data sets and in multivariate analysis of variance (see, e.g., \citet{Muirhead_book}).

\begin{definition} The Airy$_1$ point process is a random infinite sequence of reals
$$
\aa_1>\aa_2>\aa_3>\dots
$$
that can be defined through the following proposition.

\begin{proposition}[\citet{Forrest_spectr},\citet{Tracy_Widom}] \label{Proposition_Airy_Gauss} Let $X_N$ be an $N\times N$ matrix of i.i.d.~$\mathcal{N}(0,2)$ Gaussian random variables and let $\mu_{1;N}\ge \mu_{2;N}\ge \dots \mu_{N;N}$ be eigenvalues of $\frac{1}{2}\left(X_N+X_N^*\right)$. Then, in the sense of convergence of finite-dimensional distributions,
	\begin{equation}
	\label{eq_GOE_to_Airy}
	\lim_{N\to\infty} \left\{N^{1/6}\left(\mu_{i;N}-2\sqrt{N}\right) \right\}_{i=1}^N = \{ \aa_i\}_{i=1}^\infty.
	\end{equation}
\end{proposition}
\end{definition}

The marginals of the Airy$_1$ point process can be calculated via various methods (see, e.g., \citet{forrest} for more details).

\subsubsection{Theorems}

The null $H_0$ for \eqref{var_k_restr} is not a point hypothesis, as it does not specify $\Gamma_i,\,i=1,\ldots,k-1$. A simplifying procedure when we are faced with such a composite space of the maintained hypothesis is to assume some fixed values of the parameters as a proxy for the null hypothesis. Along these lines, for the next theorems we are going to introduce additional restrictions and specify the values of $\Gamma_i,\, i=1,\ldots,k-1$. Thus, our model is going to be fully specified (up to a constant $\mu$, which will disappear in the testing procedure).
We proceed to implement this approach in testing the
hypothesis of no cointegration and introduce the restricted $\widehat H_0$\footnote{We discuss the consequences of using $\widehat H_0$ for testing the null $H_0$ after Theorem \ref{Theorem_J_stat}.}:
 \begin{equation}\label{eq_strong_vark}
      \widehat H_0:\, \Pi=\Gamma_1=\Gamma_2=\dots=\Gamma_{k-1}=0.
     \end{equation}
In other words, under $\widehat H_0$ the data generating process turns into
\begin{equation}\label{eq_H0}
  \Delta X_t=\mu+\eps_t,\qquad t=1,\ldots,T,
\end{equation}
where $\mu$ is an (unknown) $N$-dimensional vector.

\begin{theorem}
\label{Theorem_vark_approximation}
Fix $C>0$, and suppose that $T,N\to\infty$ in such a way that $\frac{T}{N}\in[k+1+C^{-1},C]$. For the data generating process \eqref{var_k_restr} with restrictions $\widehat H_0$ given by \eqref{eq_strong_vark}, one can couple (i.e.,~define on the same probability space) the eigenvalues $\tilde{\lambda}_1\ge \tilde{\lambda}_2\ge\ldots\ge \tilde{\lambda}_N$ of the matrix $\tilde S_{k0} \tilde S_{00}^{-1} \tilde S_{0k}\tilde S_{kk}^{-1}$ and eigenvalues $x_1\ge \dots\ge x_N$ of the Jacobi ensemble $\J(N;\frac{N}{2}, \frac{T-(k+1)N}{2})$ in such a way that, for each $\epsilon>0$, we have\footnote{One can show that the probability in \eqref{eq_prob_to_1} is exponentially close to $1$: there exists a constant $\delta>0$, which depends on $\epsilon$, $C$, and $k$, such that, for all $N$ and $T$ satisfying $\frac{T}{N}\in[k+1+C^{-1},C]$, the probability under the limit in Eq.~\eqref{eq_prob_to_1} is larger than $1-\delta^{-1}\exp(\delta^{-1} N^\delta)$. Analyzing the proof of Theorem \ref{Theorem_vark_approximation}, we can obtain this inequality by combining \eqref{eq_MM_bound} with large deviations bounds for the smallest and largest eigenvalues of the Jacobi ensemble (see e.g., \citet[Section 2.6.2]{anderson2010introduction} for the latter).}
\begin{equation}
\label{eq_prob_to_1}
   \lim_{T,N\to\infty} \mathrm{Prob}\left( \max_{1\le i \le N} |\tilde{\lambda}_i-x_i|< \frac{1}{N^{1-\epsilon}}\right)=1.
\end{equation}
\end{theorem}

The proof of Theorem \ref{Theorem_vark_approximation}  relies on two steps. First, we modify our matrix $\tilde{\mathcal C}$ a bit, which leads to a surprising appearance of the Jacobi ensemble, as shown in Section \ref{Section_Jacobi_model}. Second, in Section \ref{Section_perturbation} we show that the distance between the original model and the modified one becomes small as $N\to\infty$.

Combining Theorem \ref{Theorem_vark_approximation} with known asymptotic results for the Jacobi ensemble, which we recall in Proposition \ref{Theorem_Jacobi_as} in Section \ref{Section_Jacobi}, we derive the asymptotics of \eqref{LR_NT} in the following theorem.

\begin{theorem} \label{Theorem_J_stat}
Fix $C>0$, and suppose that $T,N\to\infty$ in such a way that $\frac{T}{N}\in[k+1+C^{-1},C]$. For the data generating process \eqref{var_k_restr} with restrictions $\widehat H_0$ given by \eqref{eq_strong_vark}, for each finite $r=1,2,\dots$, we have convergence in distribution for the largest eigenvalues defined in Eq.~\eqref{eq_vark_eig}:
	\begin{equation}
	\label{eq_statistic_limit}
	 \frac{\sum_{i=1}^{r} \ln(1-\tilde{\lambda}_i)- r \cdot c_1(N,T)}{ N^{-2/3}  c_2(N,T)}  \, \xrightarrow[T,N\to\infty]{d} \sum_{i=1}^r \aa_i,
	\end{equation}
	where
	\begin{equation}
     \label{eq_constants_c1_c2}
	c_1\left(N,T\right)=\ln\left(1-\lambda_+\right), \qquad
	c_2\left(N,T\right)=-\frac{2^{2/3} \lambda_+^{2/3}}{(1-\lambda_+)^{1/3} (\lambda_+-\lambda_-)^{1/3}} \left(\p+\q\right)^{-2/3}  <0,
	\end{equation}
	\begin{equation}\label{pq_def}
	\p=2, \qquad \q=\frac{T}{N}-k,\qquad \lambda_\pm=\frac{1}{(\p+\q)^2}\left[\sqrt{\p(\p+\q-1)}\pm \sqrt{\q}  \right]^2.
	\end{equation}
\end{theorem}

\begin{remark}
The condition $\frac{T}{N}\in[k+1+C^{-1},C]$ is another way to require that $T$ and $N$ grow to infinity proportionally. For example, it is guaranteed by the joint limit \eqref{eq_limit_regime}. The role of $C$ is only to make sure that $T/N$ does not get too close to $k+1$ (if $T/N$ approaches $k+1$, then $\lambda_+$ approaches $1$ and $c_1$ explodes) or $+\infty$ (if $T/N$ becomes large, then $\lambda_+-\lambda_-$ and $\lambda_+$ tend to $0$ at the same speed and $c_2$ vanishes).
\end{remark}

Theorem \ref{Theorem_J_stat} gives us the basis of cointegration testing in the large $N,T$ setting. Treating $\widehat H_0$ as a proxy for $H_0$, we can use our asymptotic results to test high-dimensional VARs for the presence of cointegration. Formally, to perform testing, one first needs to calculate the statistic $LR_{N,T}(r)$ following Procedure \ref{sscc_BG}. We recommend using small\footnote{In Theorem \ref{Theorem_J_stat} $r$ is kept fixed as $N$ and $T$ grow. The role of this choice and the motivations for sticking to it are discussed in detail in \cite[Section 3.2]{BG}.} values of $r$, such as $r=1,2$, or $3$. Then, one needs to calculate $\frac{LR_{N,T}(r)- r \cdot c_1(N,T)}{ N^{-2/3}  c_2(N,T)}$, as in Theorem \ref{Theorem_J_stat}, and compare the result with quantiles of the sum of Airy$_1$, $\sum_{i=1}^r \aa_i$. If the rescaled statistic is larger than the $\alpha$ quantile, we reject the null of no cointegration at the $(1-\alpha)$ level. We report the quantiles for $r=1,2,3$ in Table \ref{airy_quantiles}. See also \citet{vignette_largevars} for more detailed tables for $r=1,\ldots,10$.

\begin{table}[t]
	\begin{tabular}{c|c|c|c|c}
		\hline
		\diagbox[width=1.5cm, height=0.65cm]{$r$}{$\alpha$} & $0.9$ & $0.95$ & $0.975$ & $0.99$\\
		\hline
		\hline
		1 & 0.44  & 0.97  & 1.45  & 2.01 \\
		\hline
        2 & -1.88 &	-1.09 & -0.40& 	0.41\\
        \hline
        3 & -5.91 &	-4.91 & -4.03 & -2.99\\
        \hline
        \hline
        \multicolumn{5}{c}{}
	\end{tabular}
	\caption{Quantiles of $\sum\limits_{i=1}^r \aa_i$\, for $r=1,2,3$ (based on $10^6$ Monte Carlo simulations of $10^8\times 10^8$ tridiagonal matrices of \citet{dumitriu_edelman}).}	
  \label{airy_quantiles}
\end{table}

Note that, although the asymptotic result \eqref{eq_statistic_limit} is shown under the restrictions $\widehat{H}_0$, we believe that it extends well beyond $\widehat H_0$: the same asymptotic results and testing procedures continue to hold in many situations with nonzero $\Gamma_i$ in Eq.~\eqref{var_k_restr}. While we do not have a full rigorous proof, we expect the following to be true:

\smallskip
\emph{For the data generating process \eqref{var_k_restr}, assume that the ranks of all $\Gamma_i$ are bounded, as are the norms of all the matrices and vectors involved in the specification of the process (see Section \ref{Section_Appendix_2} for more details). Then conclusion \eqref{eq_statistic_limit} of Theorem \ref{Theorem_J_stat} should continue to hold.
}

\smallskip

We collect extensive evidence supporting this statement. In Section \ref{Section_MC} we report results from Monte Carlo simulations consistent with it. Further, in Section \ref{appendix_H0} we present a precise mathematical conjecture in this direction and give a heuristic argument for its validity. The intuition is that generic small-rank matrices are negligible relative to the scale of the rest of the process and, thus, their addition does not change the asymptotics. For this intuition to hold, it is important to correctly specify the parameter $k$ in the procedure to be equal to (or greater than) its true value. Otherwise (i.e., if we do not regress on the relevant $\Delta X_{t-i}$ in the procedure), the presence of $\Gamma_i$ can have an effect similar to that of the presence of nonzero $\Pi$: it leads to the appearance of special highly correlated linear combinations of rows of $\tilde R_0$ and $\tilde R_k$, which changes the behavior of the largest canonical correlations $\tilde \lambda_i$; see also the simulations in Section \ref{section_MC_var_order}.

\bigskip

Theorem \ref{Theorem_J_stat} means that under $\widehat H_0$ the largest eigenvalues $\tilde \lambda_i$ are close to $\lambda_+$, which is the right point of the support of the Wachter distribution in Eq.~\eqref{eq_Jacobi_equilibrium_1}. The relevance of this theorem for cointegration testing stems from the fact that we expect some of the eigenvalues to be much larger than $\lambda_+$ when cointegrating relationships are present. As an illustration, see Figure \ref{small_rk_pic}, where $\Pi$ of rank $1$ leads to the largest eigenvalue being to the right of $\lambda_+$ and separated from the other eigenvalues. The separation is due to the small rank of $\Pi$. However, even if the rank of $\Pi$ is large, we expect the largest eigenvalue to be significantly larger than $\lambda_+$, and, thus, the test remains relevant (see Section \ref{section_power}). Providing rigorous results on the consistency of the test is an important task for future research. We present the first result in this direction in Corollary \ref{Corollary_power} in Section \ref{appendix_power}, where we produce a lower bound on the power of the test against a particular ``one cointegrating relationship'' alternative and show that the power tends to $1$ as $T/N$ tends to infinity.

\section{Monte Carlo simulations}
\label{Section_MC}
\subsection{Size}
We refer to \citet[Section 5]{BG} for the finite sample size performance of our test for $k=1$. The results for VAR($k$) are similar, and we do not show them in much detail here. For illustration purposes and to represent the comparative statics, Table \ref{rej_rate_var_k} reports the empirical size for $T=522,\,N=92$ (those numbers correspond to our empirical example in Section \ref{Section_SP}) for tests based on VAR($k$), $k=1,2,3,4$ procedures.\footnote{Depending on the assumed order of autoregression, we have different numbers of regressors in Procedure \ref{sscc_BG}.} We can see that the numbers are close to the desired $5\%$ and, for the same $N$ and $T$, a lower order of VAR leads to slightly better results.
\begin{table}[t]
	\begin{tabular}{c|c|c|c}
		\hline
		VAR($1$) & VAR($2$) & VAR($3$) & VAR($4$)\\
		\hline
		\hline
		$5.81\%$  & $5.92\%$  & $6.12\%$  & $6.95\%$ \\ 
		\hline
        \hline
        \multicolumn{4}{c}{}
	\end{tabular}
	\caption{Empirical size under no cointegration hypothesis ($5\%$ nominal level) based on VAR($k$) tests, $k=1,2,3,4$. Data generating process: ${\Delta X_{it}=\eps_{it}}$, $\eps_{it}\thicksim$ i.i.d.~$\mathcal{N}(0,1),\, T=522,\,N=92,\, MC = 1,000,000$ replications.}	
  \label{rej_rate_var_k}
\end{table}


\begin{figure}[t!]
\begin{subfigure}{.45\textwidth}
  \centering
  \includegraphics[width=1.0\linewidth]{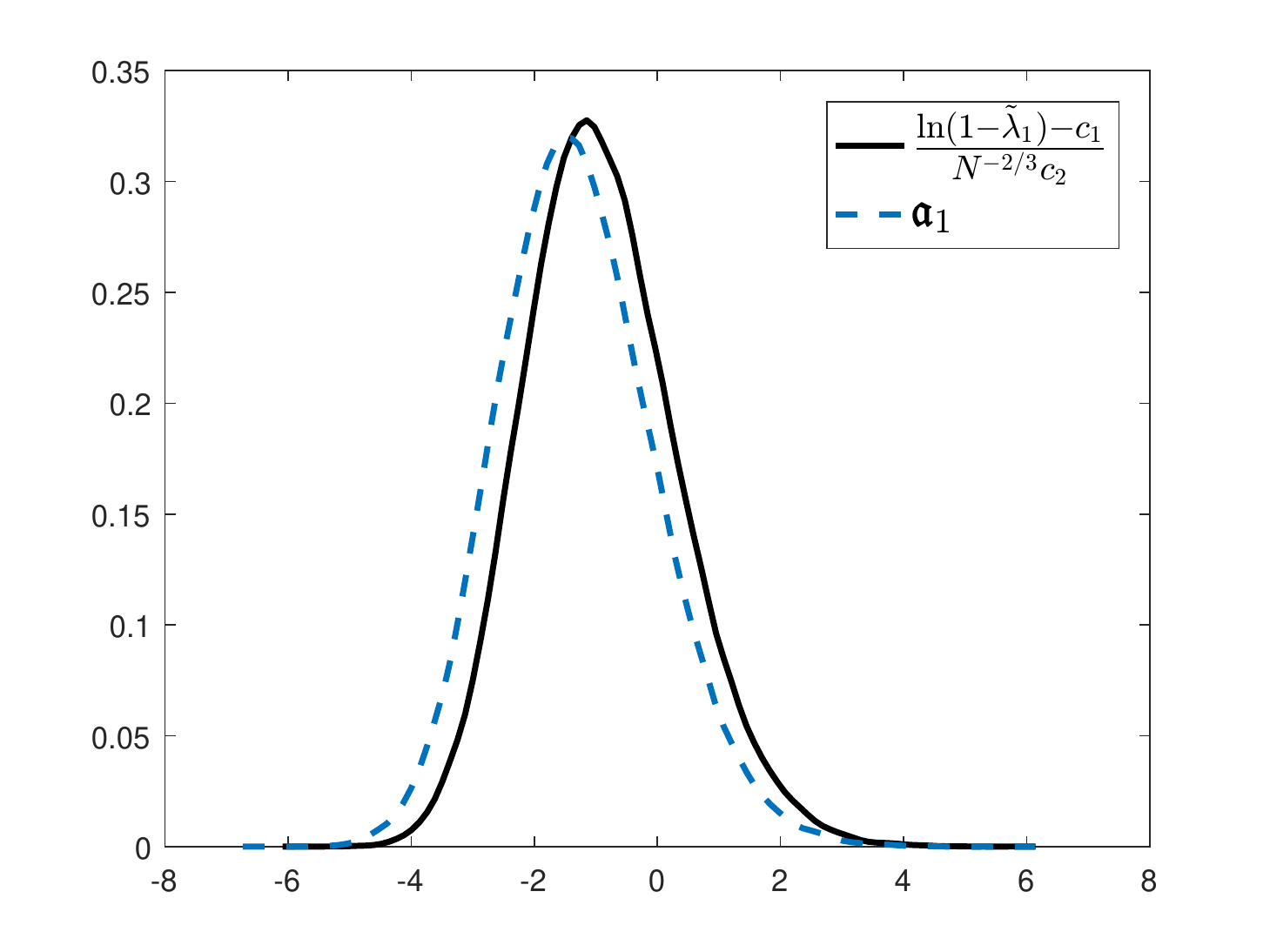}
  \caption{$\Gamma_1=0.95E_{11}$.}
  \label{G11}
\end{subfigure}%
\begin{subfigure}{.45\textwidth}
  \centering
  \includegraphics[width=1.0\linewidth]{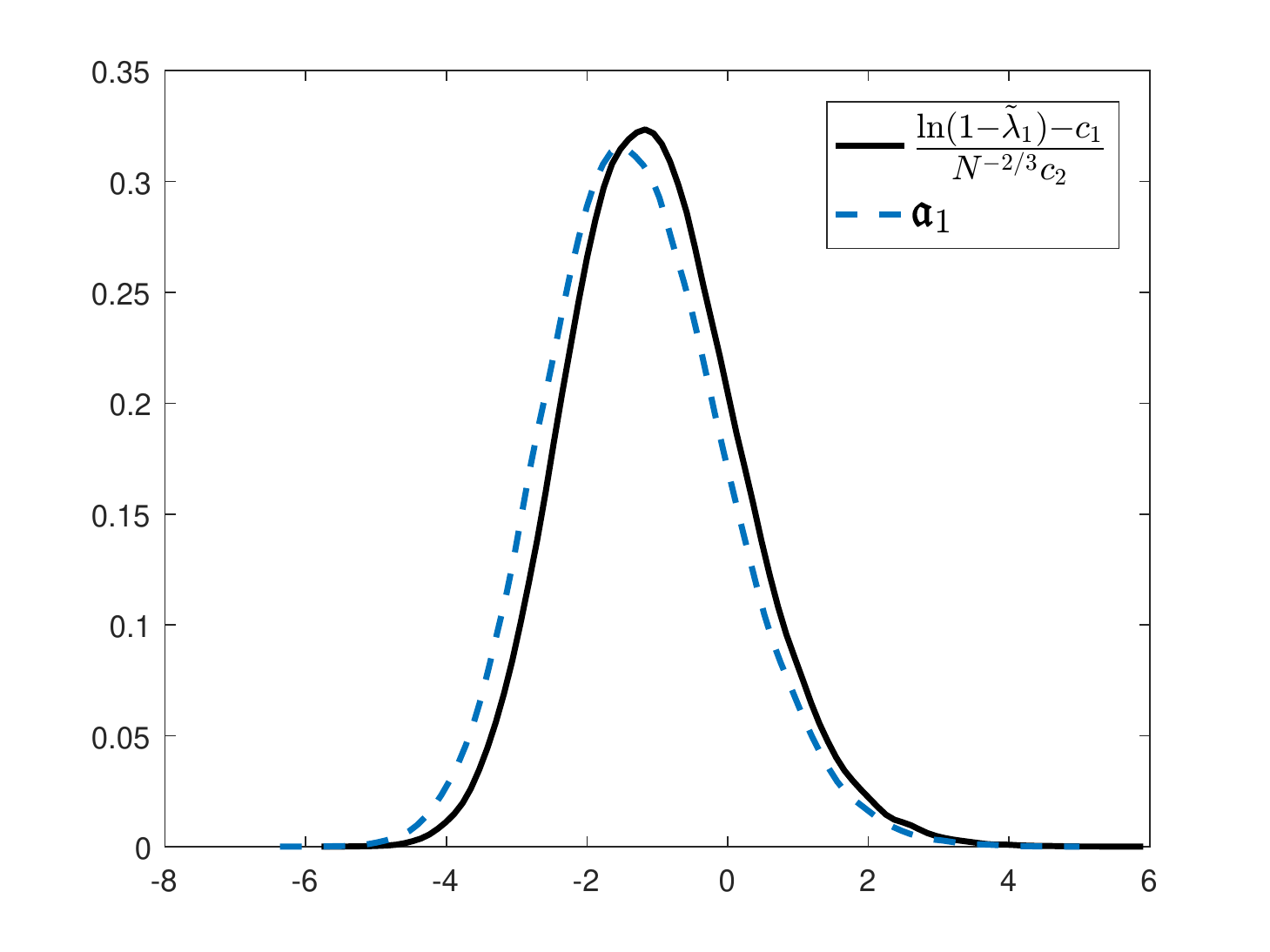}
  \caption{$\Gamma_1=0.1E_{\cdot1}+0.95E_{12}$.}
  \label{G12}
\end{subfigure}
\caption{Airy$_1$ and asymptotic distribution of the rescaled $\ln(1-\tilde{\lambda}_1)$ under $H_0$.
Data generating process: $\Delta X_{t}=\Gamma_1\Delta X_{t-1}+\eps_{t}$, $\eps_{it}\thicksim$ i.i.d.~$\mathcal{N}(0,1)$, $T=500$, $N=100$, $MC=100,000$ replications.}
\label{nonzeroGamma}
\end{figure}

\begin{figure}[t!]
    \centering
	{\scalebox{0.7}{\includegraphics{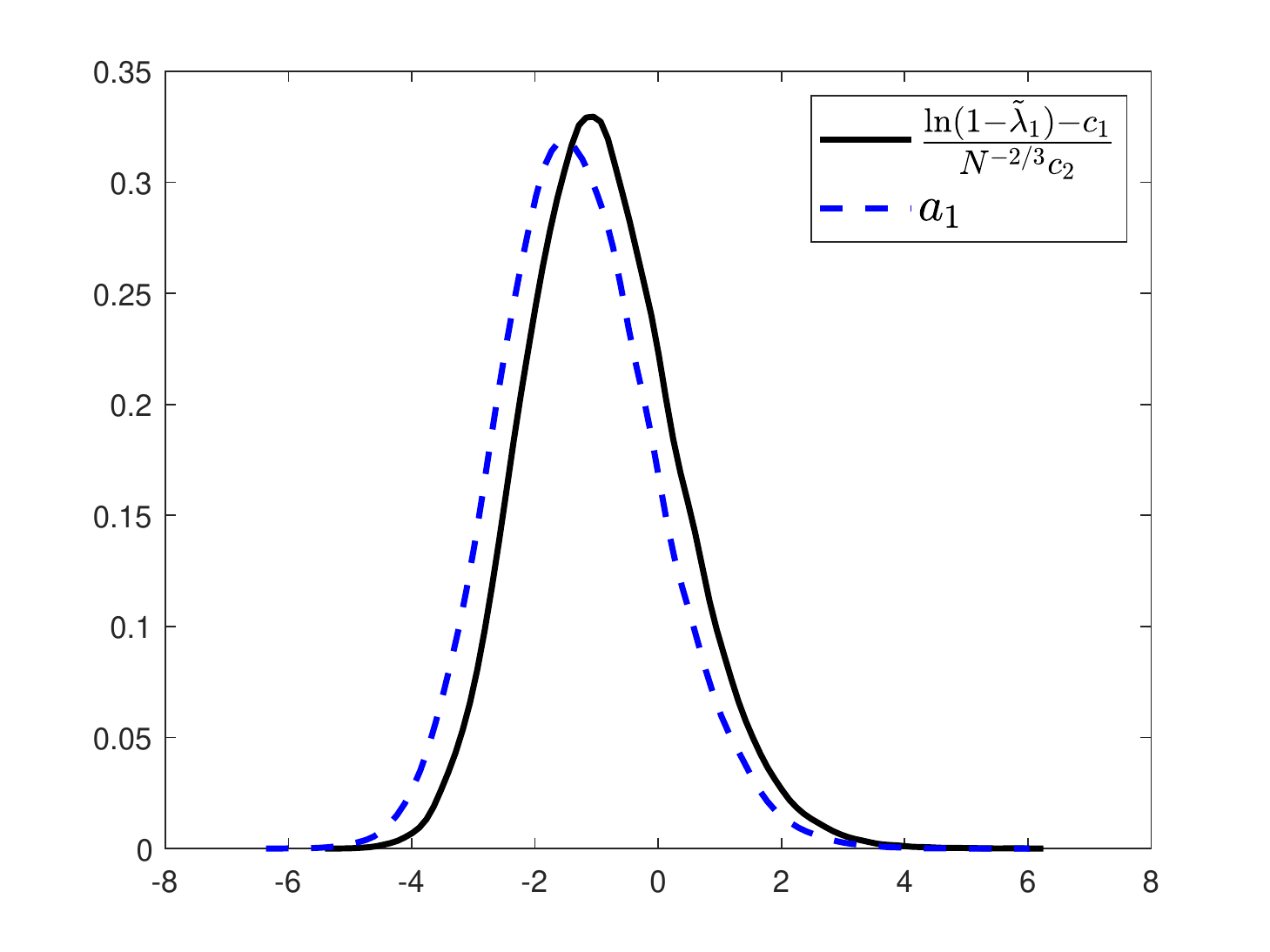}}}
\caption{Airy$_1$ and asymptotic distribution of the rescaled $\ln(1-\tilde{\lambda}_1)$ under $H_0$.
Data generating process: $\Delta X_{t}=\Gamma_1\Delta X_{t-1}+\Gamma_2\Delta X_{t-2}+\eps_{t}$, $\Gamma_1=E_{11},\,{\Gamma_2=-\tfrac2{9}E_{11}}$, $\eps_{it}\thicksim$ i.i.d.~$\mathcal{N}(0,1)$, $T=500$, $N=100$, $MC=100,000$ replications.}
\label{nonzeroGamma_VAR3}
\end{figure}

\subsection{$\mathbf{H_0}$ vs.~$\mathbf{\widehat H_0}$} \label{section_H_tilde} An important aspect of our analysis for $k>1$ is the introduction of the additional restrictions $\widehat H_0$ maintained under the null. We would like to check whether Theorem \ref{Theorem_J_stat} can hold under the less restrictive $H_0$ instead of $\widehat H_0$. Some theoretical results in this direction are provided in Section \ref{appendix_H0}. Here we complement them with Monte Carlo simulations. For $N=100$, $T=500$ we simulate the data based on i.i.d.~$\mathcal{N}(0,1)$ errors $\eps_{it}$, zero $\Pi$, and nonzero $\Gamma_i$ (i.e., this corresponds to $H_0$ but not $\widehat H_0$). We then compare the density of the test based on the largest eigenvalue $\tilde{\lambda}_1$ ($r=1$ case of Theorem \ref{Theorem_J_stat} with statistic $\frac{\ln(1-\tilde{\lambda}_1)-c_1(N,T)}{N^{-2/3}  c_2(N,T)}$), with the density of the first coordinate of the Airy$_1$ point process, $\aa_1$. If the densities coincide, then it means that we can still use the asymptotics from Theorem \ref{Theorem_J_stat} to test the null of no cointegration.

Let $E_{ij}$ be a matrix with $1$ at the cell $(i,j)$ and $0$s everywhere else and let $E_{\cdot j}$ be a matrix with $1$s filling the entire column $j$ and $0$s everywhere else. In the first two experiments we take $k=2$. We set $\Gamma_1=0.95E_{11}$ in the first one, which guarantees stationarity of $\Delta X_t$ but allows for strong time correlations in the first coordinate via the $0.95$ factor. In this case the rank of $\Gamma_1$ is $1$. In the second experiment we consider an asymmetric matrix $\Gamma_1=0.1E_{\cdot1}+0.95E_{12}$, which has a close to $1$ singular value because of the $0.95$ factor; the rank of $\Gamma_1$ is $2$ in this case. The results are illustrated in Figure \ref{nonzeroGamma}. In the third experiment, we take $k=3$, $\Gamma_1=E_{11}$, and $\Gamma_2=-\tfrac2{9}E_{11}$, so that both matrices are of rank $1$. The value $-\tfrac{2}{9}$ guarantees stationarity of $\Delta X_t$, since $1-z+\tfrac{2}{9}z^2=\left(1-\tfrac1{3}z\right)\left(1-\tfrac{2}{3}z\right)$. The result is shown in Figure \ref{nonzeroGamma_VAR3}.
We interpret the outcomes of these three experiments as a strong argument toward the validity of an analogue of Theorem \ref{Theorem_J_stat} well beyond the $\widehat H_0$ setting.\footnote{The minor mismatches between densities as in Figures \ref{nonzeroGamma} and \ref{nonzeroGamma_VAR3} should be expected even under  $\widehat H_0$. Theorem \ref{Theorem_vark_approximation} (after multiplication of the result by $N^{2/3}$, as in Eq.\ \eqref{eq_statistic_limit}) predicts errors of at least ${\rm const}\cdot N^{-1/3}$ in the approximations under  $\widehat H_0$.}

\begin{figure}[t]
	\centering
	{\scalebox{0.7}{\includegraphics{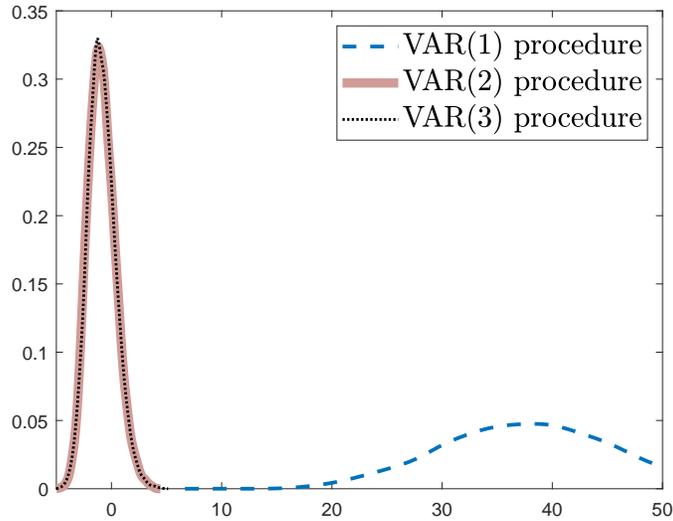}}}
	\caption{Density of rescaled $\ln(1-\tilde{\lambda}_1)$ obtained from various procedures. The correct procedure corresponds to VAR(2), $k=2$.
Data generating process: \mbox{$\Delta X_{t}=0.95E_{11}\Delta X_{t-1}+\eps_{t}$,} $\eps_{it}\thicksim$ i.i.d.~$\mathcal{N}(0,1)$, $T=500$, $N=100$, $MC=10,000$ replications. }
	\label{rk_VAR_matters}
\end{figure}

\subsection{Order of VAR}\label{section_MC_var_order}

\begin{figure}[t]
\begin{subfigure}{.33\textwidth}
  \centering
  \includegraphics[width=1.0\linewidth]{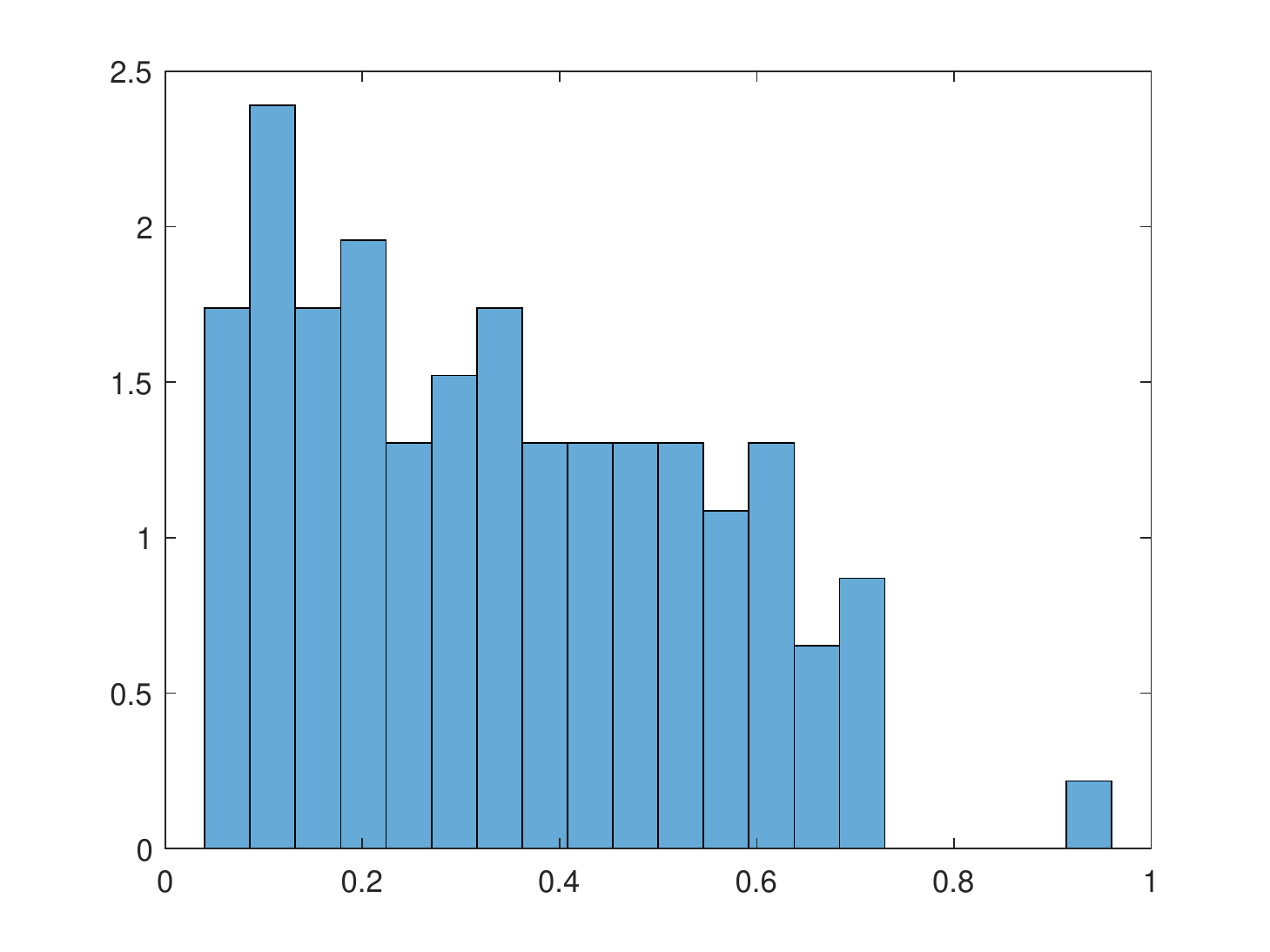}
  \caption{VAR(1) procedure.}
  \label{var2_var1test}
\end{subfigure}%
\begin{subfigure}{.33\textwidth}
  \centering
  \includegraphics[width=1.0\linewidth]{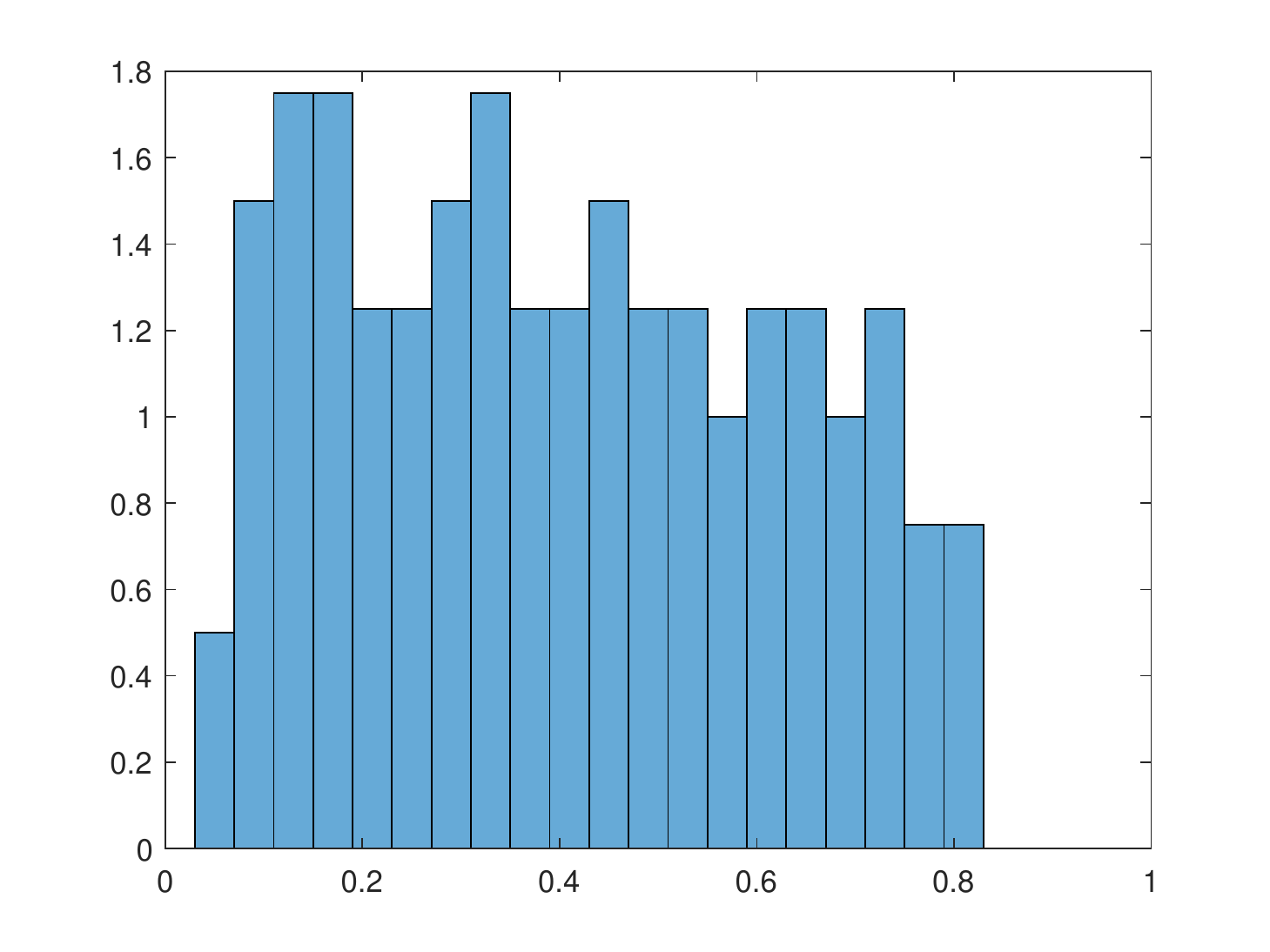}
  \caption{VAR(2) procedure.}
  \label{var2_var2test}
\end{subfigure}
\begin{subfigure}{.33\textwidth}
  \centering
  \includegraphics[width=1.0\linewidth]{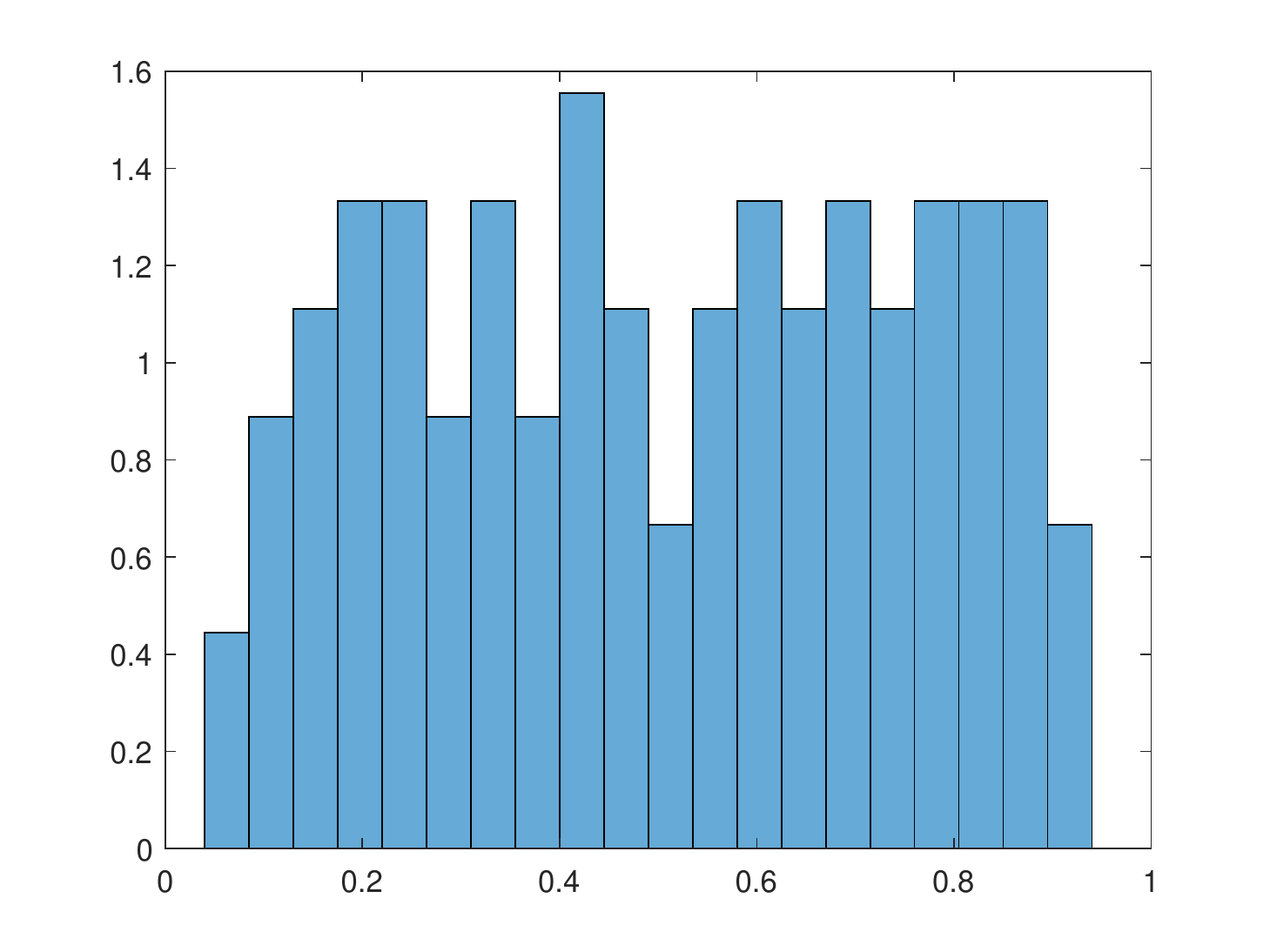}
  \caption{VAR(3) procedure.}
  \label{var2_var3test}
\end{subfigure}
\caption{Eigenvalues obtained from various procedures. The correct procedure corresponds to VAR(2), $k=2$.
Data generating process: \mbox{$\Delta X_{t}=0.95E_{11}\Delta X_{t-1}+\eps_{t}$,} $\eps_{it}\thicksim$ i.i.d.~$\mathcal{N}(0,1)$, $T=500$, $N=100$.}
\label{var2_var_k_test}
\end{figure}

\begin{figure}[t]
\begin{subfigure}{.33\textwidth}
  \centering
  \includegraphics[width=1.0\linewidth]{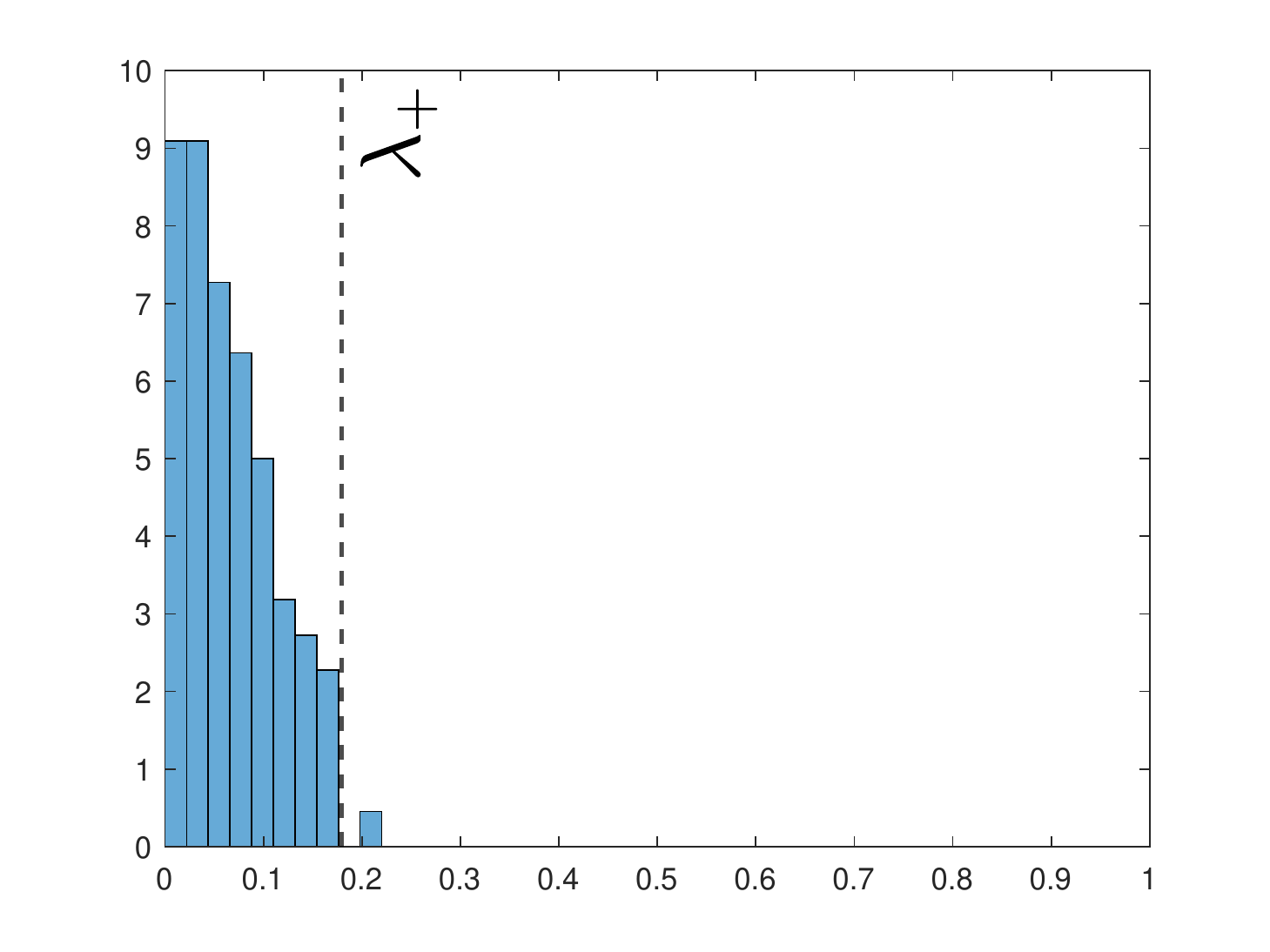}
  \caption{VAR(1) procedure.}
  \label{var5_var1test}
\end{subfigure}%
\begin{subfigure}{.33\textwidth}
  \centering
  \includegraphics[width=1.0\linewidth]{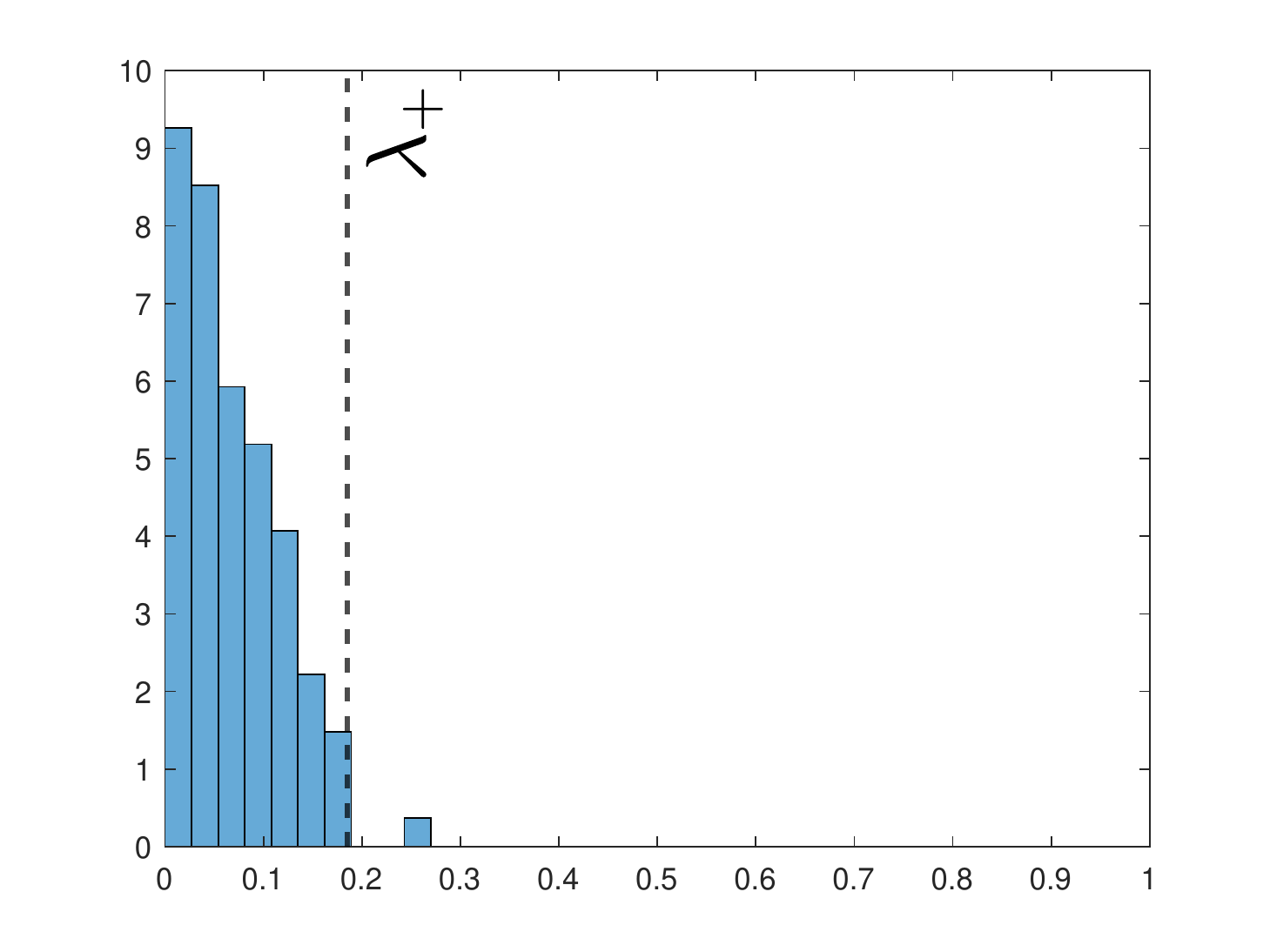}
  \caption{VAR(2) procedure.}
  \label{var5_var2test}
\end{subfigure}
\begin{subfigure}{.33\textwidth}
  \centering
  \includegraphics[width=1.0\linewidth]{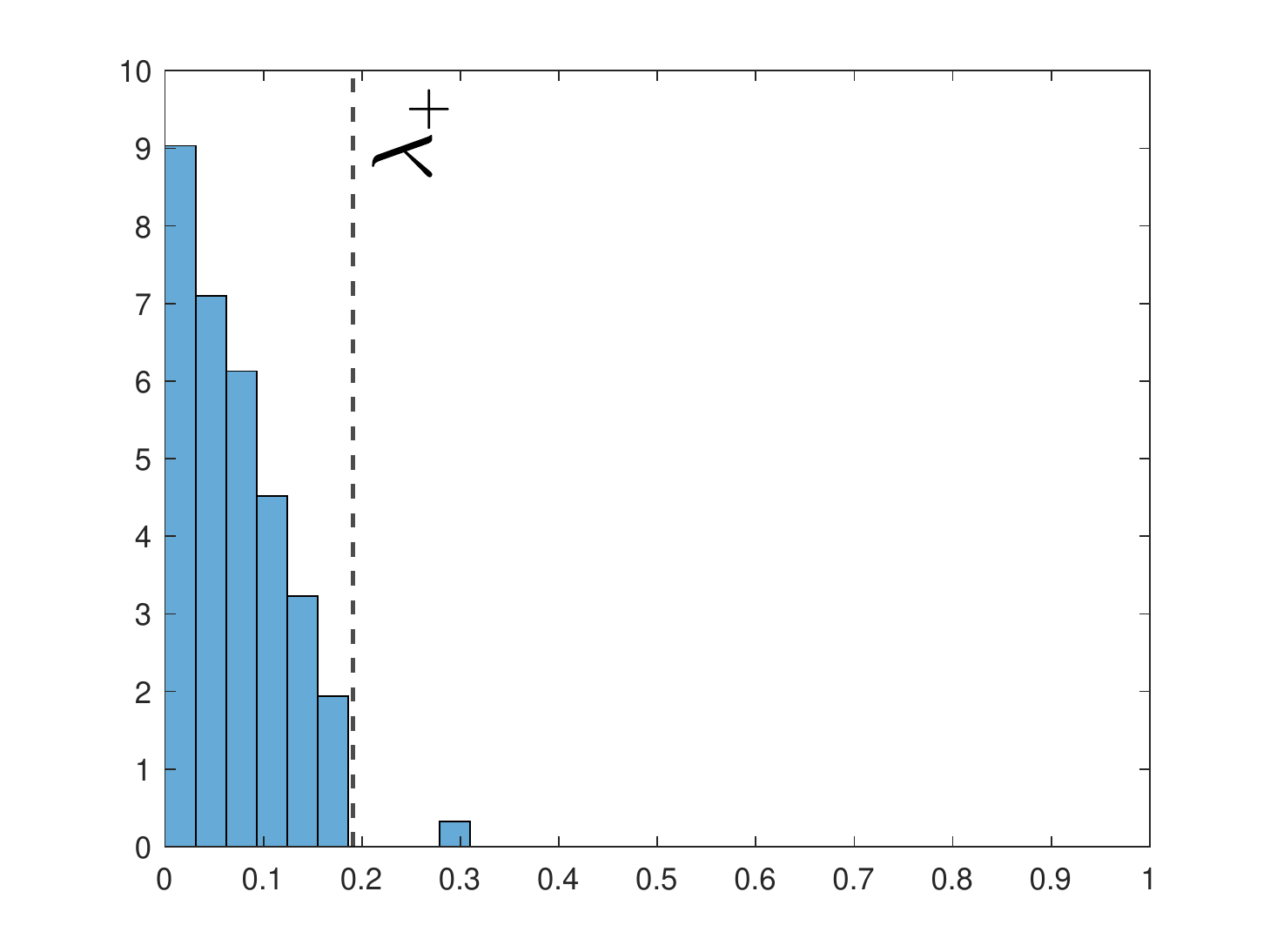}
  \caption{VAR(3) procedure.}
  \label{var5_var3test}
\end{subfigure}
\begin{subfigure}{.33\textwidth}
  \centering
  \includegraphics[width=1.0\linewidth]{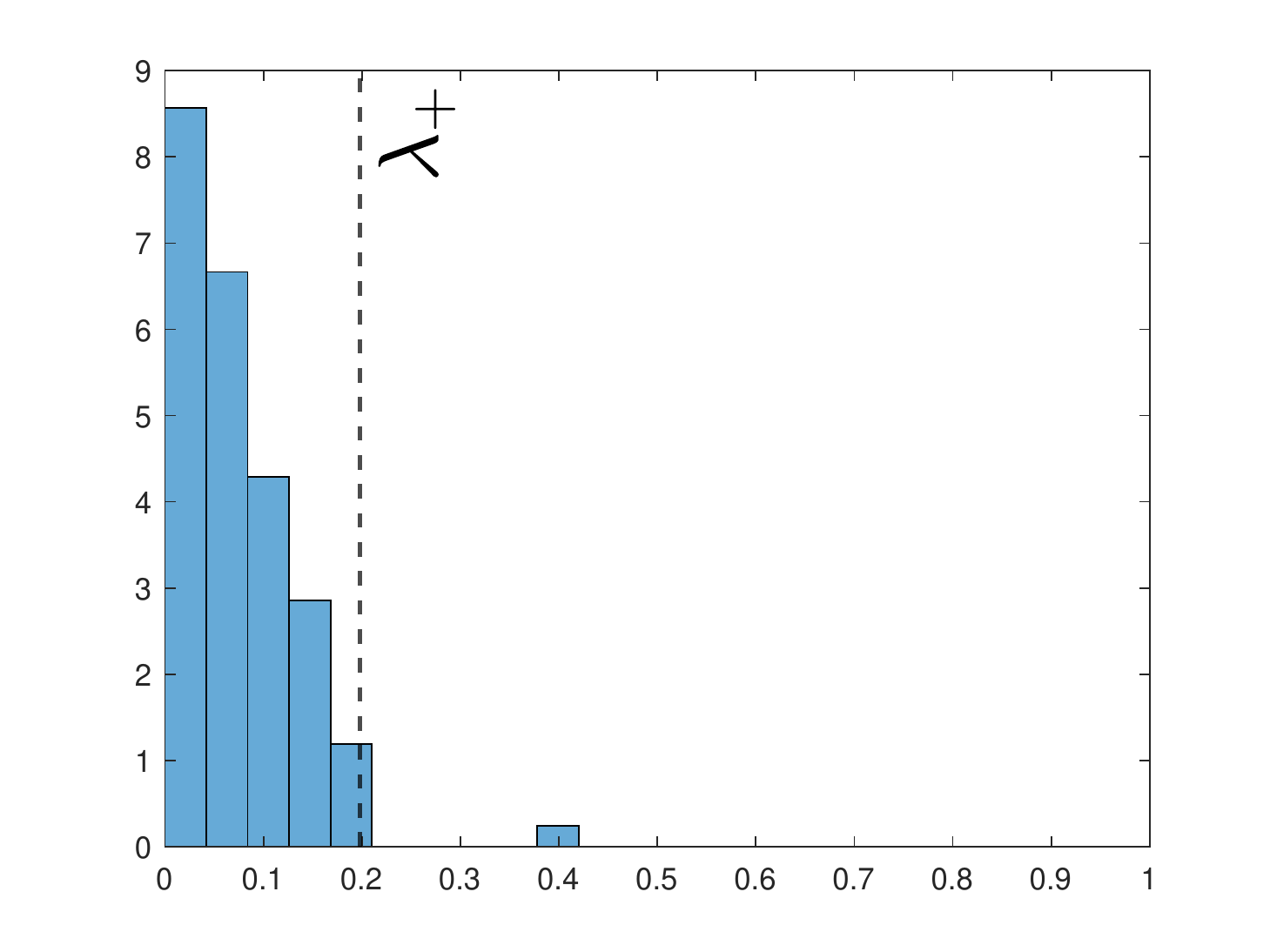}
  \caption{VAR(4) procedure.}
  \label{var5_var4test}
\end{subfigure}%
\begin{subfigure}{.33\textwidth}
  \centering
  \includegraphics[width=1.0\linewidth]{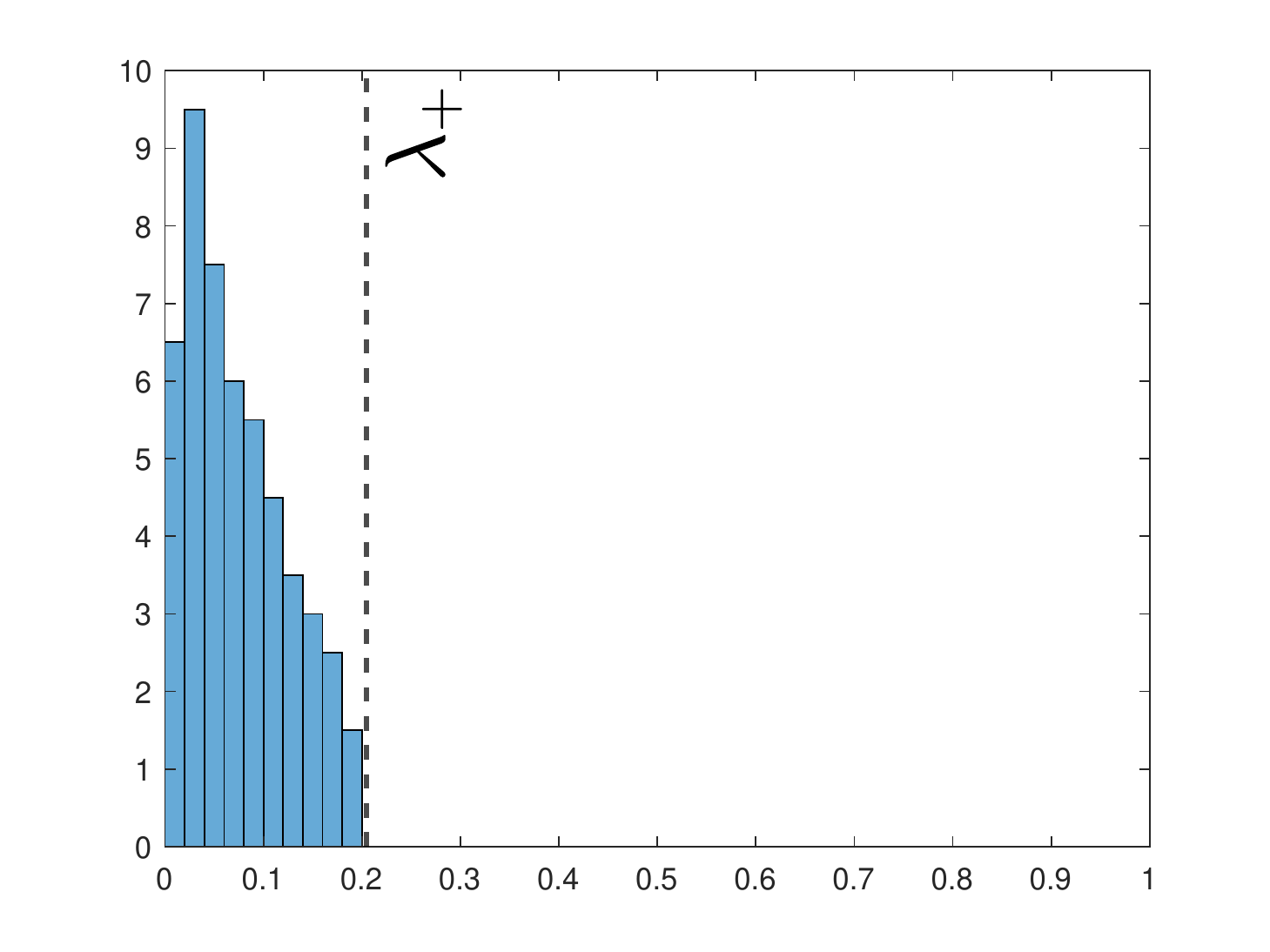}
  \caption{VAR(5) procedure.}
  \label{var5_var5test}
\end{subfigure}
\begin{subfigure}{.33\textwidth}
  \centering
  \includegraphics[width=1.0\linewidth]{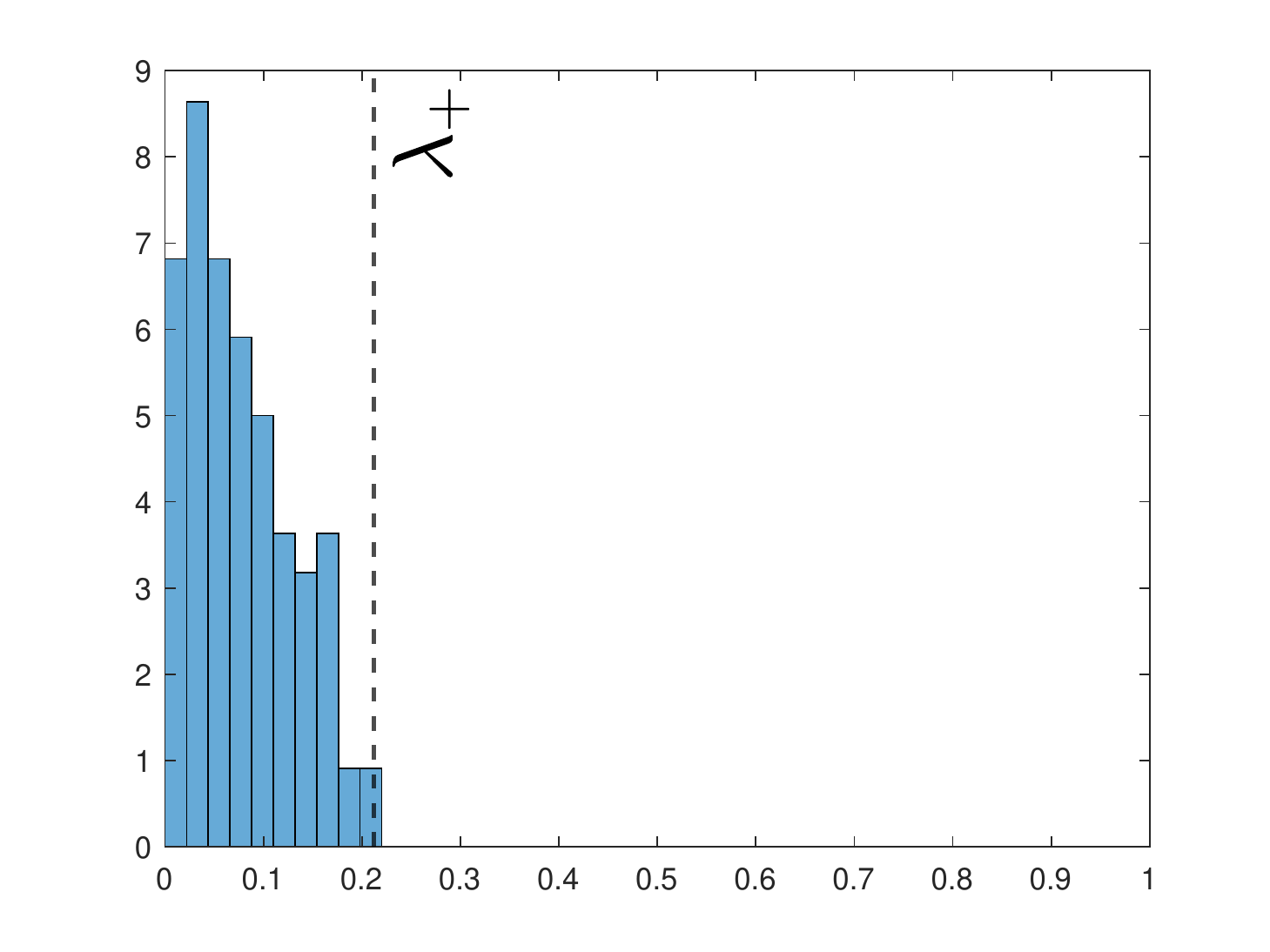}
  \caption{VAR(6) procedure.}
  \label{var5_var6test}
\end{subfigure}
\caption{Eigenvalues obtained from various procedures. The correct procedure corresponds to VAR(5), $k=5$.
Data generating process: \mbox{$\Delta X_{t}=0.9E_{11}\Delta X_{t-4}+\eps_{t}$,} $\eps_{it}\thicksim$ i.i.d.~$\mathcal{N}(0,1)$, $T=3000$, $N=100$.}
\label{var5_var_k_test}
\end{figure}

It is essential for the experiments in the last paragraph that the data generating process is VAR(2) and that the procedure we use also corresponds to VAR(2), i.e.,~$k=2$ in the notations of Sections \ref{Section_setting} and \ref{Section_modified_test}. As illustrated in Figure \ref{rk_VAR_matters}, using a larger $k$ would lead to similar results, while incorrectly using a VAR(1) procedure when the data generating process is VAR(2) would imply wrong centering and scaling. Moreover, one can spot in Figure \ref{var2_var_k_test} that underestimation of the order of the VAR can be misinterpreted as a presence of cointegration\footnote{Related simulations for $\Gamma_1=\theta E_{11}$ are also reported in \cite[Section 7.3]{BG}: For small $\theta$, such as $\theta=0.5$, the VAR(1) procedure still performs well. However, as $\theta$ grows to $1$, the performance quickly deteriorates ($\theta=0.95$ in Figure \ref{rk_VAR_matters}).} (largest eigenvalue separated from the rest leading to the large value of the test statistic). However, as we increase the order, the largest eigenvalue becomes inseparable from the rest, and no sign of false cointegration remains present. Thus, practitioners are encouraged to experiment with the order of the VAR to make sure that they are detecting cointegration and not simply using the  wrong model.

Note that one should be careful if using classical information criteria for estimating the order $k$ of a VAR in our situation. They are known to be unreliable in high-dimensional settings and may underestimate $k$ (see, e.g., the simulations in \citet{gonzalo_pitarakis2002}). A possible approach to choosing $k$ is to look sequentially at histograms of eigenvalues at $k=1,2,\dots$. If the outlier eigenvalues larger than $\lambda_+$ exist for the procedures with all $k$ and perhaps move closer to $\lambda_+$ as $k$ grows (corresponding to a decrease in the power of the test), then this is a strong indication of the presence of cointegration. On the other hand, if there is a sharp transition---i.e., outlier eigenvalues are present when we use the VAR($k'$) procedure for $k'<k$ and abruptly disappear at $k=k'$---then this is an indication that the true model is VAR($k$) without cointegration (see Figure \ref{var5_var_k_test}).

All of the above reinforces the importance of using the VAR($k$) rather than the VAR($1$) procedure.


\subsection{Small ranks}\label{section_mc_small_rank}
To illustrate the importance of small ranks (e.g., \eqref{eq_rank_restriction} in Theorem \ref{Theorem_empirical}), we also redo the same procedure for a matrix $\Gamma_1$ of full rank and set $\Gamma_1$ to be $0.95\1_N$, where $\1_N$ is an $N\times N$ identity matrix. The result is shown in Figure \ref{095I}. While the shape and the scale (corresponding to $N^{2/3}$ rescaling in Theorem \ref{Theorem_J_stat}) of the distribution remain similar, the location changes. Thus, the small-rank restriction of Eq.\ \eqref{eq_rank_restriction} is important not only in the context of Theorem \ref{Theorem_empirical} but also for correct centering in possible generalizations of Theorem \ref{Theorem_J_stat}.

\begin{figure}[t]
	\centering
	{\scalebox{0.7}{\includegraphics{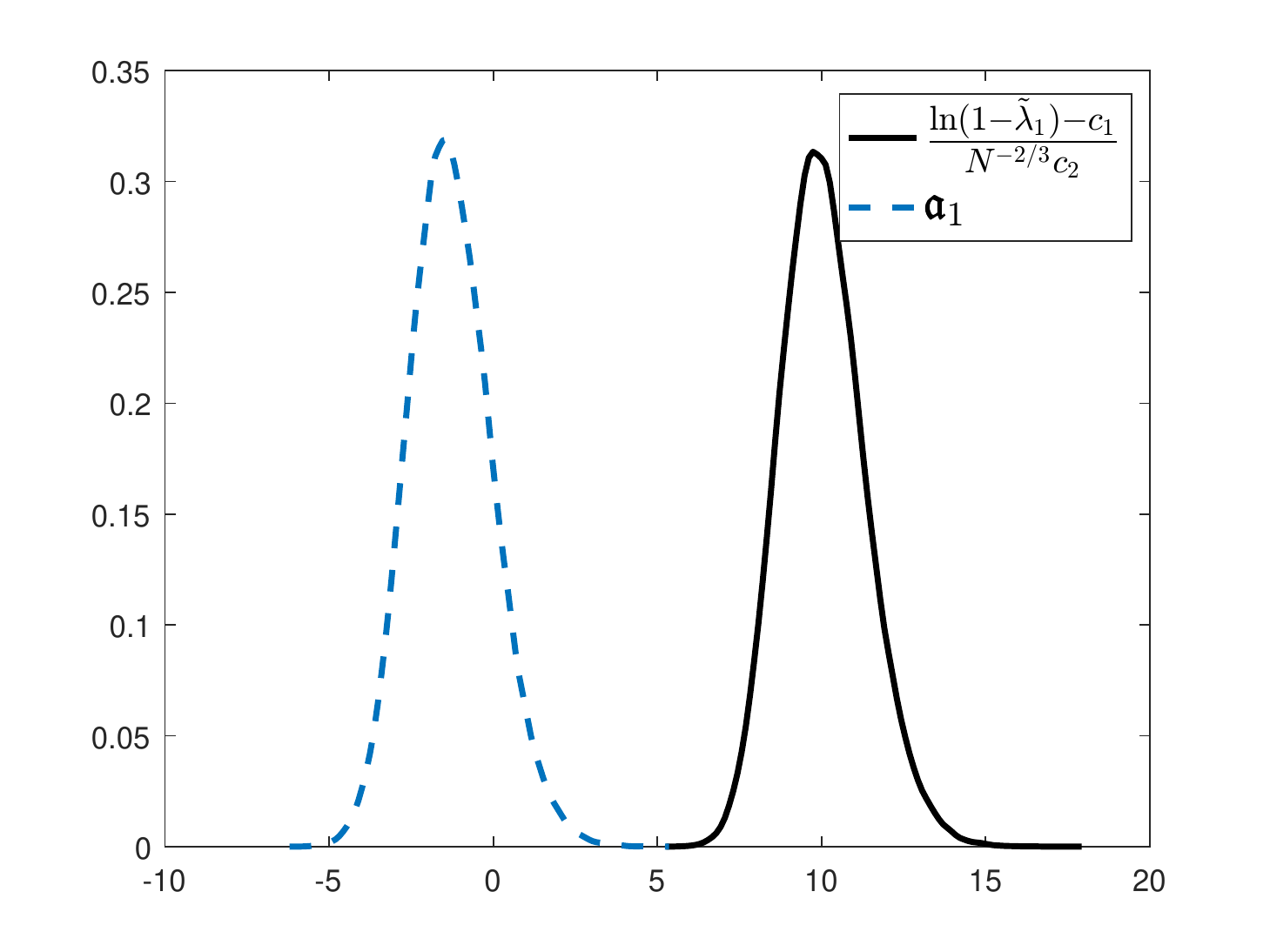}}}
	\caption{Airy$_1$ and asymptotic distribution of the rescaled $\ln(1-\tilde{\lambda}_1)$ under $H_0$ with $\Gamma_1$ of full rank.
Data generating process: $\Delta X_{it}=0.95\Delta X_{it-1}+\eps_{it}$, $\eps_{it}\thicksim$ i.i.d.~$\mathcal{N}(0,1)$, $T=500$, $N=100$, $MC=100,000$ replications. }
	\label{095I}
\end{figure}

\begin{figure}[h!]
	\centering
	{\scalebox{0.7}{\includegraphics{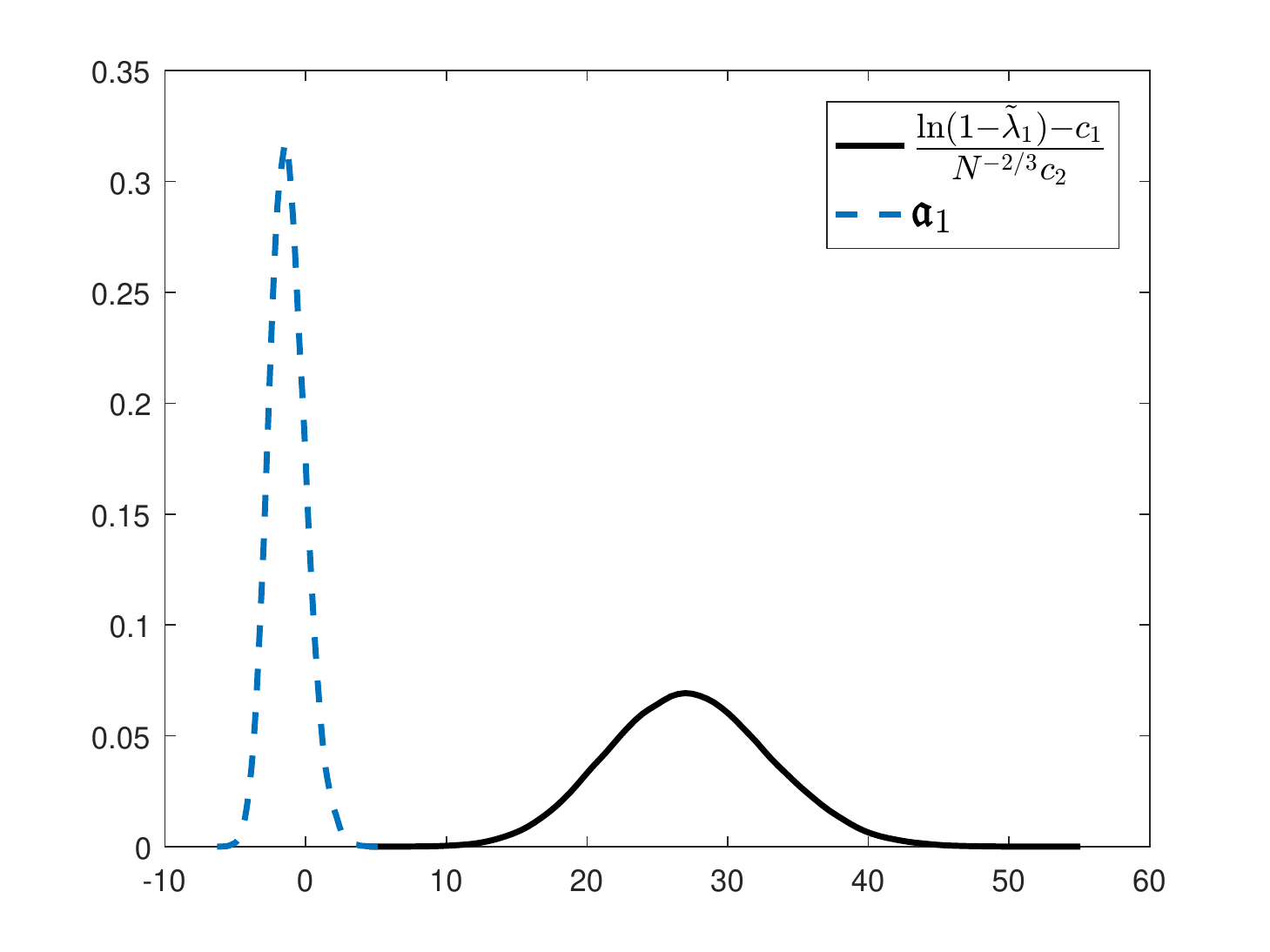}}}
	\caption{Airy$_1$ and asymptotic distribution of the rescaled $\ln(1-\tilde{\lambda}_1)$ under $H_1$.
Data generating process: $\Delta X_{t}=-0.95E_{11}X_{t-2}+\eps_{t}$, $\eps_{it}\thicksim$ i.i.d.~$\mathcal{N}(0,1)$, $T=500$, $N=100$, $MC=100,000$ replications. }
	\label{Pi-095E11}
\end{figure}

\subsection{Power} \label{section_power}
Finally, we simulate the process based on $\Pi\neq0$ to assess the power of our cointegration testing procedure. We refer the reader to \citet[Section 5.2]{BG} for many simulations in the $k=1$ case and do not repeat similar experiments here.

For the first experiment, we use $k=2$, $\Gamma_1=0$, and $\Pi=-0.95E_{11}$. Figure \ref{Pi-095E11} shows the results of the simulation. Two curves are separated; moreover, the black straight line (test distribution) is flatter than the blue dashed curve ($\aa_1$). First, the separation of the curves is in line with the usefulness of Theorem \ref{Theorem_J_stat} in cointegration hypothesis testing, since the test statistic was designed to distinguish between $\Pi=0$ and $\Pi\neq 0$. Second, the distinct variances are due to the fact that (as we expect from a comparison with results on spiked random matrices in the literature; see, e.g., \citet{BBP}) under the alternative ($\Pi\neq0$) the test needs to be scaled differently: Instead of $N^{2/3}$ rescaling one should use $N^{1/2}$. This result is in line with the power analysis of our test in Section \ref{appendix_power}.

\begin{figure}[t]
\begin{subfigure}{.33\textwidth}
  \centering
  \includegraphics[width=1.0\linewidth]{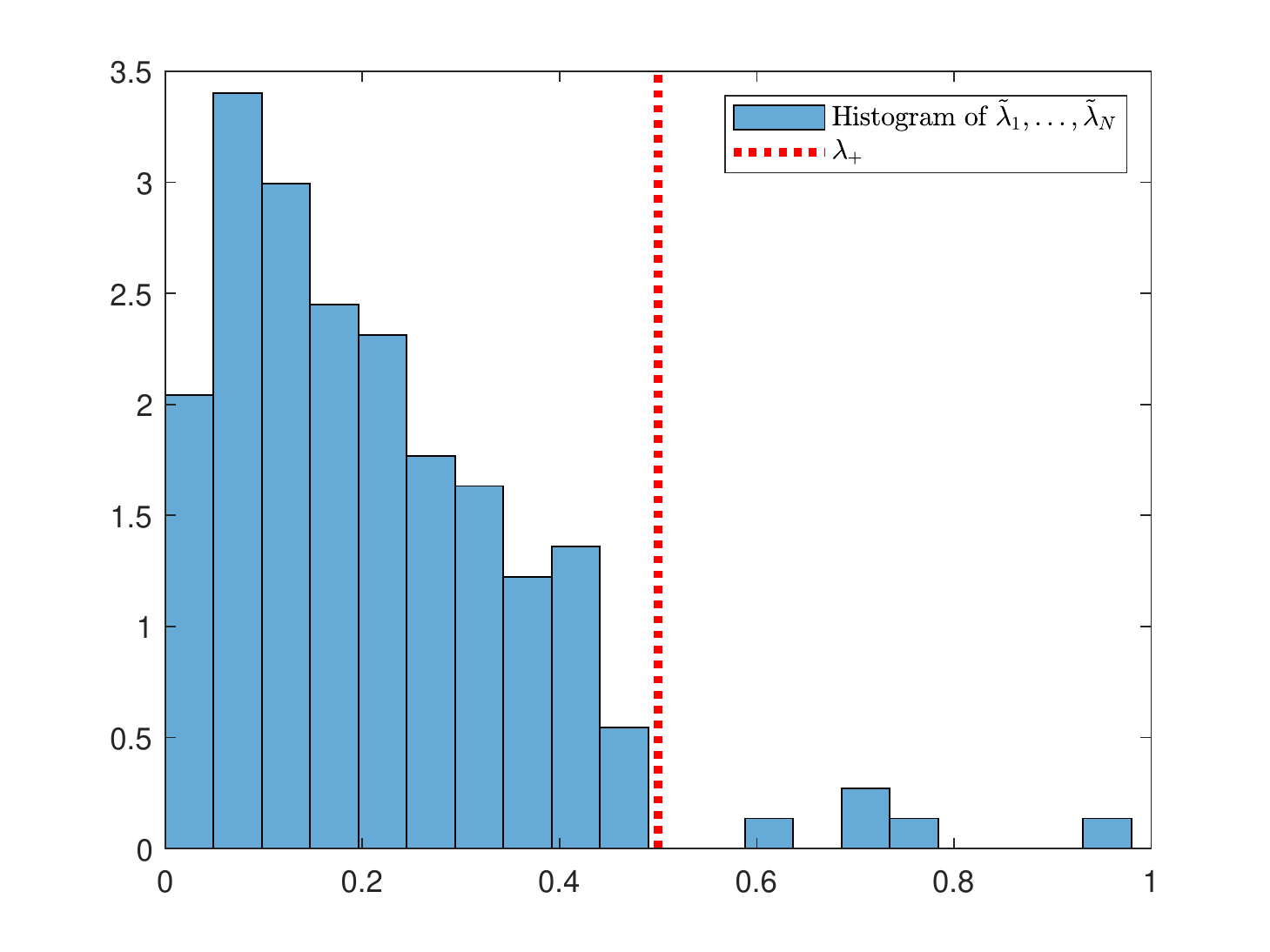}
  \caption{$\rho=\rank(\Pi)=5$.}
  \label{rk_Pi_5}
\end{subfigure}%
\begin{subfigure}{.33\textwidth}
  \centering
  \includegraphics[width=1.0\linewidth]{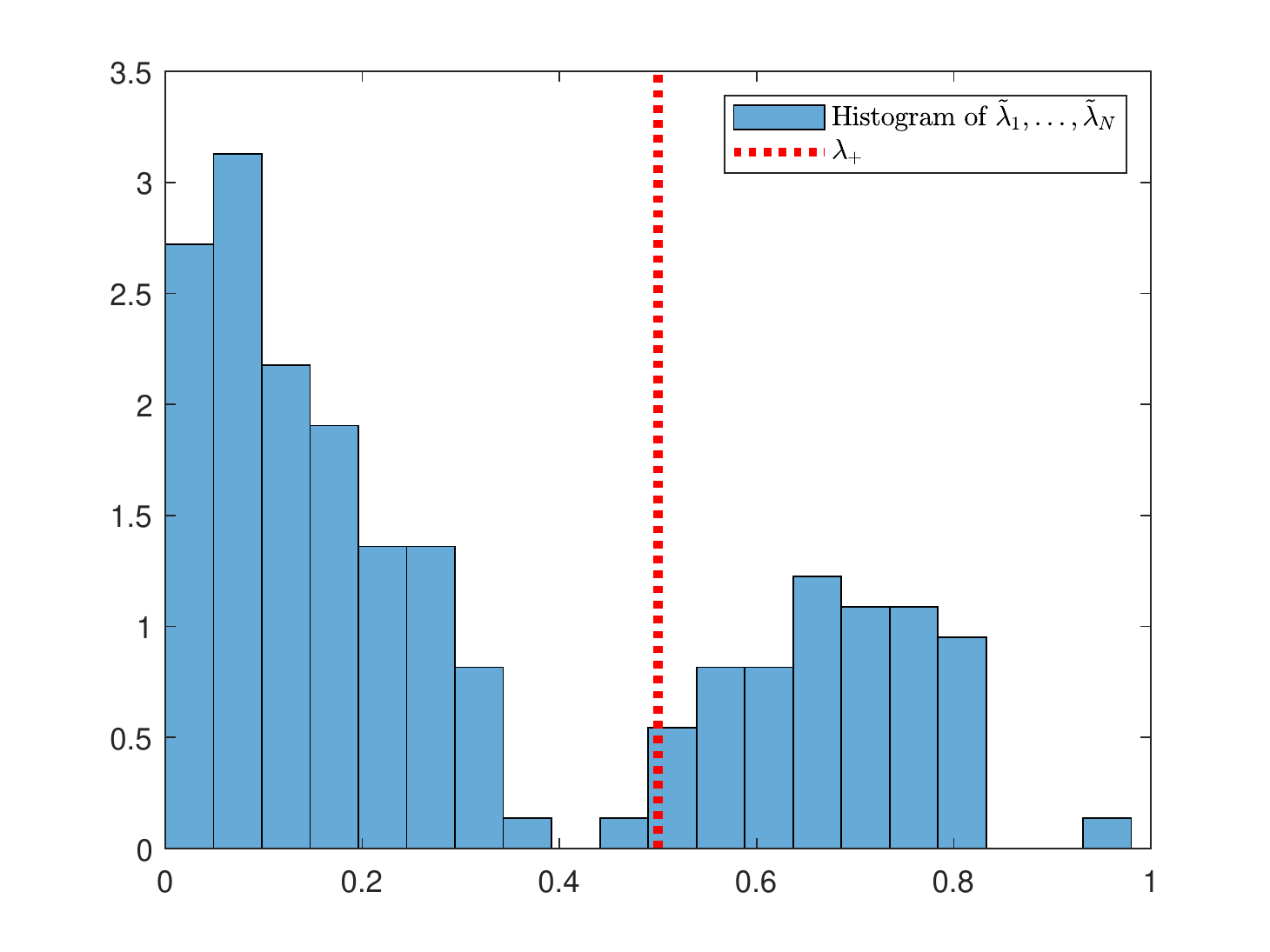}
  \caption{$\rho=\rank(\Pi)=50$.}
  \label{rk_Pi_50}
\end{subfigure}
\begin{subfigure}{.33\textwidth}
  \centering
  \includegraphics[width=1.0\linewidth]{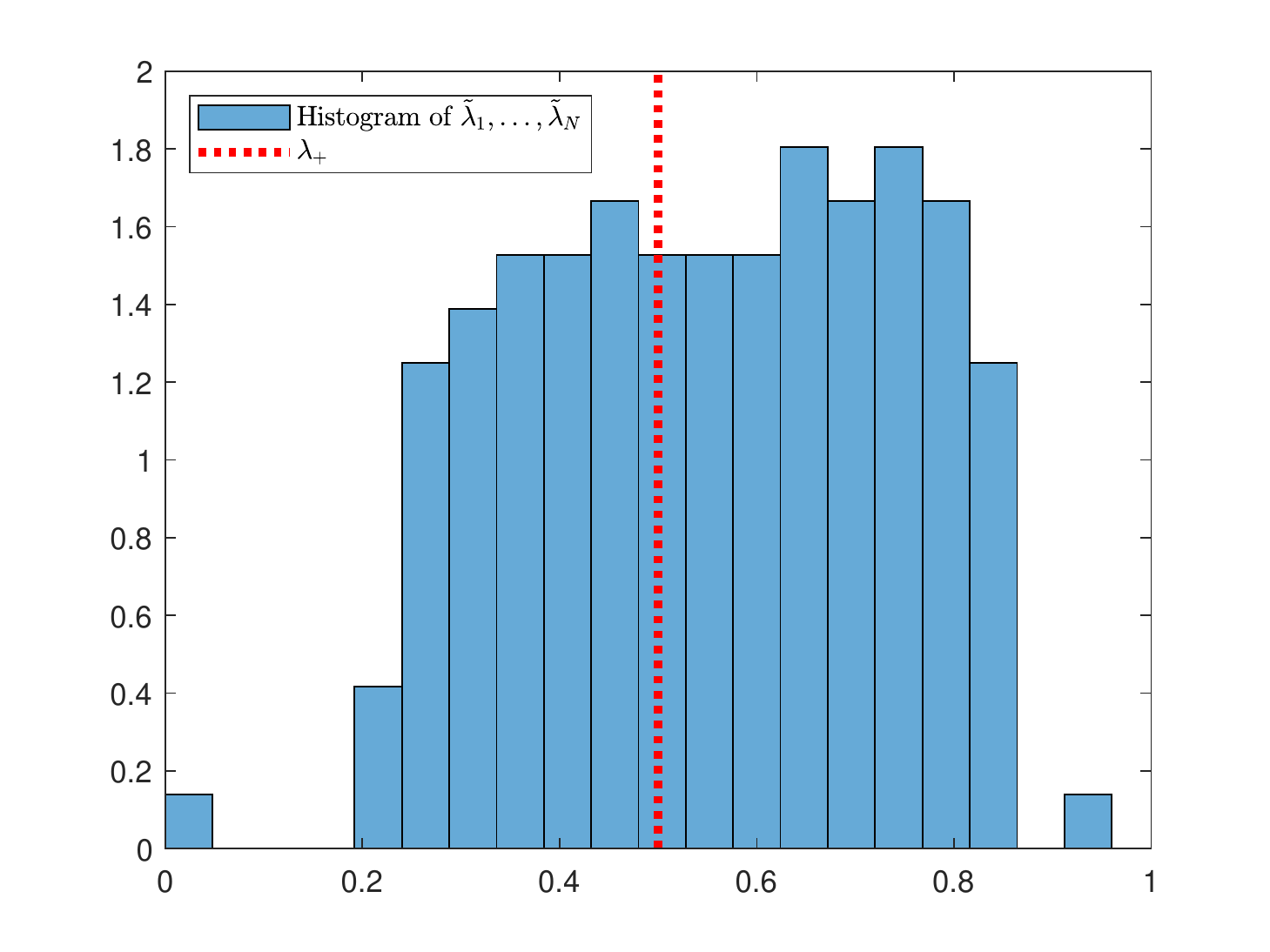}
  \caption{$\rho=\rank(\Pi)=150$.}
  \label{rk_Pi_150}
\end{subfigure}
\caption{Separation of eigenvalues from $\lambda_+$ under various ranks of $\Pi$.
Data generating process: $\Delta X_{t}=1_N+0.95E_{12}\Delta X_{t-1}-0.8I_\rho X_{t-2}+\eps_{t}$, $\eps_{it}\thicksim$ i.i.d.~$\mathcal{N}(0,1)$, $T=1500$, $N=150$, $I_\rho$ is a matrix with ones on the first $\rho$ diagonal elements and zeros elsewhere.}
\label{rk_Pi_separation}
\end{figure}

For the second experiment, we use $k=2$, $N=150$, $T=1500$, $\Gamma_1=0.95E_{12}$, and $\Pi=-0.8I_\rho X_{t-2}$, where $I_\rho$ is a matrix with ones on the first $\rho$ diagonal elements and zeros elsewhere. The histograms of eigenvalues $\tilde{\lambda}_1,\ldots,\tilde{\lambda}_N$ for $\rho=5,50, 150$ are shown in Figure \ref{rk_Pi_separation}. As $\rho$ becomes large, we are no longer in the framework of Theorem \ref{Theorem_empirical}, and the result about the convergence to the Wachter distribution does not apply. Nevertheless, we observe that the largest eigenvalues are significantly larger than $\lambda_+$ of Theorem \ref{Theorem_J_stat}. Recall that our testing procedure is based on comparing the largest eigenvalues\footnote{More precisely, logarithms of 1 minus eigenvalues vs.~$\log(1-\lambda_+)$.} with $\lambda_+$. Thus, this leads to the conclusion that our test is useful for all values of $\rho\in [1,N]$: the test  rejects $H_0$ of no cointegration (i.e., the $\Pi=0$ hypothesis) at a very high statistical significance level.

\section{Empirical illustrations}\label{Section_empirical}

\subsection{S$\&$P$\mathbf{100}$}
\label{Section_SP}

\begin{figure}[t]
\begin{subfigure}{.44\textwidth}
  \centering
  \includegraphics[width=1.0\linewidth]{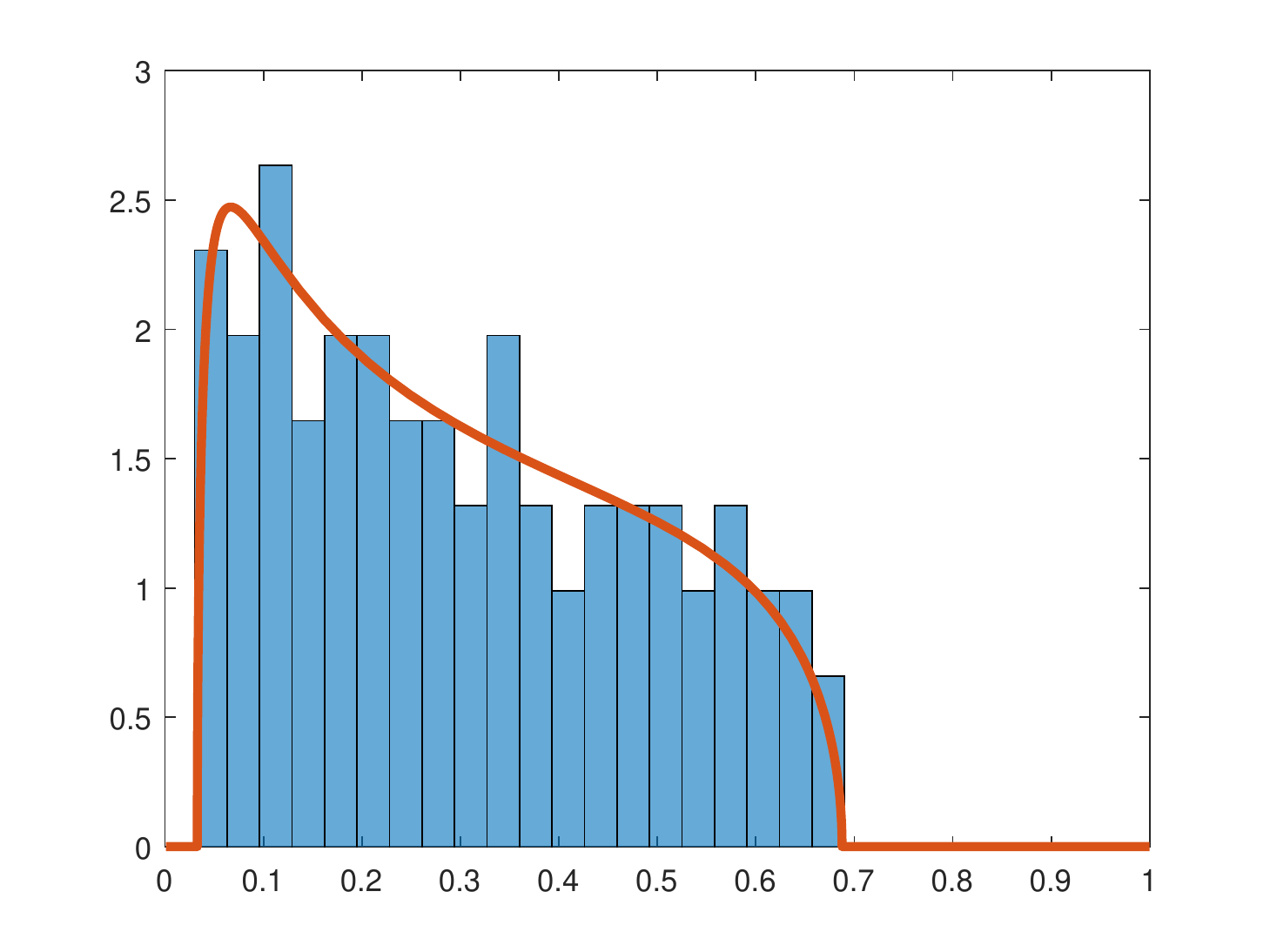}
  \caption{VAR($1$).}
  \label{SPvar1}
\end{subfigure}%
\begin{subfigure}{.44\textwidth}
  \centering
  \includegraphics[width=1.0\linewidth]{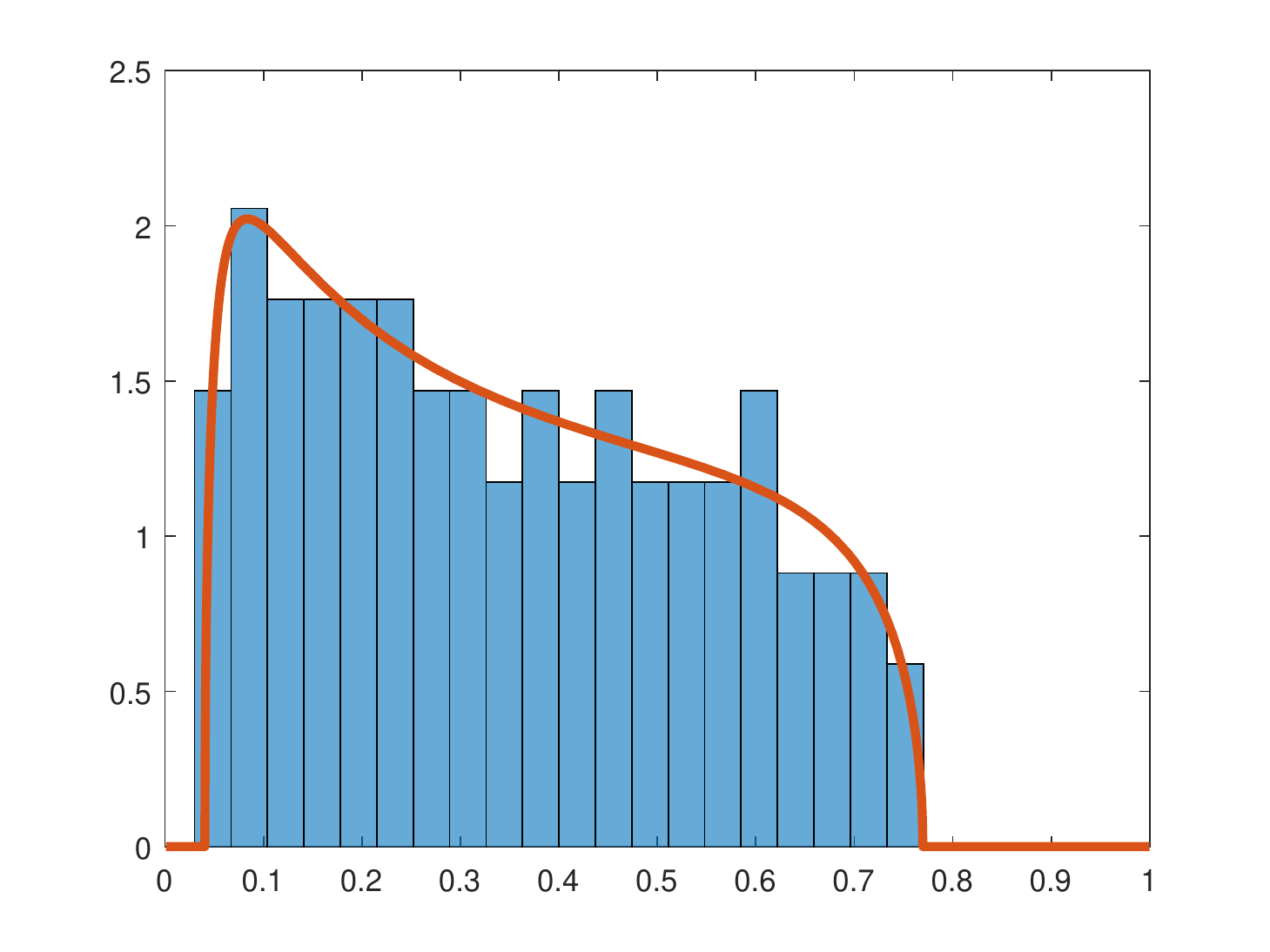}
  \caption{VAR($2$).}
  \label{SPvar2}
\end{subfigure}

\begin{subfigure}{.44\textwidth}
  \centering
  \includegraphics[width=1.0\linewidth]{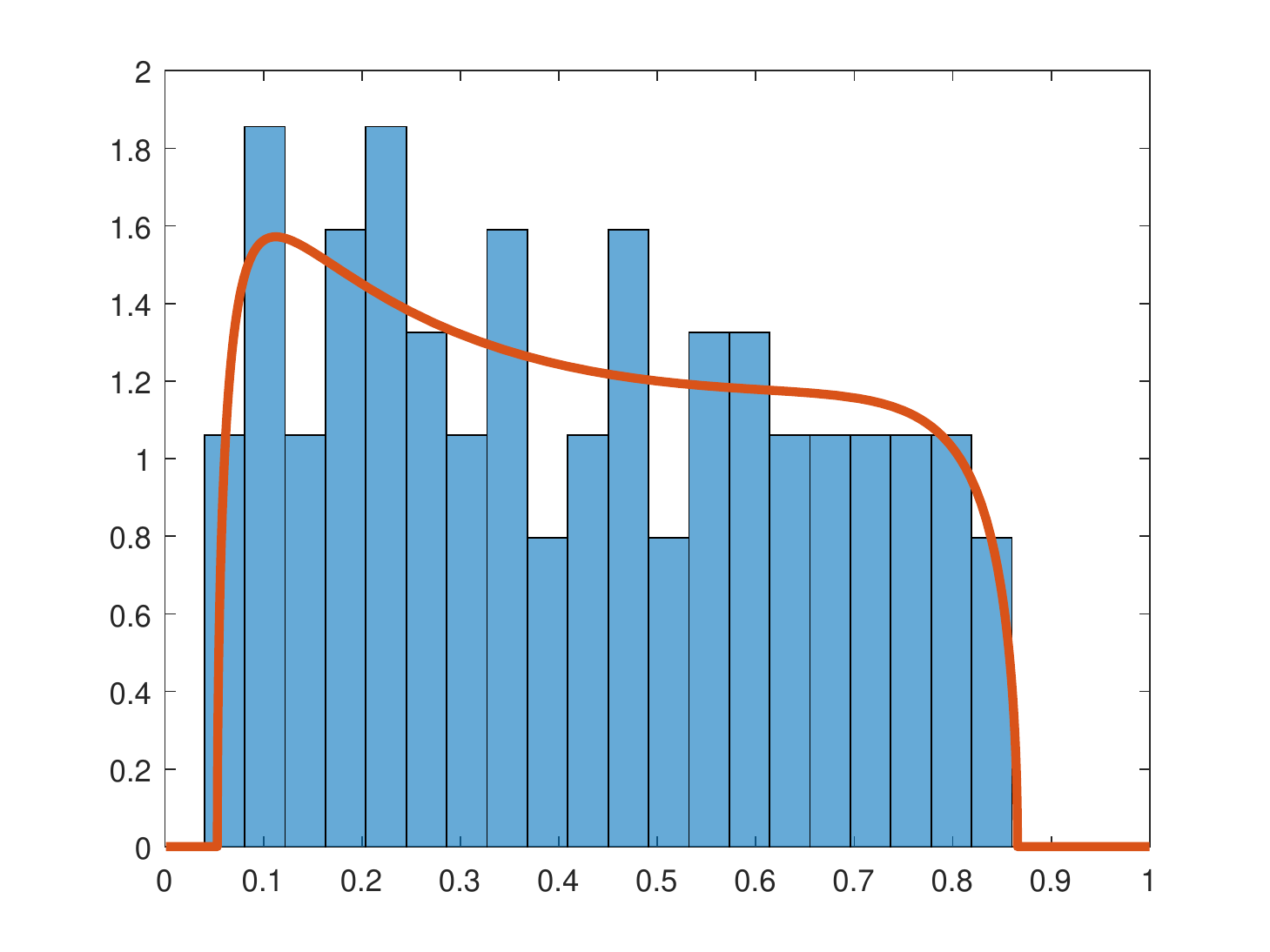}
  \caption{VAR($3$).}
  \label{SPvar3}
\end{subfigure}
\begin{subfigure}{.44\textwidth}
  \centering
  \includegraphics[width=1.0\linewidth]{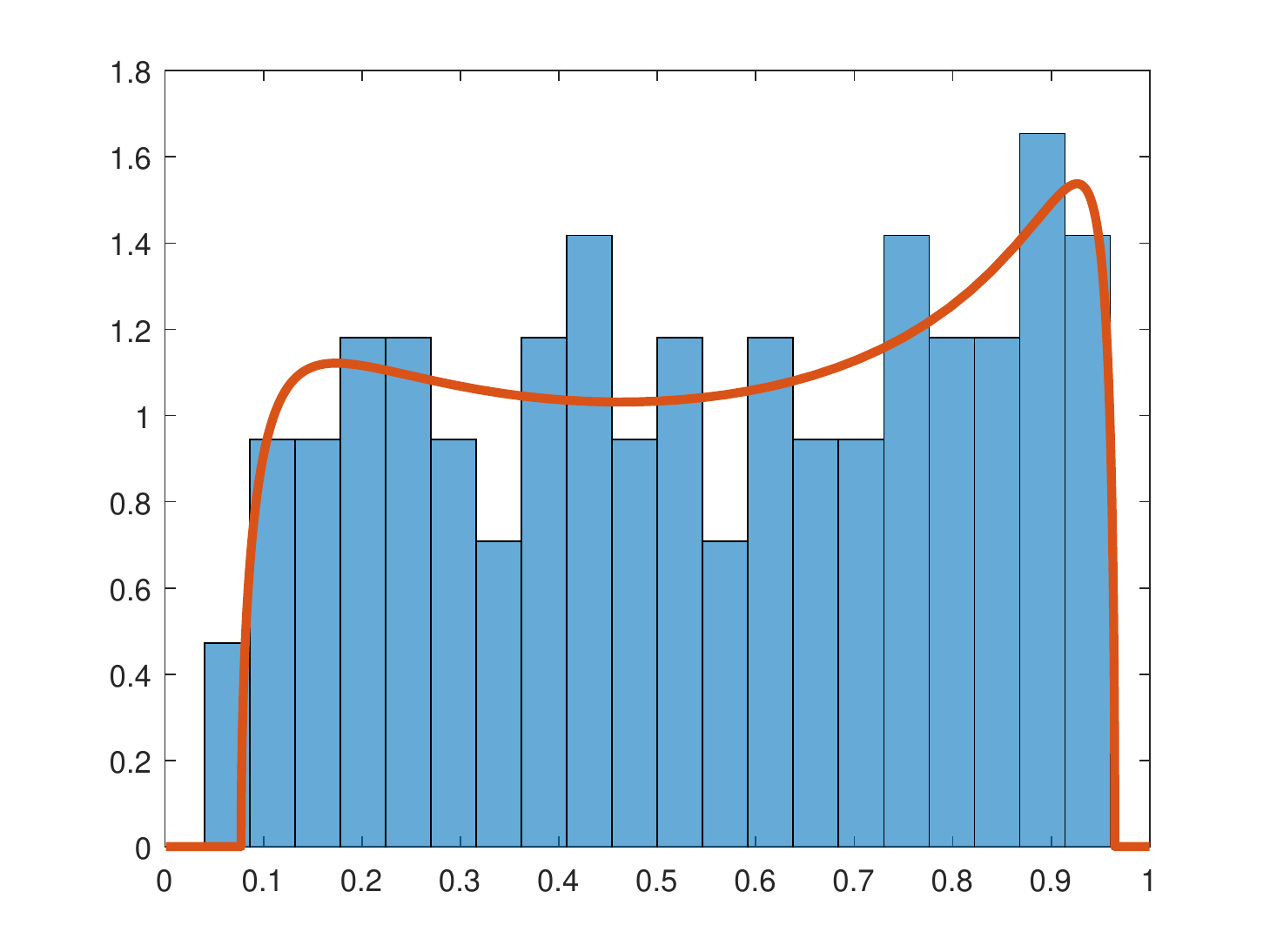}
  \caption{VAR($4$).}
  \label{SPvar4}
\end{subfigure}
\caption{Eigenvalues from S$\&$P data (blue histogram) and Wachter distribution (orange line) based on various VAR($k$) settings.}
\label{SP100pic}
\end{figure}

We illustrate our asymptotic theorems on the S$\&$P100 data. We use logarithms of weekly prices of assets in the S$\&$P100 over ten years (January 1, 2010, to January 1, 2020), which gives us $522$ observations across time. More detailed  description of the variables can be found in \citet[Section 6]{BG}.

For the S$\&$P100 data set we use Procedure \ref{sscc_BG} to calculate $\tilde{\lambda}_1,\ldots,\tilde{\lambda}_N$. We do this for various choices of $k$. The case $k=1$ (VAR(1)) corresponds to \citet{BG}. The results are shown in Figure \ref{SP100pic}.

We see a striking match between the histograms and Wachter densities for all $k=1,2,3,4$, which is an indication that the setting of Theorem \ref{Theorem_empirical} is a proper modeling for the S$\&$P data. We do not see any outliers in the largest eigenvalues, which would appear if there were cointegration. Indeed, our test statistics based on Theorem \ref{Theorem_J_stat} are $-0.28,\,-0.71,\,-1.07,\,-3.84$ for $k=1,2,3,4$, respectively, while the $5\%$ and $10\%$ critical values are $0.97$ and $0.44$.
Because the former numbers are smaller than the latter numbers, we do not reject the ``no cointegration'' hypothesis.

\subsection{Cryptocurrencies}

In this subsection we redo the calculations for cryptocurrencies instead of S$\&$P stocks. We use the data from \citet{crypto_coint_paper} (25 series from the Github repository). Logarithms of daily prices for two years (from October $5$, $2017$, to October $4$, $2019$) are shown in Figure \ref{crypto_lprice_pic}. The results of Procedure \ref{sscc_BG} used to calculate $\tilde{\lambda}_1,\ldots,\tilde{\lambda}_N$ are shown in Figure \ref{crypto_pic}.

Similarly to the S$\&$P example in the previous subsection, we see a match between the eigenvalues and the Wachter distribution.\footnote{The Wachter distribution depends on the order of the VAR, $k$, and on the ratio $T/N$. Thus, the orange curves in Figures \ref{SP100pic} and \ref{crypto_pic} have different shapes and supports.} However, there is a major difference between Figures \ref{SP100pic} and \ref{crypto_pic}: The latter has around 3 eigenvalues to the right of the support of the orange curve (Wachter distribution). This is an indication of the presence of approximately $3$ cointegrating relationships. This is reinforced by our test, which has p-values below $0.01$ for all four choices of the order of VAR($k$) ($k=1,2,3,4$).

\begin{figure}[t!]
	\centering
	{\scalebox{0.66}{\includegraphics{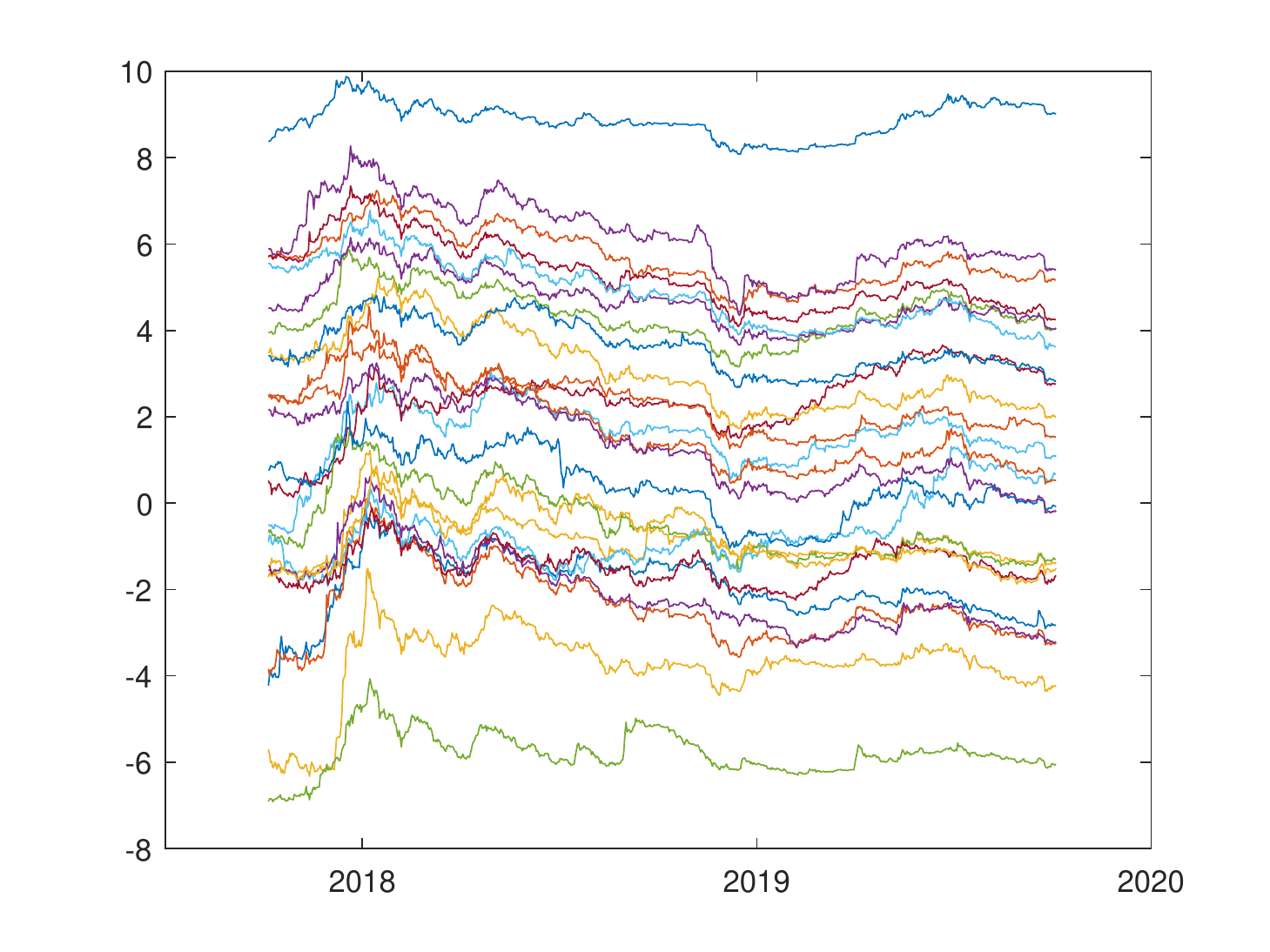}}}
	\caption{Time series of daily log prices for 25 cryptocurrencies.}
	\label{crypto_lprice_pic}
\end{figure}

\begin{figure}[t!]
\begin{subfigure}{.44\textwidth}
  \centering
  \includegraphics[width=1.0\linewidth]{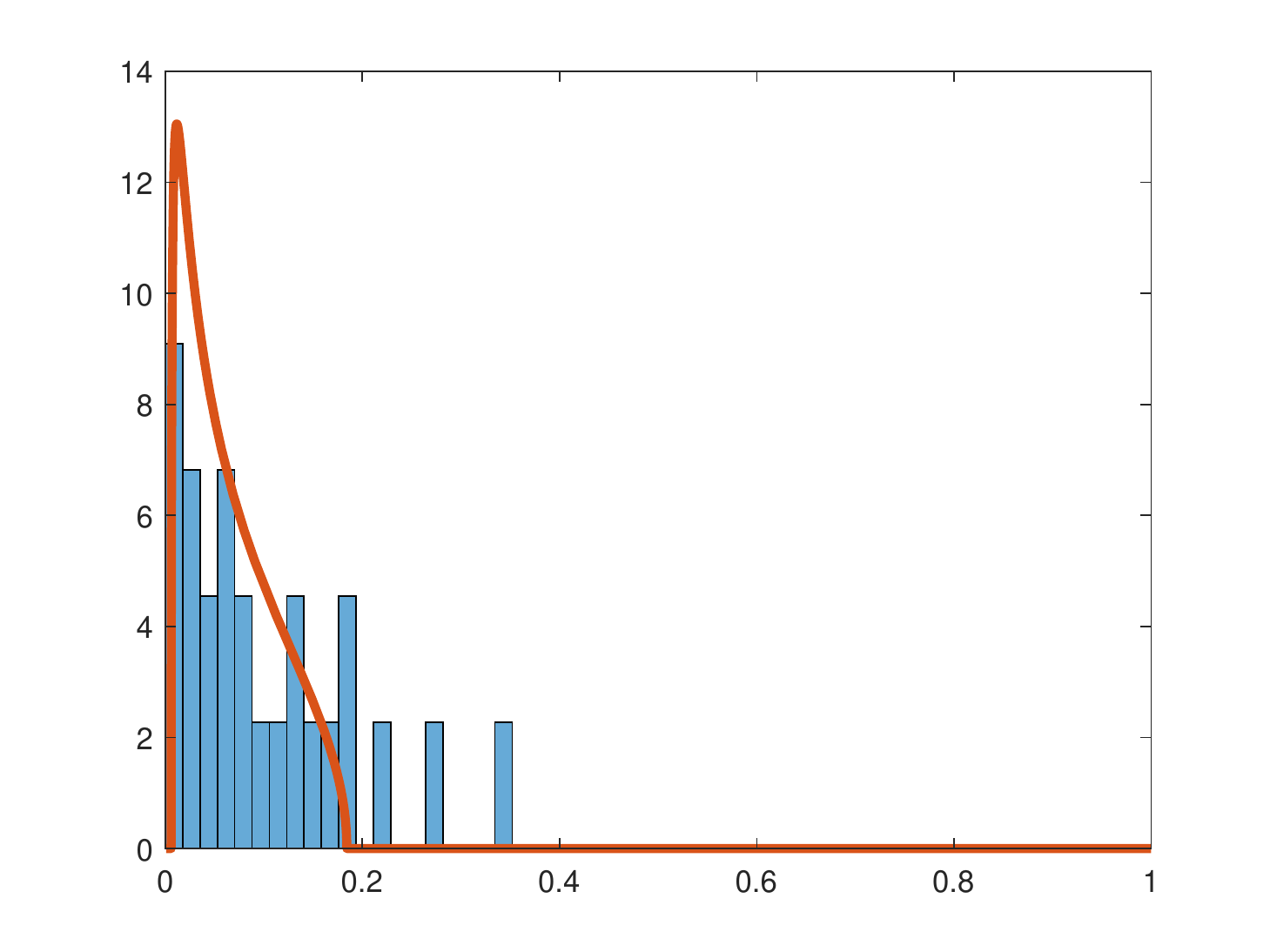}
  \caption{VAR($1$).}
  \label{cryptovar1}
\end{subfigure}%
\begin{subfigure}{.44\textwidth}
  \centering
  \includegraphics[width=1.0\linewidth]{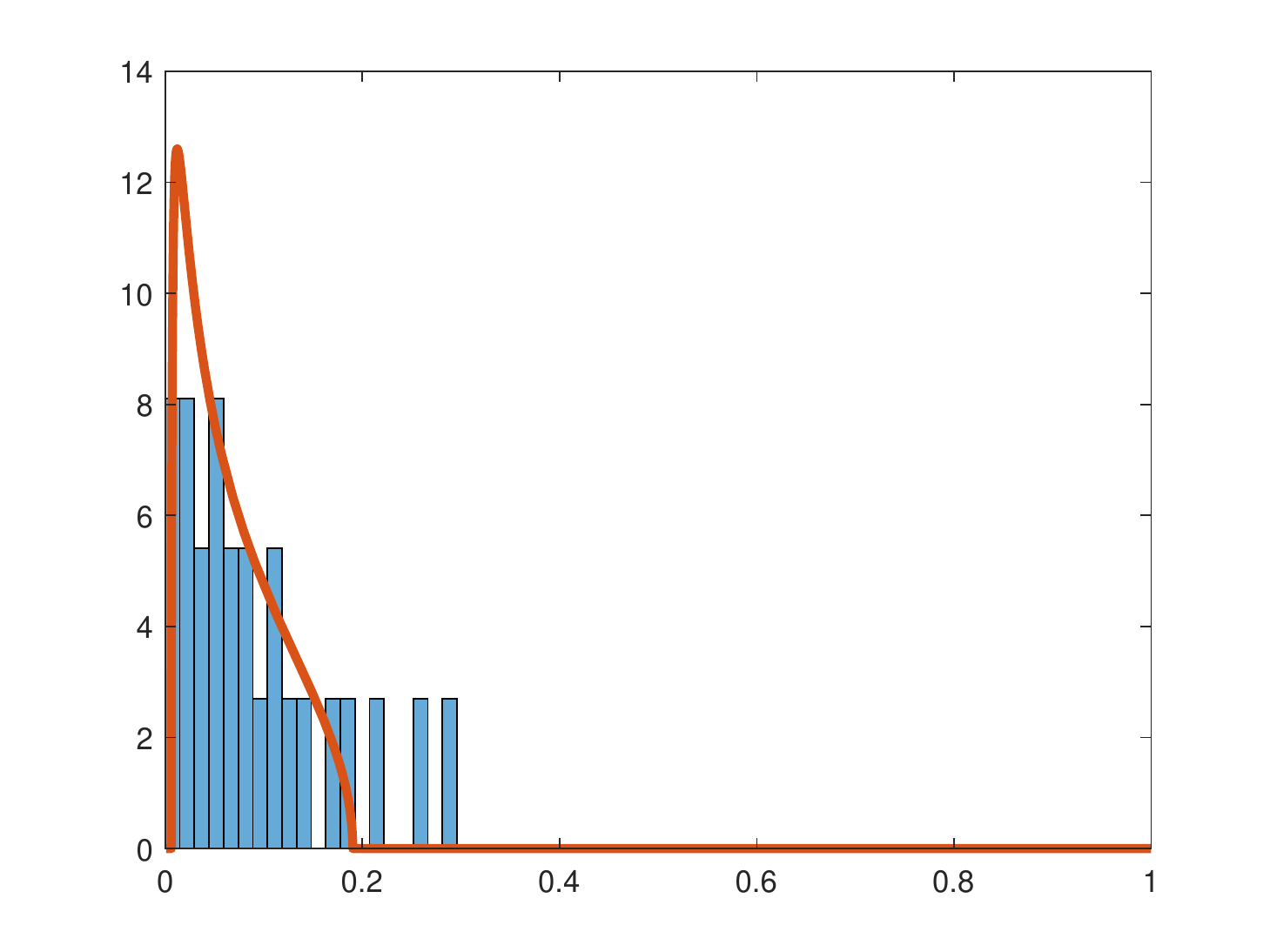}
  \caption{VAR($2$).}
  \label{cryptovar2}
\end{subfigure}

\begin{subfigure}{.44\textwidth}
  \centering
  \includegraphics[width=1.0\linewidth]{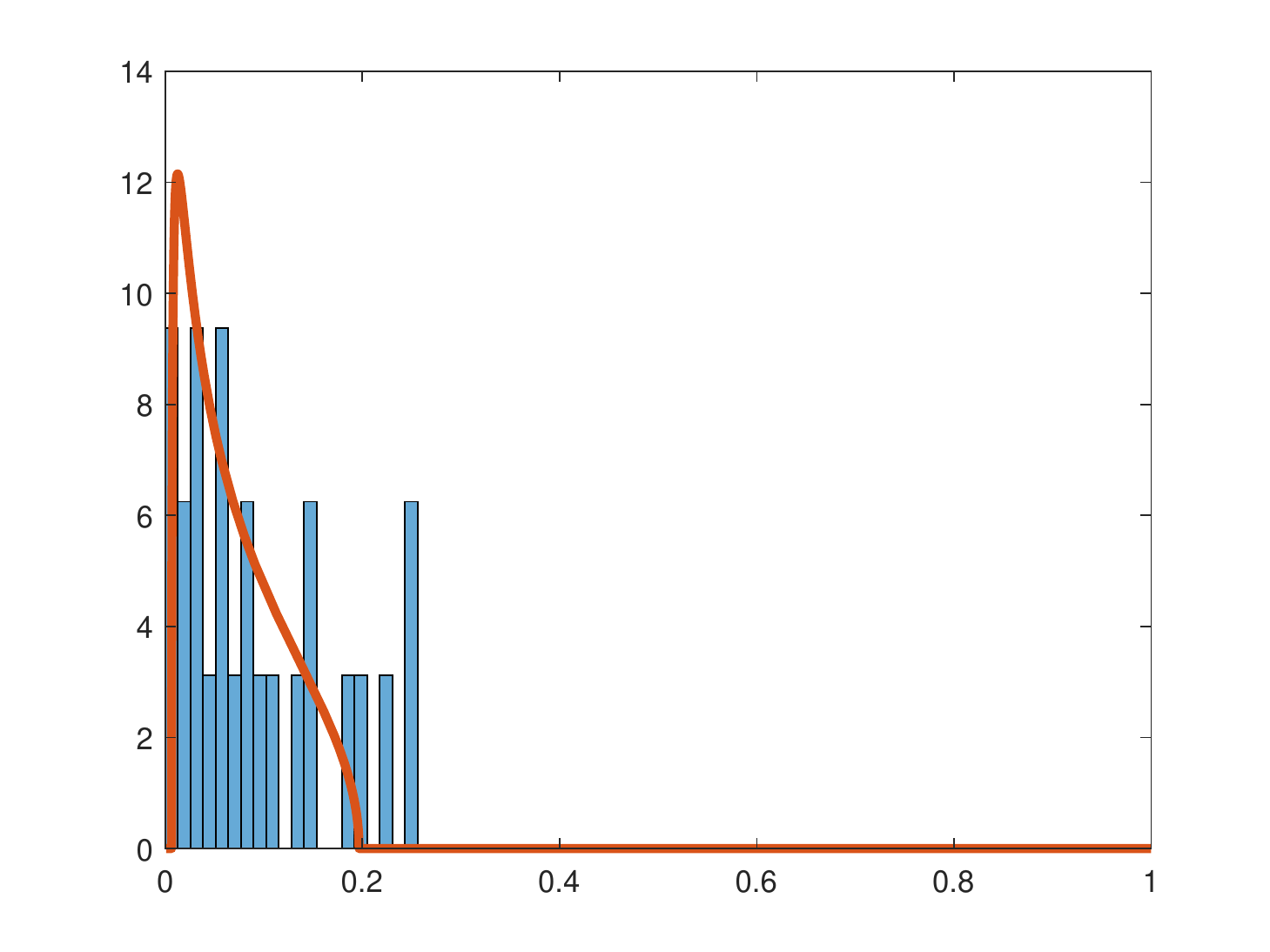}
  \caption{VAR($3$).}
  \label{cryptovar3}
\end{subfigure}
\begin{subfigure}{.44\textwidth}
  \centering
  \includegraphics[width=1.0\linewidth]{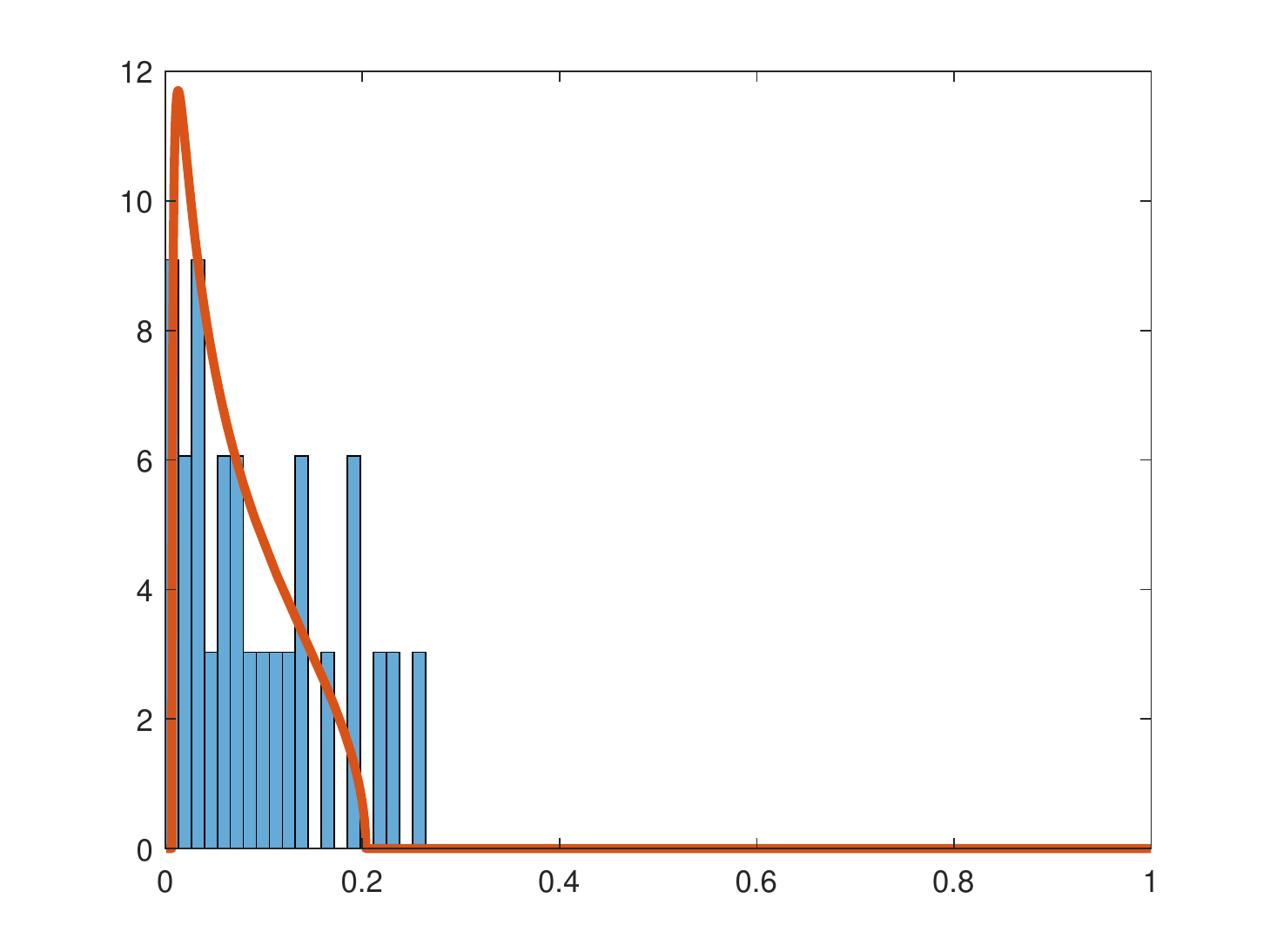}
  \caption{VAR($4$).}
  \label{cryptovar4}
\end{subfigure}
\caption{Eigenvalues from cryptocurrency data (blue histogram) and Wachter distribution (orange line) based on various VAR($k$) settings.}
\label{crypto_pic}
\end{figure}

The difference in results for traditional stocks and cryptocurrencies can be explained by the fact that the cryptocurrency market is still very inefficient and, thus, has numerous trading possibilities. The presence of cointegration can be one such inefficiency.

\section{Conclusion}
\label{Section_conclusion}

High-dimensional data are becoming increasingly widespread in economics and other sciences. Thus, appropriate machinery for handling such data is needed. We believe that the use of random matrix theory is inevitable for the development of the area: As soon as dimensions are high, random matrices start to contribute. Along these lines, in our paper the central role is played by random matrix objects: the Wachter distribution, Airy$_1$ point process, and Jacobi ensemble.

The present paper focused on nonstationary high-dimensional VARs and presented the asymptotic limit of the Johansen LR test for cointegration and its modifications. Because the limit is nonrandom, the appropriate second-order statistic was derived, and a new test for the presence of cointegration was proposed. The new test builds upon the Johansen LR, while having some extra modifications. This new test is suitable for a vector autoregression of order $k$ with an intercept.

The main focus of the present paper is the null of no cointegration. The next essential step is to be able to test whether the cointegration rank is $r$ for $r>0$, i.e.,~to find the true rank of cointegration. Heuristics for finding the correct value of $r$ can already be seen in our simulations and data sets (cf.\ Figures \ref{small_rk_pic}, \ref{SP100pic}, and \ref{crypto_pic}): When there are no cointegrations, all eigenvalues (squared canonical correlations) are to the left of the end-point of the support of the Wachter distribution. In contrast, we expect each cointegrating relationship to lead to an eigenvalue between the right end-point of the support of the Wachter distribution and $1$. Identifying the exact conditions under which this heuristic is correct represents an important problem for future research.

\section{Appendix 1: Proofs}

This appendix contains the proofs of Theorems \ref{Theorem_empirical}, \ref{Theorem_vark_approximation}, and \ref{Theorem_J_stat} from the main text.

First, in Section \ref{Section_Jacobi} we collect known statements about the asymptotics of the Jacobi ensemble of Definition \ref{Definition_Jacobi}, which will be used in our subsequent proofs.

Second, in Theorem \ref{Theorem_Var_k} of Section \ref{Section_Jacobi_model} we introduce a novel random matrix model for the Jacobi ensemble. Our proof of Theorem \ref{Theorem_Var_k} proceeds through certain intricate inductive computations of large-dimensional matrix integrals.

Third, in Section \ref{Section_perturbation} we connect the matrix model of Section \ref{Section_Jacobi_model} to the cointegration setting: for that we use the rotational symmetry of the Gaussian law to express the squared sample canonical correlations solving \eqref{eq_vark_eig} under the hypothesis $\widehat H_0$ in terms of a certain deterministic orthogonal matrix. Replacing this deterministic matrix by a uniformly random one, we arrive at the Jacobi ensemble of Theorem \ref{Theorem_Var_k}. We proceed by bounding the error in this replacement, which relies on the rigidity estimate for orthogonal matrices \eqref{eq_MM_bound}, but needs special care due to various matrix inversions involved in our procedures. Eventually, we arrive at Theorem \ref{Theorem_vark_approximation}. This theorem is our main technical result. Combining Theorem \ref{Theorem_vark_approximation} with Proposition \ref{Theorem_Jacobi_as} from Section \ref{Section_Jacobi}, we finish the proof of Theorem \ref{Theorem_J_stat} from the main text.

Finally, in Section \ref{Section_small_rank} we prove Theorem \ref{Theorem_empirical} by combining Theorem \ref{Theorem_vark_approximation} with Proposition \ref{Theorem_Jacobi_as} of Section \ref{Section_Jacobi} and general statements about small rank perturbations.

\subsection{Asymptotic of Jacobi ensemble}
\label{Section_Jacobi}

In this section we review the asymptotic results for the Jacobi ensemble $\J(N; p,q)$ introduced in the Definition \ref{Definition_Jacobi} as $N\to\infty$.

We assume that as $N\to\infty$, also $p,q\to\infty$, in such a way that
\begin{equation}
\label{eq_asymptotic_pars}
\frac{p}{N}=\frac{1}{2} (\p-1), \quad \p\ge 1,\qquad \frac{q}{N}=\frac{1}{2} (\q-1), \quad \q\ge 1,
\end{equation}
where $\p$ and $\q$ are two parameters, which stay bounded away from $1$ and from $\infty$ as $N\to\infty$.\footnote{\citet[Theorem 1]{Johnstone_Jacobi} suggests to use $\frac{p-1}{N}$ and $\frac{q-1}{N}$ instead of $\frac{p}{N}$ and $\frac{q}{N}$, respectively, in order to improve the speed of convergence. However, we found in \cite{BG} that for the tests in VAR($1$) case the usefulness of this correction depends on the exact value of the ratio $T/N$ and we are not going to pursue this direction here.}  We further define the \emph{equilibrium measure} $\mu_{\p,\q}$ of the Jacobi ensemble through:
\begin{equation}
\label{eq_Jacobi_equilibrium}
\mu_{\p,\q}(x)\, d x = \frac{\p+\q}{2\pi} \cdot \frac{\sqrt{(x-\lambda_-)(\lambda_+-x)}}{x (1-x)} \mathbf 1_{[\lambda_-,\lambda_+]}\, \d x,
\end{equation}
where the support $[\lambda_-,\lambda_+]$ of the measure is defined via
\begin{equation}
\lambda_\pm=\frac{1}{(\p+\q)^2}\left(\sqrt{\p(\p+\q-1)}\pm \sqrt{\q}  \right)^2.
\end{equation}
One can check that $0<\lambda_-<\lambda_+<1$ for every $\p,\q>1$.
Further, define
\begin{equation}
c_\pm=\frac{(\p+\q)}{2} \frac{\sqrt{\lambda_+-\lambda_-}}{\lambda_\pm (1-\lambda_\pm)},
\end{equation}
and note that
$$
\mu_{\p,\q}(x-\lambda_{\pm})\approx \frac{c_\pm}{\pi} \sqrt{|x-\lambda_{\pm}|}, \text{ as } x\to\lambda_{\pm}\quad \text{ inside }\quad  [\lambda_-,\lambda_+],
$$
where the normalization $\frac{1}{\pi} \sqrt{|x-\lambda_{\pm}|}$ was chosen to match the behavior of the Wigner semicircle law $\frac{1}{2\pi} \sqrt{4-x^2}$ near edges $\pm 2$.
\begin{proposition}[See \citet{Johnstone_Jacobi}, \cite{forrest}, and \cite{HanPanZhang_2016}]
	\label{Theorem_Jacobi_as}
	Suppose that $N,p,q\to\infty$ in such a way that $\p\ge 1 $ and $\q\ge 1 $ in \eqref{eq_asymptotic_pars} stay bounded.  For the second conclusions we additionally assume that $\q$ is bounded away from $1$ and for the third conclusion we additionally require $\p$ to be bounded away from $1$. Let $x_1\ge x_2\ge \dots\ge x_N$ be $N$ random eigenvalues of Jacobi ensemble $\J(N;p,q)$. Then
	\begin{enumerate}
		\item $\displaystyle \lim_{N\to\infty} \left|\frac{1}{N} \sum_{i=1}^N \delta_{x_i}- \mu_{\p,\q} \right|=0,$ weakly in probability.
		
		This means that for any continuous function $f(x)$ we have convergence in probability:
		\begin{equation}
		\label{eq_Jacobi_LLN}
		\lim_{N\to\infty} \left|\frac{1}{N}\sum_{i=1}^N f(x_i)-\int_{0}^1 f(x)\mu_{\p,\q}(x)\d x\right|=0.
		\end{equation}
		\item For $\{\aa_i\}_{i=1}^{\infty}$ as in Proposition \ref{Proposition_Airy_Gauss}, we have convergence in finite-dimensional distributions for the largest eigenvalues:
		\begin{equation}
		\label{eq_Jacobi_up_edge}
		\lim_{N\to\infty} \left\{ N^{2/3} c_+^{2/3} \left(x_i- \lambda_+\right)  \right\}_{i=1}^{\infty}\to \{\aa_i\}_{i=1}^{\infty}.
		\end{equation}
		In particular, $N^{2/3} c_+^{2/3} \left(x_1- \lambda_+\right)$ converges to the Tracy-Widom distribution $F_1$.
		\item We also have convergence in distribution for the smallest eigenvalues\footnote{The limiting processes $\{\aa_i\}_{i=1}^{\infty}$ arising for the largest and smallest eigenvalues are independent.}
		\begin{equation}
		\label{eq_Jacobi_low_edge}
		\lim_{N\to\infty} \left\{ N^{2/3} c_-^{2/3} \left(\lambda_--x_{N+1-i}\right)  \right\}_{i=1}^{\infty}\to \{\aa_i\}_{i=1}^{\infty}.
		\end{equation}
	\end{enumerate}
\end{proposition}

\subsection{A new model for the Jacobi ensemble}

\label{Section_Jacobi_model}

The Jacobi ensemble appearing in Theorem \ref{Theorem_vark_approximation} originates in the following computation of exact distribution. In addition to real symmetric matrices ($\beta=1$ in the usual random matrix notations) it also covers the case of complex Hermitian matrices ($\beta=2$).

\begin{theorem} \label{Theorem_Var_k}
	Fix $k=1,2,\dots$ and assume $\T \ge (k+1)N$. Let $\V$ be an $N$--dimensional subspace in the $\T$--dimensional space and let $O$ be uniformly random orthogonal $\T\times \T$ matrix with determinant $1$ if $\beta=1$ ($O$ is a uniformly random unitary matrix if $\beta=2$). Let $P$ be an orthogonal projector on the space orthogonal to $O\V$, $O^2\V$,\dots, $O^{k-1}\V$. Let $P_1$ be a projector on the subspace  $P\V$ and $P_2$ be a projector on the subspace $PO^{k-1}(\1_\T+O)^{-1}\V$. Then non-zero eigenvalues of $P_1P_2P_1$ coincide with those of the Jacobi ensemble of  $N\times N$ real symmetric if $\beta=1$ (complex Hermitian if $\beta=2$)   matrices of density proportional to
	\begin{equation}
	\label{eq_var_k_answer}
	\det (\mathcal M)^{\frac{\beta}{2}N+\beta-2} \det(\1_N-\mathcal M)^{\frac{\beta}{2}(\T-(k+1)N+1)-1}\, d \mathcal M,\quad 0\le \mathcal M\le \1_N, \qquad \beta=1,2.
	\end{equation}
\end{theorem}
\begin{remark}
	\label{Remark_replace_in_def}
	We can replace $P O^{k-1}(\1_\T+O)^{-1}$ in the definition of $P_2$ with $P O^{k}(\1_\T+O)^{-1}$. Indeed, if $k>1$, then $P O^{k}(\1_\T+O)^{-1}\V+ P O^{k-1}(\1_\T+O)^{-1}\V= P O^{k-1} \V=0$. If $k=1$, then $P$ disappears (one can say that it becomes an identical operator) and the random operators $(\1_\T+O)^{-1}$ and $O(\1_\T+O)^{-1}=(\1_T+O^{-1})^{-1}$ has the same law, since the uniform measure on the orthogonal group is invariant under the inversion $O\mapsto O^{-1}$.
	
	By a similar argument applied inductively we can replace $P O^{k-1}(\1_\T+O)^{-1}$ with $P O(\1_\T+O)^{-1}$. Note, however, that for $k>1$ we can not replace it with $P (\1_\T+O)^{-1}$.
\end{remark}

The $k=1$ case of Theorem \ref{Theorem_Var_k} is established in \cite[Theorem 6 in Appendix]{BG}. The proof of Theorem \ref{Theorem_Var_k} uses the following three auxiliary ingredients.

\begin{lemma}[Block matrix inversion formula] \label{Lemma_block_inversion} For matrices $\mathbf A$, $\mathbf B$, $\mathbf C$, $\mathbf D$, we have:
	\begin{equation}
	\label{eq_block_inversion}
	\begin{pmatrix} \mathbf{A}& \mathbf{B}\\ \mathbf{C}& \mathbf{D}\end{pmatrix}^{-1} = \begin{pmatrix} \mathbf{Q}  & - \mathbf{Q} \mathbf{B}\mathbf{D}^{-1} \\ - \mathbf{D}^{-1}\mathbf{C} \mathbf{Q}
	& \mathbf{D}^{-1}+ \mathbf{D}^{-1} \mathbf{C} \mathbf{Q} \mathbf{B} \mathbf{D}^{-1} \end{pmatrix} ,\qquad \mathbf{Q} =  (\mathbf{A}- \mathbf{B} \mathbf{D}^{-1}\mathbf{C} )^{-1},
	\end{equation}
\end{lemma}
\begin{proof} Direct computation. \end{proof}

\begin{lemma}[Cayley transform] Suppose that all eigenvalues of $N\times N$ matrix $O$ are different from $-1$. Then $O$ is an orthogonal matrix with determinant $1$, if and only if the matrix $\R$ defined through
	\begin{equation}
	\label{eq_Cayley}
	\R=(\1_N-O)(\1_N+O)^{-1}=\frac{\1_N-O}{\1_N+O},\quad \text{ so that } \quad O=\frac{\1_N-\R}{\1_N+\R}
	\end{equation}
	is skew-symmetric, i.e., it satisfies $\R^*=-\R$.
\end{lemma}
\begin{proof} The formulas \eqref{eq_Cayley} imply that $O^*=O^{-1}$ if and only if $\R^*=-\R$. On the other hand, for a skew-symmetric $\R$, we have $\det(\1_N+\R)=\det(\1_N+\R^*)=\det(\1_N-\R)$. Hence, $\det(O)=\det\left(\frac{\1_N-\R}{\1_N+\R}\right)=1$.
\end{proof}

\begin{lemma}\label{Lemma_Haar_project} Choose two positive integers $M$, $N$ and set $\T=M+N$. Let $O$ be a uniformly random $\T\times \T$ orthogonal matrix with determinant $1$. Write $O$ in the block form according to $\T=M+N$ splitting:
	$$
	O=\begin{pmatrix} A& B\\ C& D\end{pmatrix}
	$$
	Then $\tilde O:= A-B(\1_N+D)^{-1} C$ is a $M\times M$ orthogonal matrix of determinant $1$ uniformly distributed among all such matrices. In addition, the  random matrices $\tilde O$ and $B(\1_N+D)^{-1}$ are independent.
\end{lemma}
\begin{remark} \label{Remark_off_diag_corner} The law of $\widehat W:=B(\1_N+D)^{-1}$ is explicit. The computation \eqref{eq_Haar_decomposition} below implies that the density of $\widehat W$ is proportional to
	$$
	\det\bigl(\1_N+\widehat W^* \widehat W\bigr)^{1/2-M-N/2}\, d \widehat W.
	$$
\end{remark}
\begin{remark} The computation of the law of $\tilde O$ is mentioned in \cite{Olshanski_Harmonic} and \cite{Neretin_Hua} with its roots going back to \cite{Hua}. However, we could not locate the statements concerning also $B(\1_N+D)^{-1}$ in the literature.
\end{remark}

\begin{proof}[Proof of Lemma \ref{Lemma_Haar_project}]
	First, note that the distribution of eigenvalues of $D$ is absolutely continuous and, hence, $\1_N+D$ is almost surely invertible and the matrix $\tilde O$ is well-defined. Our next task is to show that $\tilde O$ is an orthogonal matrix with determinant $1$. We use Cayley transform for that. Combining \eqref{eq_Cayley} with \eqref{eq_block_inversion} we have
	\begin{multline}
	\label{eq_R_form}
	\R=\frac{\1_\T-O}{\1_\T+O}=\frac{2}{\1_\T+O}-\1_\T\\ =  \begin{pmatrix} 2 Q-\1_M  && - 2 Q B (\1_N+D)^{-1} \\ - 2 (\1_N+D)^{-1} C  Q
	&& 2(\1_N+D)^{-1}+ 2(\1_N+D)^{-1} C Q B (\1_N+D)^{-1}-\1_N \end{pmatrix},\\  Q =  (\1_M+A- B (\1_N+D)^{-1}C )^{-1}.
	\end{multline}
	Since $\R$ is skew-symmetric, so is its top--left $M\times M$ corner $\tilde\R=2 Q-\1_M$. We claim that $\tilde O$ is the Cayley transform of $\tilde \R$, which would imply that $\tilde O$ is orthogonal of determinant $1$. Indeed,
	\begin{equation}
	\label{eq_tO_tR}
	\frac{\1_M-\tilde \R}{\1_M+\tilde \R}=  \frac{2 \1_M-2Q}{2Q}= Q^{-1}-\1_M=A- B (\1_N+D)^{-1}C =\tilde O.
	\end{equation}
	
	It remains to compute the distributions of $\tilde O$ and $B(\1_N+D)^{-1}$ and show their independence. In terms of $\R$ the distribution of $O$ (as a uniformly random orthogonal matrix of determinant $1$) is given by the density proportional to
	\begin{equation}
	\label{eq_Cayley_distribution_1}
	\det(\1_\T-\R^2)^{-\frac{1}{2}\T+\frac{1}{2}}d\R=\det(\1_\T-\R)^{1-\T}d\R=\det(\1_\T+\R)^{1-\T}d\R,
	\end{equation}
	see, e.g., \citet[(2.55)]{forrest} and notice that $\det(\1_\T-\R)=\det(\1_T+\R)$ for the two equalities. We rewrite the block form \eqref{eq_R_form} of $\R$ as
	$$
	\R=\begin{pmatrix} \tilde \R & -W \\ W^* & \R_2\end{pmatrix},
	$$
	where $\R_2$ is  $N\times N$ skew-symmetric and $W$ is an arbitrary $M\times N$ matrix. We further introduce the notation
	$\widehat W:= (\1_M+\tilde \R)^{-1} W$. Recalling that $\tilde\R=2 Q-\1_M$, we transform
	\begin{equation}
	\label{eq_wWdef}
	\widehat W= 2(\1_M+\tilde \R)^{-1} Q B (\1_N+D)^{-1}= B (\1_N+D)^{-1}.
	\end{equation}
	We also define
	$$
	\widehat \R_2:= \bigl(\1_N+ \widehat W^* \widehat W\bigr)^{-1/2} \bigl(-\R_2+\widehat W^*\tilde \R \widehat W\bigr)  \bigl(\1_N+ \widehat W^* \widehat W\bigr)^{-1/2}.
	$$
	Note that $(\tilde R, \widehat W, \widehat R_2)$ is an alternative parameterization of $\R$, in which  $\widehat W$ is an arbitrary $M\times N$ matrix and $\widehat R_2$ is an arbitrary $N\times N$ skew-symmetric matrix.  Using the formula for the determinant of a block matrix
	\begin{equation}
	\label{eq_block_det}
	\det \begin{pmatrix} \mathbf{A}& \mathbf{B}\\ \mathbf{C}& \mathbf{D}\end{pmatrix}=\det \mathbf{A}\,\cdot\, \det (\mathbf{D} - \mathbf{C} \mathbf{A}^{-1}\mathbf{B}),
	\end{equation}
	we rewrite \eqref{eq_Cayley_distribution_1} as
	\begin{multline}
	\label{eq_Haar_decomposition_0}
	\det(\1_\T-\R)^{1-\T}\,d\tilde \R\, dW\, d \R_2=\det\begin{pmatrix} \1_M-\tilde \R & W \\ -W^* & \1_N-\R_2\end{pmatrix}^{1-\T}\,d\tilde \R\, dW\, d \R_2
	\\
	=\det(\1_M-\tilde \R)^{1-\T} \det(\1_N-\R_2 + W^* (\1_M-\tilde \R)^{-1} W)^{1-\T}
	\,d\tilde \R\, dW\, d \R_2
	\\=\det\bigl(\1_M-\tilde \R\bigr)^{1-\T}   \det\bigl(\1_N-\R_2+\widehat W^* (\1_M+\tilde \R) \widehat W\bigr)^{1-\T}\,d\tilde \R\, dW\, d \R_2
	\end{multline}
	Further, notice that
	$$
	\bigl(\1_N+\widehat W^* \widehat W\bigr)^{1/2}\bigl(\1_N+ \widehat \R_2\bigr)\bigl(\1_N+\widehat W^* \widehat W\bigr)^{1/2}= \1_N-\R_2+\widehat W^* (\1_M+\tilde \R) \widehat W.
	$$
	Hence, the last line of \eqref{eq_Haar_decomposition_0} is transformed into
	\begin{multline}
	\label{eq_Haar_decomposition}
	\det\bigl(\1_M-\tilde \R\bigr)^{1-\T} \det\bigl(\1_N+\widehat W^* \widehat W\bigr)^{1-\T}   \det\bigl(\1_N+ \widehat \R_2\bigr)^{1-\T}\,d\tilde \R\, dW\, d \R_2
	\\
	=\det\bigl(\1_M+\tilde \R\bigr)^{1-M} \det\bigl(\1_N+\widehat W^* \widehat W\bigr)^{1/2-M-N/2} \det\bigl(\1_N+ \widehat \R_2\bigr)^{1-M-N}\, d\tilde \R\, d \widehat W\, d \widehat \R_2,
	\end{multline}
	where in the last line we use $\T=M+N$ and change variables $d W\mapsto d\widehat W$ and $d \R_2\mapsto d\widehat \R_2$ using the general Jacobian computations:
	\begin{itemize}
		\item The map $Z\mapsto Q Z$ on $n\times m$ matrices has the Jacobian
		\begin{equation}
		\label{eq_mult_Jacobian}
		\left| \frac{\partial (Q Z)}{\partial Z}\right| = |\det Q|^{m},
		\end{equation}
		\item The map $Z\mapsto Q Z Q^*$ from the space of $n\times n$ skew-symmetric  matrices to itself has the Jacobian
		\begin{equation}
		\label{eq_skew_Jacobian}
		\left|\frac{\partial (Q Z Q^*)}{\partial Z}\right|=|\det Q|^{n-1}.
		\end{equation}
	\end{itemize}
	The first identity \eqref{eq_mult_Jacobian} follows from the observation that each column of $Z$ is transformed by linear map $Q$ and there are $m$ such columns. The second is similar and we refer to   \citet[(1.35)]{forrest} for details.
	
	The key important feature of the last line of \eqref{eq_Haar_decomposition} is that it has a product form, which implies the joint independence of $\tilde \R$, $\widehat W$, and $\widehat \R_2$. Hence, the density of $\tilde \R$ is proportional to $\det\bigl(1+\tilde \R\bigr)^{1-M}d \tilde \R$. Comparing with \eqref{eq_Cayley_distribution_1} and noting that the dimension changed from $\T$ to $M$, we conclude that $\tilde O$ is a uniformly random $M\times M$ orthogonal matrix of determinant $1$.
	
	Next, recalling that $\tilde O$ is a deterministic function of $\tilde \R$ by \eqref{eq_tO_tR}, we conclude that  $\tilde O$ is independent with $\widehat W$, which is precisely $B (\1_N+D)^{-1}$ by \eqref{eq_wWdef}.
\end{proof}

\begin{proof}[Proof of Theorem \ref{Theorem_Var_k}] We only give a proof for the real case $\beta=1$; the complex case can be proven by the same argument. The proof is induction in $k$ with base case $k=1$ being  \cite[Theorem 6 in Appendix]{BG} and the induction step being based on Lemma \ref{Lemma_Haar_project}.
	
	{\bf Step 1.} We first note that the particular choice of \emph{deterministic} space $\V$ in the statement of the theorem is not important: any other deterministic choice of $\V$ can be achieved by a change of basis of the $\T$--dimensional space, which keeps the probability distribution of $O$ and, hence, entire construction invariant. In particular, the probability distribution of $P_1 P_2 P_1$ is unchanged. However, we need to be more careful, if we would like to make $\V$ random, as correlations with $O$ might cause issues.
	
	\smallskip
	
	{\bf Step 2.} Take any $r \in\mathbb Z$. We claim that replacement of $\V$ with $O^r\V$ everywhere in the statement of Theorem \ref{Theorem_Var_k} does not change the eigenvalues of $P_1 P_2 P_1$. Indeed, the only important feature of $O^r$ here is that it is an orthogonal operator commuting with $O$. Hence,  the change $\V\mapsto O^r\V$ leads to the image of the projector $P$ being multiplied by $O^r$; in more details, the transformation takes the form $P\mapsto O^r P O^{-r}$. Further, $P\V$ gets transformed to $O^r P\V$ and $P_1$ undergoes a similar transformation: $P_1\mapsto O^r P_1 O^{-r}$. The same is true for $P_2$:
	it undergoes the transformation $P_2\mapsto O^r P_2 O^{-r}$. We conclude that the product $P_1 P_2 P_1$ is transformed into $O^{r} P_1 P_2 P_1 O^{-r}$. Since conjugations do not change eigenvalues, we are done.

	\smallskip
	
	The arguments of Steps 1 and 2 might give a feeling that we can actually replace $\V$ by any random space. However, this is not the case. Repeating the same arguments, we see that replacement $\V\mapsto A \V$ leads to the same eigenvalues of the projector $P_1 P_2 P_1$ as if we replaced $O\mapsto A^* O A$. In both Steps 1 and 2 $O$ had the same distribution as $A^* O A$, hence, the eigenvalues were unchanged. But in general, if $A$ is correlated with $O$ in a non-trivial way, then the distribution might change.\footnote{For instance, if $\V$ is spanned by eigenvectors of $O$, then the spaces $O\V$, $O^2 \V$,\dots,$O^{k-1} \V$ all coincide, which is a very different behavior from the case of deterministic $\V$.}
	
	\smallskip
	
	{\bf Step 3.} We now transform the statement of Theorem \ref{Theorem_Var_k} by replacing $\V$ with $O^{1-k} \V$ and further replacing $O$ by $O^{-1}$ everywhere. Since the uniform (Haar) measure on the orthogonal matrices is invariant under inversion, the law of eigenvalues of $P_1 P_2 P_1$ is unchanged and the ingredients of Theorem \ref{Theorem_Var_k} are now as follows:
	
	\begin{itemize}
		\item $O$ is a uniformly random $\T\times \T$ orthogonal matrix with determinant $1$ and $\V$ is an arbitrary (deterministic) $N$--dimensional subspace of the $\T$--dimensional space, whose choice is irrelevant for the statement.
		\item $P$ is the projector on the orthogonal complement of $O^{k-2}\V$, $O^{k-3}\V$, \dots, $O\V$, $\V$.
		
		\item $P_1$ is the projector on the subspace $P O^{k-1}\V$ and $P_2$ is the projector on the subspace $P O(\1_T+O)^{-1} \V$. (The latter can be replaced by $P  (\1_T+ O)^{-1} \V$ without changing the outcome. Indeed, for that we need start from $PO^{k}(\1_T+O)^{-1}$ instead of ${PO^{k-1}(\1_T+O)^{-1}}$, which is possible by Remark \ref{Remark_replace_in_def}).
		
		\item The claim is that the eigenvalues of $P_1 P_2 P_1$ are distributed as \eqref{eq_var_k_answer}.
	\end{itemize}
	We are going to prove this last statement by induction in $k$. For that we choose $\V$ to be the span of the last $N$ coordinate vectors, split $\T=M+N$ with $M=\T-N$ and project everything on the first $M$ coordinate vectors (which are orthogonal complement to $\V$). We rely on Lemma \ref{Lemma_Haar_project} and use $A,B,C,D$ and $\tilde O$ notation from that lemma.

	\smallskip
	
	{\bf Step 4.} We claim that the subspace in $M$--dimensional space spanned by the first $M$ coordinates of $O^{k-2}\V$, $O^{k-3}\V$, \dots, $O\V$ (since $\V$ has zero projection on the first $M$ coordinates, we do not need it here) is the same as the subspace spanned by $\tilde O^{k-3} \langle B\rangle$, $\tilde O^{k-4}\langle B\rangle$, \dots, $\langle B \rangle$, where  $\langle B \rangle$ is $N$--dimensional space spanned by columns of the $M\times N$ matrix $B$.
	
	Indeed, the first $M$ coordinates of $O\V$ are $\langle B \rangle$ by definition of the block structure in Lemma \ref{Lemma_Haar_project}. Further, to go from powers of $O$ to powers of $\tilde O$ we make the following observation: take a vector $w$ in $\T$--dimensional space and write it as $w={w_2 \choose w_1}$, where $w_1$ is $N$-dimensional vector (one can think of $w_1$ being in $\V$) and $w_2$ is $M$--dimensional vector (one can think of $w_2$ being in the orthogonal complement of $\V$) and write
	$$
	O{w_2\choose w_1}={u_2 \choose u_1}.
	$$
	Then the $M$--dimensional vector $u_2$ takes the form
	$$
	u_2=A w_2 + B w_1= \tilde O w_2 + B( (\1_N+D)^{-1} C w_2+ w_1).
	$$
	Since we only care about the linear span of columns and $\langle B \rangle$ already belongs to the desired linear span, the last term can be ignored and we arrive at $\tilde O w_2$, which then implies the claim.
	
	\smallskip
	
	{\bf Step 5.} Next, consider the projection of $O^{k-1} \V$ on the orthogonal complement to $O^{k-2}\V$, $O^{k-3}\V$, \dots, $O\V$, $\V$. This is the same as the the projection of the first $M$ coordinates of $O^{k-1}\V$ on the orthogonal complement (in $M$--dimensional space) to first $M$ coordinates of $O^{k-2}\V$, $O^{k-3}\V$, \dots, $O\V$. Hence, combining with the argument of Step 4, this is the same as the projection of $\tilde O^{k-2}  \langle B\rangle$ on the orthogonal complement of $\tilde O^{k-3} \langle B\rangle$, $\tilde O^{k-4}\langle B\rangle$, \dots, $\langle B \rangle$. It is convenient to note that $\langle B \rangle=\langle B (\1_N+ D)^{-1}\rangle$.
	
	\smallskip
	
	{\bf Step 6.} Finally, consider the projection of $(\1_T+O)^{-1} O \V$ on the orthogonal complement of $O^{k-2}\V$, $O^{k-3}\V$, \dots, $O\V$, $\V$. By Steps 4 and 5 this is the same as the projection of  the first $M$ coordinates of $(\1_T+O)^{-1} O \V$ on the orthogonal complement of  $\tilde O^{k-3} \langle B(\1_N+D)^{-1}\rangle$, $\tilde O^{k-4}\langle B(\1_N+D)^{-1}\rangle$, \dots, $\langle B (\1_N+D)^{-1} \rangle$. Representing $(\1_T+O)^{-1} O\V $ in the block form,  the first $M$ coordinates of $(\1_T+O)^{-1} O \V$ are the span of the columns of the sum of the top--left corner of $(\1_T+O)^{-1}$ multiplied by $B$ plus the top-right corner of $(\1_T+O)^{-1}$ multiplied by $D$. Using Lemma \ref{Lemma_block_inversion}, we get the span of the columns of
	\begin{multline*}
	(\1_M+A-B(\1_N+D)^{-1}C)^{-1} B   - (\1_M+A-B(\1_N+D)^{-1}C)^{-1}  B (\1_N+D)^{-1} D\\=(\1_M+\tilde O)^{-1} B(\1_N+D)^{-1},
	\end{multline*}
	which is the same as  $(\1_M+\tilde O)^{-1} \langle B (\1_N+D)^{-1}\rangle$.
	
	\smallskip
	{\bf Step 7.} Combining the results of Steps 6 and 7 with Lemma \ref{Lemma_Haar_project}, we identify the eigenvalues of $P_1 P_2 P_1$ with the eigenvalues of $\tilde P_1 \tilde P_2 \tilde P_1$ obtained by the following procedure:
	
	\begin{itemize}
		\item  $\tilde O$ is a uniformly random $M\times M$ orthogonal matrix with determinant $1$, where $M=\T-N$.
		
		\item $\tilde P$ is the projector on orthogonal complement of $\tilde O^{k-3}\langle B (\1_N+D)^{-1}\rangle $, $\tilde O^{k-4}{\langle B (\1_N+D)^{-1}\rangle}$, \dots, $\langle B(\1_N+D)^{-1}\rangle$.
		
		\item $\tilde P_1$ is the projector on the subspace $\tilde P \tilde O^{k-2}  \langle B (\1_N+D)^{-1}\rangle$ and $\tilde P_2$ is the projector on the subspace ${\tilde P (\1_M+\tilde O)^{-1} \langle B (\1_N+D)^{-1}\rangle}$.
	\end{itemize}
	Since  $\langle B(\1_N+D)^{-1}\rangle$ is independent from $\tilde O$ by Lemma \ref{Lemma_Haar_project}, this is the same form as the one at the end of Step 3, but with $k$ decreased by $1$, $\T$ decreased by $N$, and $\V$ replaced by  $\langle B (\1_N+D)^{-1}\rangle$. Decreasing $k$ by $1$ and $\T$ by $N$ leaves the formula \eqref{eq_var_k_answer} unchanged, hence, we can invoke the induction assumption, thus, finishing the proof.
\end{proof}

\subsection{A perturbation of the Jacobi ensemble.}

\label{Section_perturbation}

In this section we use Theorem \ref{Theorem_Var_k} to prove Theorem \ref{Theorem_vark_approximation}.

\smallskip

Recall the cyclic shift\footnote{Note that in \cite{BG} we expressed all the operators in terms of $F=L_c^{-1}$ rather than $L_c$.} operator $L_c$ acting in $T$--dimensional space. Let $V$ be the $(T-1)$--dimensional space orthogonal to the vector $(1,1,\dots,1)$, i.e., $V=\{(x_1,\dots,x_T)\mid x_1+\dots+x_T=0\}$. Note that $V$ is an invariant space for $L_c$ and let $L_V$ denote the restriction of $L_c$ on the subspace $V$.

Take a uniformly-random orthogonal (or unitary if $\beta=2$) operator $\tilde O$ acting in $(T-1)$--dimensional space $V$ and define an operator $\tilde L$ acting in $V$:
$$
\tilde L=- \tilde O L_V \tilde O^*.
$$

\begin{proposition}
	\label{Prop_gaussian_rotation}
	Assume $T> (k+1)N$ and let $\tilde L$ be as above. Take an arbitrary $N$--dimensional subspace $\mathcal U$ in $(T-1)$--dimensional space $V$. Let $P$ be the orthogonal projector on the space orthogonal to $\tilde L \mathcal U$, $\tilde L^2 \mathcal U$,\dots, $\tilde L^{k-1} \mathcal U$. Let $P_1$ be the projector on the subspace  $P \mathcal U$ and $P_2$ be the projector on the subspace $P \tilde L^{k}(\1_{V}+\tilde L)^{-1} \mathcal U$. Then the distributions of non-zero eigenvalues of $P_1P_2P_1$ coincides with that of the squared sample canonical correlations solving \eqref{eq_vark_eig} under the hypothesis $\widehat H_0$.
\end{proposition}

Comparing Proposition \ref{Prop_gaussian_rotation} with Theorem \ref{Theorem_Var_k}  and Remark \ref{Remark_replace_in_def} one notices that the differences are in restricting on the subspace $V$ (hence, decreasing the dimension by $1$) and in replacement $O\leftrightarrow \tilde L$.

\begin{proof}[Proof of Proposition \ref{Prop_gaussian_rotation}] \quad
	
	{\bf Step 1.} We start by transforming the Gaussian noise $\eps_t$. Let $\eps$ be $N\times T$ matrix, whose $t$-th column is $\eps_t$. Take any non-degenerate $N\times N$ matrix $A$ and transform $\eps\mapsto A\eps$. Thus, we leave $X_0$ unchanged and recalculate $X_t$, $1\le t \le T$. We claim that the canonical correlations solving Eq.~\eqref{eq_vark_eig} are unchanged. Indeed, the linear subspace $\mathcal W$ stays the same and so does the projector $P_{\bot \mathcal W}$. For each $t=1,2,\dots,T$, the vector $\Delta X_t$ is transformed by $\Delta X_t\mapsto A\Delta X_t+(\1_N-A)\mu$ and $\tilde X_t$ is transformed by
	$ \tilde X_t\mapsto A\tilde X_t + (\1_N-A) X_0.$
	Recall that the space $\mathcal W$ includes vector $(1,\dots,1)$, which leads to the projector $P_{\bot \mathcal W}$ canceling  the additional terms $(\1_N-A)\mu$ and $(\1_N-A)X_0$ in the last two formulas. Hence, the matrices $\tilde R_0$ and $\tilde R_k$ are transformed by $\tilde R_0\mapsto A \tilde R_0$ and $\tilde R_k\mapsto A \tilde R_k$. Therefore,
	$$
	\tilde S_{k0} \tilde S_{00}^{-1} \tilde S_{0k} \mapsto A \tilde S_{k0} \tilde S_{00}^{-1} \tilde S_{0k} A^*, \qquad \tilde S_{kk}\mapsto A \tilde S_{kk} A^*.
	$$
	We conclude that Eq.~\eqref{eq_vark_eig} is multiplied by $\det(A) \det(A^*)$ and, hence, its roots are preserved.
	
	By choosing an $A=\Lambda^{-1/2}$ the covariance matrix $\Lambda$ becomes identical. Hence, for the rest of the proof we assume without loss of generality that $\Lambda$ is identical, which means that the matrix elements of $\eps$ are i.i.d.\ standard Gaussians.
	
	\medskip
	
	{\bf Step 2.} Let us now reduce the canonical correlations solving Eq.~\eqref{eq_vark_eig} to eigenvalues for a product of projectors. By definition the canonical correlations are eigenvalues of $N\times N$ matrix
	$$
	\tilde S_{k0} \tilde S_{00}^{-1} \tilde S_{0k} \tilde S_{kk}^{-1}= \tilde R_k \tilde R_0^* (\tilde R_0 \tilde R_0^*)^{-1} \tilde R_0 \tilde R_k^* (\tilde R_k \tilde R_k^*)^{-1}
	$$
	Note that for any two rectangular matrices $A$ and $B$ of the same sizes the non-zero eigenvalues of $A B^*$ and of $B^* A$ coincide. Hence, the desired canonical correlations are also eigenvalues of $T\times T$ matrix
	$$
	\bigl[ \tilde R_0^* (\tilde R_0 \tilde R_0^*)^{-1} \tilde R_0 \bigr] \cdot \bigl[\tilde R_k^* (\tilde R_k \tilde R_k^*)^{-1} \tilde R_k \bigr].
	$$
	The last matrix is a product of two projectors:\footnote{For a closer match to the proposition that we are proving, note also that if $P_1$ and $P_2$ are projectors, then eigenvalues of $P_1P_2$ and $P_1 P_2 P_1$ are the same.} the first one projects on the space spanned by columns of $\tilde R_0^*$ and the second one projects on columns of $\tilde R_k^*$.
	
	{\bf Step 3.}
	The next step is to express via $\eps$ various matrices involved in constructing $\tilde R_0$ and $\tilde R_k$. Let $\mathcal P$ be the orthogonal projector on the subspace $V$. Under $\widehat H_0$ we have $\Delta X_t=\mu+\eps_t$. Also
	$$
	(\Delta X \mathcal P)_t= \mu+\eps_t-\frac{1}{T}\sum_{\tau=1}^T (\mu+\eps_\tau)= \eps_t-\frac{1}{T}\sum_{\tau=1}^T \eps_\tau, \qquad \Delta X \mathcal P=\eps \mathcal P .
	$$
	Further, we define the $T\times T$ summation matrix  $\Phi$. It has $1$'s below the diagonal and $0$'s on the diagonal and everywhere above the diagonal:
	
	$$
	\Phi=\begin{pmatrix} 0&0&0&\dots&0\\ 1& 0 &0&\dots &0\\ 1& 1& 0 &\dots & 0 \\ && \ddots \\ 1&1 &\dots & 1 &0\end{pmatrix}.
	$$
	We set
	$$
	\tilde \Phi =\mathcal P \Phi \mathcal P.
	$$
	By a straightforward linear algebra (see \cite[Section 9.2]{BG} for some details) one shows that the linear operator $\tilde \Phi$ preserves the space $V$ (orthogonal to $(1,1,\dots,1)$). In addition, its restriction on the subspace $V$ coincides with $L_V(\1_V-L_V)^{-1}$,
	where $\1_V$ is the identical operator acting in $V$.
	
	We can write
	\begin{multline}
	\tilde X_t= X_{t-1}-\frac{t-1}{T}(X_T-X_0)=X_0 + (t-1)\mu+ \sum_{\tau=1}^{t-1} \eps_\tau - \frac{t-1}{T} \left(T\mu +\sum_{\tau=1}^{T} \eps_\tau\right)\\= X_0 +  \sum_{\tau=1}^{t-1} \eps_\tau - \frac{t-1}{T} \sum_{\tau=1}^{T} \eps_\tau.
	\end{multline}
	We claim that $\tilde X \mathcal  P= \eps \tilde \Phi^*$. Indeed, $ \tilde X \mathcal P$ coincides with $\dbtilde X \mathcal P $, where
	$$
	\dbtilde X_t=\tilde X_t- X_0= \sum_{\tau=1}^{t-1} \eps_\tau - \frac{t-1}{T} \sum_{\tau=1}^{T} \eps_\tau=\sum_{\tau=1}^{t-1} \left( \eps_\tau - \frac{1}{T} \sum_{s=1}^{T} \eps_s\right).
	$$
	Since $\Phi$ is the summation operator, we have $\dbtilde X = ( \Phi ( \eps \mathcal P)^* )^* = \eps \mathcal P \Phi^*$ and the claim is proven because $\tilde \Phi^*=\mathcal P\Phi^*\mathcal P$.

	{\bf Step 4.} Previous steps yield the following expressions for $\tilde R_0$ and $\tilde R_k$. Take the $N$--dimensional space $\tilde{\mathcal U}$ (belonging to $(T-1)$-dimensional space $V$) spanned by the columns of $\mathcal P \eps^*$. Let $\tilde P$ be the orthogonal projector on the space orthogonal to $L_c \tilde{\mathcal U}$, $L_c^2 \tilde{\mathcal U}$, \dots, $L_c^{k-1} \tilde{\mathcal U}$. (Note that $L_c$ can be replaced by $L_V$ in the last definition without changing $\tilde P$). Then the space spanned by $N$ columns of $\tilde R_0^*$ is $\tilde P \tilde{\mathcal U}$. On the other hand, the space spanned by $N$ columns of $\tilde R_k^*$ is $\tilde P L_c^{k-1} \tilde \Phi \tilde{\mathcal U}=\tilde P L_V^{k} (\1_V-L_V)^{-1} \tilde{\mathcal U}$. At this point, we see strong similarities with objects in the statement of Proposition \ref{Prop_gaussian_rotation} with main difference being in the assignment of randomness: $L_c$ is deterministic and $\tilde{\mathcal U}$ is random, but $\tilde L$ is random and $\mathcal U$ is deterministic. Thus, it remains to relocate the random part.
	
	For that we notice that due to the rotational invariance of the Gaussian law (here it is important that we made the covariance matrix $\Lambda$ identical on the first step), the space $\tilde{\mathcal U}$ spanned by the columns of $\mathcal P \eps^*$ has the same law as $\tilde O^* \mathcal U$. The reason is that both laws give uniformly random $N$--dimensional subspace of $(T-1)$--dimensional space $V$.
	
	Since everything was previously expressed through the span  of columns of $\mathcal P \eps^*$, denote $\tilde {\mathcal U}$, we now simply replace those by the columns of $\tilde O^* \mathcal U$.
	Then the space orthogonal to $L_c \tilde{\mathcal U}$, $L_c^2 \tilde{\mathcal U}$, \dots, $L_c^{k-1} \tilde{\mathcal U}$ becomes the space orthogonal to $L_c \tilde O^* \mathcal U $, $L_c^2 \tilde O^* \mathcal U $, \dots, $L_c^{k-1} \tilde O^* \mathcal U $. Equivalently, this is the space orthogonal to $ \tilde O^* \tilde L  \mathcal U $, $ \tilde O^* \tilde L^2 \mathcal U $, \dots, $\tilde O^* \tilde L^{k-1} \mathcal U $. $\tilde P$ is the projector on this space.
	We conclude that the law of canonical correlations \eqref{eq_vark_eig} is the same as the law of non-zero eigenvalues of the product of two projectors: the first one projects on the subspace $\tilde P \tilde O^* \mathcal U$ and the second one projects on the subspace $\tilde P  L_c^k (\1_V-L_c)^{-1} \tilde O^* \mathcal U=\tilde P \tilde O^*  \tilde L^k (\1_V+\tilde L)^{-1} \mathcal U $. Up to a change of basis (by matrix $\tilde O$), which does not change the eigenvalues, we have arrived precisely at the expression from the statement of the proposition.
\end{proof}

The next proposition explains the effect of the replacement $O\leftrightarrow  \tilde L$  on the eigenvalues of the product of projectors in Theorem \ref{Theorem_Var_k} and Proposition \ref{Prop_gaussian_rotation}. We need to introduce some additional notations.

Choose positive integers $k$, $N$, and $\T$, such that $\T\ge (k+1) N$ and an arbitrary $N$--dimensional subspace $\mathcal V$ in $\T$--dimensional space.
Let
$$f^{k,N,\T;\mathcal V}: SO(\T)\to \{0\le x_1\le x_2\le\dots \le x_N\le 1 \}$$ be a map from the group $SO(\T)$ of orthogonal $\T\times \T$ matrices of determinant $1$ to $N$--tuples of reals on $[0,1]$ interval, defined by the following procedure: Take $O\in SO(\T)$.  Let $P$ be the orthogonal projector on the space orthogonal to $O \mathcal V$, $O^2 \mathcal V$,\dots, $O^{k-1} \mathcal V$. Let $P_1$ be the projector on the subspace  $P \mathcal V$ and $P_2$ be the projector on the subspace $P O^{k}(\1_{\T}+O)^{-1} \mathcal V$. Then $f^{k,N,\T;\mathcal V}$ maps $O$ to $N$ largest eigenvalues of  $P_1P_2P_1$.

We also need three norms:
\begin{enumerate}
	\item $\|v\|_2$ is the $L_2$ norm of a vector $v=(v_1,v_2,\dots,v_N)$, defined as $\|v\|_2=\sqrt{\sum_{i=1}^N v_i^2}$.
	\item $\|v\|_{\infty}$ is the supremum norm of a vector $v=(v_1,\dots,v_N)$, defined as $\|v\|_{\infty}=\max_{i} |v_i|$.
	\item $\|A\|_2$ is the spectral norm of a matrix $A$, defined as the square root of the largest eigenvalue of $A A^*$. Equivalently, $\|A\|_2=\max_{v} \frac{\|Av\|_2}{\|v\|_2}$.
	
\end{enumerate}

\begin{proposition} \label{Proposition_map_continuity} Suppose that $k$ is fixed, while $N$ is growing and $\T$ depends on $N$ in such a way that $\frac{\T}{N}\in [k+1 + C_1, C_2]$ for some $C_1,C_2>0$. Let $O_1$ and $O_2$ be two $\T\times \T$ random matrices, such that:
	\begin{itemize}
		\item $O_1$ is a uniformly random $\T\times \T$ orthogonal matrix with determinant $1$.
		\item The eigenvalues of $O_2$ are almost surely different from $-1$.
		\item For each $\eps>0$ we have
		\begin{equation}
		\lim_{N\to\infty} {\rm Prob} \left( \|O_1-O_2\|_2< \frac{1}{N^{1-\eps}}  \right)=1.
		\end{equation}
	\end{itemize}
	Then for each $\eps>0$ we have
	\begin{equation}
	\lim_{N\to\infty} {\rm Prob} \left( \|f^{k,N,\T;\mathcal V}(O_1)-f^{k,N,\T;\mathcal V}(O_2)\|_\infty< \frac{1}{N^{1-\eps}}  \right)=1.
	\end{equation}
\end{proposition}
Proposition \ref{Proposition_map_continuity} claims a continuity of map $f^{k,N,\T;\mathcal V}$. The proof needs care because of the inversions in the definition of the map $f^{k,N,\T;\mathcal V}$.

\begin{remark} \label{Remark_map_continuity_complex}
	Proposition \ref{Proposition_map_continuity} has a version for complex numbers, in which all orthogonal matrices are replaced by unitary matrices. The proof of the complex version is the same.
\end{remark}

The proof of Proposition \ref{Proposition_map_continuity} relies on three lemmas which we prove later in this section. For these lemmas we  write matrices $O_1$ and $O_2$ of Proposition \ref{Proposition_map_continuity} in the block forms according to the splitting $\T=(\T-N)+N$:
\begin{equation}
\label{eq_block_notations}
O_1=\begin{pmatrix} A_1 & B_1\\ C_1 & D_1\end{pmatrix},\qquad  O_2=\begin{pmatrix} A_2 & B_2\\ C_2 & D_2\end{pmatrix}.
\end{equation}

\begin{lemma}
	\label{Lemma_no_eigenvalues}
	Let $O$ be a $\T\times T$ orthogonal matrix of determinant $1$ written in the block form  $O=\begin{pmatrix} A & B\\ C & D\end{pmatrix}$ according to the splitting $\T=(\T-N)+N$. If all eigenvalues of $O$ are different from $-1$, then so are the eigenvalues of $D$ and of $A-B(\1_N+D)^{-1}C$.
\end{lemma}

\begin{lemma}
	\label{Lemma_match_projections}
	Under the assumptions of Proposition \ref{Proposition_map_continuity} we have
	\begin{equation}
	\lim_{N\to\infty} {\rm Prob} \left( \| (A_1-B_1(\1_{N}+D_1)^{-1} C_1)-(A_2-B_2(\1_{N}+D_2)^{-1} C_2)\|_2< \frac{1}{N^{1-\eps}}  \right)=1.
	\end{equation}
\end{lemma}

\begin{lemma}
	\label{Lemma_match_spaces}
	Under the assumptions of Proposition \ref{Proposition_map_continuity}, let $\tilde{\mathcal  V}_0$ be the $N$--dimensional subspace of $(\T-N)$--dimensional space spanned by the last $N$ coordinate vectors.
	There exists an $(\T-N)\times (\T-N)$ orthogonal matrix $U_1$, depending only on $B_1(\1_N+D_1)^{-1}$ and an $(\T-N)\times (\T-N)$ orthogonal matrix, $U_2$, depending both on $B_1(\1_N+D_1)^{-1}$ and on $B_2(\1_N+D_2)^{-1}$, such that $U_1\tilde{\mathcal  V}_0=\langle B_1\rangle$,  $U_2\tilde{\mathcal  V}_0=\langle B_2\rangle$  and
	\begin{equation}
	\lim_{N\to\infty} {\rm Prob} \left( \|U_1-U_2\|_2< \frac{1}{N^{1-\eps}}  \right)=1.
	\end{equation}
\end{lemma}

\begin{proof}[Proof of Proposition \ref{Proposition_map_continuity}]
	The proof is induction in $k$ with the base case $k=1$ proven in \cite{BG}, see the continuity of $M(Z)$ in the proof of Proposition 13 there. For the induction step we recycle the ideas in the proof of Theorem \ref{Theorem_Var_k}.
	
	First, recall a property of function $f$, which we established in Steps 1 and 2 of Theorem \ref{Theorem_Var_k}: if $U$ is a $\T\times\T$ orthogonal matrix, then
	\begin{equation}
	\label{eq_f_conjugation}
	f^{k,N,\T;U \mathcal V}(O)= f^{k,N,\T;\mathcal V}(U^*OU).
	\end{equation}
	Note that conjugations (by the same orthogonal matrix for $O_1$ and $O_2$) leave the three conditions of Proposition \ref{Proposition_map_continuity} unchanged, hence, the \eqref{eq_f_conjugation} implies that statement of proposition remains the same for any choice $\mathcal V$.
	
	Second, we make the replacements of Steps 2 and 3 of Theorem \ref{Theorem_Var_k} individually for $f^{k,N,\T;\mathcal V}(O_1)$ and $f^{k,N,\T;\mathcal V}(O_2)$. The replacement $\mathcal V\mapsto O^{k-1} \mathcal V$ does not change the eigenvalues of $P_1 P_2 P_1$, while inversion of $O_1$ and $O_2$ keeps the conditions of Proposition \ref{Proposition_map_continuity} unchanged. Summing up, we replace the map $O\mapsto f^{k,N,\T; \mathcal V}(O)$ in Proposition \ref{Proposition_map_continuity} by a new map $O\mapsto \tilde f^{k,N,\T; \mathcal V_0}(O)$ defined through: Let $\mathcal V_0$ be the subspace spanned by the last $N$ coordinate vectors in $\T$--dimensional space. Let $P$ be the orthogonal projector on the space orthogonal to $\mathcal V, O \mathcal V$, $O^2 \mathcal V$,\dots, $O^{k-2} \mathcal V$. Let $P_1$ be the projector on the
	subspace  $P O^{k-1} \mathcal V$ and $P_2$ be the projector on the subspace $P (\1_{\T}+O) \mathcal V$ (or, equivalently, on $P O (\1_\T+O)\mathcal V$). Then $\tilde f^{k,N,\T;\mathcal V_0}(O)$ is $N$ largest eigenvalues of  $P_1P_2P_1$.

	Using the block notations \eqref{eq_block_notations}, steps 4-7 in the proof of Theorem \ref{Theorem_Var_k} imply the following almost sure identities:
	\begin{equation}
	\label{eq_x1}
	\tilde f^{k,N,\T;\mathcal V_0}(O_1)=\tilde f^{k-1,N,\T-N;\langle B_1\rangle }(A_1-B_1(\1_{N}+D_1)^{-1} C_1),
	\end{equation}
	\begin{equation}
	\label{eq_x2}
	\tilde f^{k,N,\T;\mathcal V_0}(O_2)=\tilde f^{k-1,N,\T-N;\langle B_2\rangle }(A_2-B_2(\1_{N}+D_2)^{-1} C_2).
	\end{equation}
	We would like to check that the right-hand sides of \eqref{eq_x1} and \eqref{eq_x2} are close by using the induction assumption.
	
	Using \eqref{eq_f_conjugation} and Lemma \ref{Lemma_match_spaces}, we rewrite the right-hand sides of \eqref{eq_x1} and \eqref{eq_x2} as:
	\begin{equation}
	\label{eq_x3}
	\tilde f^{k-1,N,\T-N;\tilde{\mathcal  V}_0}\bigl( U_1^*(A_1-B_1(\1_{N}+D_1)^{-1} C_1) U_1\bigr); \quad \quad \tilde f^{k-1,N,\T-N;\tilde{\mathcal  V}_0}\bigl( U_2^*(A_2-B_2(\1_{N}+D_2)^{-1} C_2) U_2\bigr).
	\end{equation}
	Let us check that we can apply the induction assumption to deduce that the expressions of \eqref{eq_x3} are close to each other:
	\begin{itemize}
		\item By Lemma \ref{Lemma_Haar_project},  $A_1-B_1(\1_{N}+D_1)^{-1} C_1$ is a uniformly random $(\T-N)\times(\T-N)$ orthogonal matrix. By Lemma \ref{Lemma_match_spaces}, $U_1$ is a function of $B_1(\1_N+D_1)^{-1}$. Hence, using Lemma \ref{Lemma_Haar_project} again, we conclude that $U_1$ is independent from $A_1-B_1(\1_{N}+D_1)^{-1} C_1$. Therefore,  $U_1^*(A_1-B_1(\1_{N}+D_1)^{-1} C_1) U_1$ is a uniformly random orthogonal matrix, as desired.
		\item By Lemma \ref{Lemma_no_eigenvalues},  $(\1_{N}+D_2)^{-1}$ is well-defined and no eigenvalues of $A_2-B_2(\1_{N}+D_2)^{-1} C_2$ are equal to $-1$. Hence, the eigenvalues of $U_2^*(A_2-B_2(\1_{N}+D_2)^{-1} C_2) U_2$ are almost surely different from $-1$.
		
		\item Combining Lemmas \ref{Lemma_match_projections} and  \ref{Lemma_match_spaces} we conclude that
		\begin{multline*}
		\lim_{N\to\infty} {\rm Prob} \left( \| U_1^*(A_1-B_1(\1_{N}+D_1)^{-1} C_1)U_1-U_2^*(A_2-B_2(\1_{N}+D_2)^{-1} C_2)U_2\|_2< \frac{1}{N^{1-\eps}}  \right)\\=1.
		\end{multline*}
	\end{itemize}
	Hence, using the $(k-1)$ statement, the expressions in \eqref{eq_x3} are close to each other as $N\to\infty$ and, therefore, \eqref{eq_x1} is close to \eqref{eq_x2}.
\end{proof}

\begin{proof}[Proof of Lemma \ref{Lemma_no_eigenvalues}]
	Let us show that $D$ has no eigenvalues $-1$. We argue by contradiction and assume that there exists an $N$--dimensional vector $v$ of length $1$ such that $Dv=-v$. Note that $B^*B+D^*D=\1_N$ by orthogonality of $O$. Hence, using the notation $\langle \cdot,\cdot\rangle$ for the scalar product, we have
	$$
	\langle Bv, Bv \rangle=   \langle B^*Bv, v \rangle=  \langle (\1_N-D^* D)v, v \rangle=\langle v,v\rangle -\langle Dv, Dv\rangle=1-1=0.
	$$
	Therefore, $Bv=0$, which readily implies that the $\T$--dimensional vector ${0\choose v}$ is an eigenvector of $O$ with eigenvalue $-1$. Contradiction.
	
	Next, for the matrix $A-B(\1_{N}+D)^{-1} C$, let us use its representation as a Cayley transform developed in \eqref{eq_tO_tR}:
	$$
	A-B(\1_{N}+D)^{-1} C=\frac{\1_{\T-N}-\R}{\1_{\T-N}+\R},
	$$
	where $\R$ is a $(\T-N)\times(\T-N)$ skew-symmetric matrix. If $v$ was an eigenvector of $A-B(\1_{N}+D)^{-1} C$ with eigenvalue $-1$, then we would have
	$$
	(\1_{\T-N}-\R)v=-(\1_{\T-N}+\R)v,
	$$
	which is impossible for non-zero $v$.
\end{proof}

\begin{proof}[Proof of Lemma \ref{Lemma_match_projections}]
	Note that whenever $X$ is a submatrix of $Y$, we have $\|X\|_2\le \|Y\|_2$. Hence, the spectral norms of the differences $A_1-A_2$, $B_1-B_2$, $C_1-C_2$, $D_1-D_2$ are all small with probability tending to $1$ as $N\to\infty$. Addition, multiplication, and inversion of matrices are all Lipschitz operations as long as  factors are bounded for the multiplication and singular values are bounded away from $0$ for the inversion.  Therefore, it remains to show that the norms of the factors $B_1$, $(\1_N+D_1)^{-1}$, and $C_1$ are uniformly bounded (since $B_1$, $(\1_N+D_1)^{-1}$, and $C_1$ are close to $B_2$, $(\1_N+D_2)^{-1}$, and $C_2$, respectively, the norms of the latter are then going to be bounded as well). For $B_1$ and $C_1$ the bound on the norm is straightforward, as they are submatrices of $O_1$, whose norm is $1$. Hence, $\|B_1\|_2\le 1$ and $\|C_1\|_2\le 1$.
	
	In order to deal with $(\1_N+D_1)^{-1}$ we rely on the fact that the distribution of the symmetric $N\times N$ matrix $Y=D_1^* D_1$ is explicit. It has density (see, e.g., \cite[(3.113) and the formula immediately after]{forrest}) proportional to:
	\begin{equation}
	\label{eq_corner_distribution}
	\det Y^{-1/2} \det(\1_N-Y)^{\frac{\T}{2}-N-1/2}\, d Y, \qquad 0< Y< \1_N.
	\end{equation}
	This is a particular case of the Jacobi ensemble of Definition \ref{Definition_Jacobi} and we can use the large $N$ asymptotic of the latter recorded in Proposition \ref{Theorem_Jacobi_as}. Therefore, there exists a constant $0<c<1$, such that all the eigenvalues of $Y$ are smaller than $c$ with probability tending to $1$ as $N\to\infty$. Hence, by the triangular inequality
	$$
	\min_{\|v\|_2=1} \|(\1_N+D_1)v\|_2\ge 1-\sqrt{c},
	$$
	with probability tending to $1$ as $N\to\infty$. We conclude that
	$$
	\lim_{N\to\infty} {\rm Prob}\left(
	\|(\1_N+D_1)^{-1}\|<\frac{1}{1-\sqrt{c}}\right)=1. \qedhere.
	$$
\end{proof}

\begin{proof}[Proof of Lemma \ref{Lemma_match_spaces}]
	We will be proving a slightly different statement, in which $\tilde {\mathcal V}_0$ is the span of the first (rather than last) coordinate vectors. The desired statement of the theorem is then obtained by replacing $
	U_1\mapsto U_1 \cdot \mathfrak S$, and $U_2\mapsto U_2 \cdot \mathfrak S$,
	where $\mathfrak S$ is the (orthogonal matrix) which swaps $i$th and $(\T-N+1-i)$th basis vectors for $i=1,2,\dots,\T-N$.

	We know that the matrices $B_1$ and $B_2$ are close to each other and our aim is to show that the orthonormal bases of $\langle B_1\rangle =\langle B_1 (\1_N+D_1)^{-1}\rangle$ and its orthogonal complement, and $\langle B_2\rangle =\langle B_2 (\1_N+D_2)^{-1}\rangle$ and its orthogonal complement can be chosen to also be close to each other. For that we need to produce some formulas for these bases, which is what we do in the rest of the proof. The delicacy of this argument stems from the fact that given a space, in general, there might be no continuous way to produce an orthogonal matrix, such that the space is spanned by its first columns. (For instance, by the hairy ball theorem one can not continuously complement a unit vector in $3$--dimensional space to an orthonormal basis.) Hence, we need to be more careful.
	
	We start by replacing $B_1$ with
	$$X_1:= B_1 (\1_N+D_1)^{-1} \bigl((\1_N+D_1^*)^{-1} B_1^* B_1 (\1_N+D_1)^{-1}\bigr)^{-1/2}$$
	and replacing $B_2$ with
	$$
	X_2:= B_2 (\1_N+D_2)^{-1} \bigl((\1_N+D_2^*)^{-1} B_2^* B_2 (\1_N+D_1)^{-1}\bigr)^{-1/2}.
	$$
	Clearly, $\langle B_1\rangle=\langle X_1\rangle$ and $\langle B_2\rangle=\langle X_2\rangle$. The advantage of $X_1$ and $X_2$ is that their columns are orthonormal. Indeed,
	\begin{multline*}
	X_1^* X_1=\bigl((\1_N+D_1^*)^{-1} B_1^* B_1 (\1_N+D_1)^{-1}\bigr)^{-1/2}(\1_N+D_1^*)^{-1}  B_1^* B_1 (\1_N+D_1)^{-1}\\ \times \bigl((\1_N+D_1^*)^{-1} B_1^* B_1 (\1_N+D_1)^{-1}\bigr)^{-1/2}=\1_N
	\end{multline*}
	and similarly for $X_2$.
	
	{\bf Claim.} $X_1$ and $X_2$ are asymptotically close to each other:
	\begin{equation}
	\label{eq_X_close}
	\lim_{N\to\infty} {\rm Prob} \left( \|X_1-X_2\|_2< \frac{1}{N^{1-\eps}}  \right)=1.
	\end{equation}
	Note that $X_1$ and $X_2$ are built out of $O_1$ and $O_2$ with operations of addition, multiplication,  inversion, and square root. The first one is Lipschitz in spectral norm, the second one is Lipschitz as long as the factors are uniformly bounded, and for the last two we additionally need the singular values of the factors to be uniformly bounded away from $0$ uniformly\footnote{For the square root operation on positive-definite matrices $x\mapsto \sqrt{x}$ we can first rescale $x$ so that its spectrum belongs to $[c_0,1]$ segment for some $c_0>0$ and then use Taylor series expansion of the square root:
		$
		\sqrt{x}=\sqrt{1+ (x-1)}= 1 + \frac{x-1}{2}- \frac{1}{4} (x-1)^2+\dots
		$
		to deduce the Lipschitz property.
	}. We already explained in the proof of Lemma \ref{Lemma_match_projections} that $B_1$ has spectral norm at most $1$ and that $(\1_N+D_1)$ (and hence also its inverse and its transpose) has singular values bounded away from $0$ and $\infty$. Hence, it remains only to deal with $B_1^* B_1$ in the definition of $X_1$. Since $B_1$ is a $(\T-N)\times N$ submatrix of uniformly random $\T\times \T$ matrix, the law of $\Lambda=B_1^* B_1$ is explicit.
	It has density (see, e.g., \cite[(3.113) and the formula immediately after]{forrest}) proportional to:
	\begin{equation}
	\det \Lambda^{\frac{\T}{2}-N-1/2} \det(\1_N-\Lambda)^{-1/2}\, d \Lambda, \qquad 0< \Lambda< \1_N.
	\end{equation}
	This is a particular case of the Jacobi ensemble of Definition \ref{Definition_Jacobi} and we can use the large $N$ asymptotic of the latter recorded in Proposition \ref{Theorem_Jacobi_as}, which implies that the eigenvalues of $\Lambda$ are bounded away from $0$ as $N\to\infty$. The claim is proven.
	
	\smallskip

	Next, we produce the desired orthogonal matrix $U_1$ by the Gramm-Schmidt orthogonalization procedure: letting $e_k$ be the $k$--th coordinate vector in $(\T-N)$--dimensional space, and $X_1^k$ be the $k$--th column of $X_1$, we start from $(\T-N)$ vectors
	$$X_1^1,X_1^2,\dots, X_1^N,\, e_{N+1},e_{N+2},\dots, e_{\T-N}$$
	and orthogonalize them. This is a valid procedure, since the Gramm matrix of the above vectors is almost surely non-degenerate (this is equivalent to the non-degeneracy of the top $N\times N$ corner of $X_1$, which is true due to absolute continuity of the distribution of this corner with respect to the Lebesgue measure on $N\times N$ matrices that can be deduced from Remark \ref{Remark_off_diag_corner}).
	
	We set the columns of $U_1$ to be the vectors from the orthogonalization procedure. Since the vectors are orthonormal, $U_1$ is orthogonal. Note that since the columns of $X_1$ are orthonormal, the first $N$ steps of the orthogonalization procedure are trivial and the first $N$ columns of $U_1$ are $X_1^1,X_1^2,\dots, X_1^N$. In particular, these $N$ columns span $\langle B_1\rangle$, as desired.
	
	We proceed to the construction of $U_2$. It is tempting to do exactly the same procedure (with all indices $1$ replaced by indices $2$), but that is not going to work: the problem is that while the top $N\times N$ corner of $X_1$ was almost surely non-degenerate, but it can have singular values arbitrary close to $0$. Eventually, this leads to unstability of the orthogonalization procedure and, hence, there is no way to guarantee that the results of orthogonalization for $X_1$ and $X_2$ are close to each other.
	
	Therefore, we proceed in a different way. Set $F:= U_1^{-1} X_2$. Because the first $N$ columns of $U_1$ are $X_1$, we have
	$$
	U_1^{-1} X_1={\1_{N}\choose 0_{(\T-2N)\times N}},
	$$
	where $0_{(\T-2N)\times N}$ stays for the $(\T-2N)\times N$ filled with $0$ matrix elements. Hence, since $X_1$ and $X_2$ were close, we have
	\begin{equation}
	\label{eq_close_to_identity}
	\lim_{N\to\infty} {\rm Prob} \left( \left\|F-{ \1_{N}\choose 0_{(\T-2N)\times N}}\right\|_2< \frac{1}{N^{1-\eps}}  \right)=1.
	\end{equation}
	Let $F^k$, $k=1,\dots,N$, denote the columns of $F$ and consider $\T-N$ vectors
	$$
	F^1,F^2,\dots,F^N,\, e_{N+1},e_{N+2},\dots,e_{\T-N}.
	$$
	We are going to orthogonalize these vectors. The advantage over the procedure we used for $X_1$ is that now the top $N\times N$ submatrix of $F$ is close to identity, which is going to make the orthogonalization procedure well-behaved. In order to make the orthogonalization procedure explicit, we are going to use a block version of the Cholesky decomposition.
	
	For that set $M=\T-2N$ and write $F$ in the block form according to the splitting $\T-N= N+M$:
	$$
	F={Y_2\choose Z_2}.
	$$
	Let $W_2$ denote the $(\T-N)\times(\T-N)$ matrix written in the $N+M$ block form as
	$$
	W_2=\begin{pmatrix} Y_2&0 \\  Z_2& \1_M\end{pmatrix}.
	$$
	We would like to perform orthogonalization of the columns of $W_2$. For that we first compute
	\begin{equation}
	\label{eq_WW}
	W_2^*W_2=\begin{pmatrix} Y_2^* Y_2+Z_2^*Z_2 & Z_2^*\\ Z_2& \1_M \end{pmatrix}.
	\end{equation}
	We further would like to represent $W_2^* W_2$ as
	\begin{equation}
	\label{eq_Cholesky}
	W_2^* W_2= \begin{pmatrix} \1_N& 0 \\ Q_2 & \1_M\end{pmatrix} \begin{pmatrix} G_2 & 0\\ 0 & H_2\end{pmatrix}  \begin{pmatrix} \1_N& Q_2^* \\ 0 & \1_M\end{pmatrix}=\begin{pmatrix} G_2& G_2Q_2^*\\ Q_2 G_2 & Q_2 G_2 Q_2^* + H_2\end{pmatrix}.
	\end{equation}
	Comparing with \eqref{eq_WW} we conclude that
	\begin{equation}
	\label{eq_Ch_formulas}
	G_2= Y_2^* Y_2+Z_2^*Z_2=F^* F, \quad  Q_2= Z_2 (F^* F)^{-1}, \quad H_2= \1_M- Z_2 (F^* F)^{-1} Z_2^*.
	\end{equation}
	Since the spectral norm of a submartix is at most the spectral norm of the matrix, \eqref{eq_close_to_identity} implies that
	\begin{equation}
	\label{eq_x7}
	\lim_{N\to\infty} {\rm Prob} \left( \left\|Y_2-\1_{N}\right\|_2< \frac{1}{N^{1-\eps}}  \right)=1, \qquad \lim_{N\to\infty} {\rm Prob} \left( \left\|Z_2-0_{M\times N}\right\|_2< \frac{1}{N^{1-\eps}}  \right)=1.
	\end{equation}
	Therefore, $W_2$ is close to $\1_{M+N}$, $G_2$ is close to $\1_N$, $Q_2$ is close to $0_{N\times M}$, $H_2$ is close to $\1_{M}$.

	Note that $G_2$ and $H_2$ are positive-definite symmetric matrices, hence, they have well-defined square roots. In addition,
	$$
	\begin{pmatrix} \1_N& Q_2^* \\ 0 & \1_M\end{pmatrix}^{-1}= \begin{pmatrix} \1_N& -Q_2^* \\ 0 & \1_M\end{pmatrix}.
	$$
	The goal of all these manipulations with matrices is to define
	\begin{equation}
	\label{eq_tU_def}
	\tilde U_2=W_2 \cdot \begin{pmatrix} \1_N& -Q_2^* \\ 0 & \1_M\end{pmatrix} \cdot  \begin{pmatrix} G_2^{-1/2} & 0\\ 0 & H_2^{-1/2}\end{pmatrix}.
	\end{equation}
	The two key properties of $\tilde U_2$ are:
	\begin{itemize}
		\item The span of the first $N$ columns of $\tilde U_2$ coincides with $\langle F\rangle$.
		\item $\tilde U_2$ is orthogonal. Indeed, using \eqref{eq_Cholesky} we have
		$$
		\tilde U_2 \tilde U_2^*= W_2  \begin{pmatrix} \1_N& -Q_2^* \\ 0 & \1_M\end{pmatrix} \begin{pmatrix} G_2^{-1} & 0\\ 0 & H_2^{-1}\end{pmatrix} \begin{pmatrix} \1_N& 0\\ -Q_2 & \1_M\end{pmatrix} W_2^* = W_2 ( W_2^* W_2)^{-1} W_2^*=\1_{\T-N}.
		$$
	\end{itemize}
	Hence, we can finally set
	$$
	U_2:=  U_1 \cdot  \tilde U_2,
	$$
	We have
	$$
	U_2 \tilde {\mathcal V}_0=\langle U_1 F \rangle=\langle X_2\rangle=\langle B_2\rangle
	,$$
	as desired. It remains to show that the matrix $\tilde U_2$ is very close to identity, as this would imply that $U_2$ is close to $U_1$. For that we consider each factor in \eqref{eq_tU_def} and see that they are close to identical matrices by \eqref{eq_x7}. Hence,
	$$
	\lim_{N\to\infty} {\rm Prob} \left( \left\|\tilde U_2-\1_{N+M}\right\|_2< \frac{1}{N^{1-\eps}}  \right)=1,
	$$
	as desired.
\end{proof}

\begin{proof}[Proof of Theorem \ref{Theorem_vark_approximation}]
	We start by explicitly constructing the desired coupling. For the Jacobi ensemble we use the realization of Theorem \ref{Theorem_Var_k} and for the matrix of the Johansen test we use the realization of Proposition \ref{Prop_gaussian_rotation}. We set $\T=T-1$ to match the notations and it remains to couple $O$ of Theorem \ref{Theorem_Var_k} with $\tilde L=-\tilde O L_V\tilde O^*$ of Proposition \ref{Prop_gaussian_rotation}.

	The eigenvalues of $L_V$ are all roots of unity of order $T$ different from $1$. In the complex case $\beta=2$ we can diagonalize $L_V$ to turn it into $\T\times \T = (T-1)\times (T-1)$ diagonal matrix with the roots of unity on the diagonal. In the real case $\beta=1$, the matrix $L_V$ should be  block-diagonalized (with blocks of size $2$ and one additional block of size $1$ corresponding to eigenvalue $-1$ if $\T$ is even): the pair of complex conjugate roots of unity $\omega $ and $\bar \omega$ gives rise to the $2\times 2$ matrix of rotation by the angle $|\arg(\omega)|$. Let us denote by $D$ the resulting (block) diagonal matrix multiplied by $-1$. In order to avoid ambiguity about the order of eigenvalues, we assume that the blocks correspond to the increasing order of $|\arg(-\omega)|$, i.e., the top-left $2\times 2$ corner of $D$ corresponds to the pair of the closest to $1$ eigenvalues of $D$.
	
	The eigenvalues of $O$ also lie on the unit circle and if $\beta=1$, then they come in complex-conjugate pairs. Hence, $O$ can be similarly block-diagonalized (we do not need to multiply by $-1$ this time) and we denote through $D^{\text{rand}}$ the result. The distinction with $L_V$ is that the eigenvalues are \emph{random} and so is $D^{\text{rand}}$. The law of the eigenvalues of $O$ is explicitly known in the random-matrix literature. Both for $\beta=1$ and $\beta=2$ they form a determinantal point process on the unit circle with explicit kernel. The repulsion between the eigenvalues leads to them being very close to evenly spaced as $T\to\infty$. We summarize this property in the following statement (which is a manifestation of a more general rigidity of eigenvalues, see, e.g., \citet{ErdosYau}), whose proof can be found in \citet[Lemma 10, $m=1$, $u=T^{\delta}$ case, and Section 5]{MeckesMeckes}.
	
	{\bf Claim.} There exist constants $c_1(\beta),c_2(\beta)>0$, such that for $\beta=1,2$, every $\delta>0$, there exists $\T_0(\delta)$ and for every  $\T>\T_0(\delta)$ we have \footnote{All the constants can be made explicit, following \citet{MeckesMeckes}.}
	\begin{equation}
	\label{eq_MM_bound}
	\mathrm{Prob}\left(\max_{1\le i,j<\T} \bigl| D-D^{\text{rand}}\bigr|_{ij} > \frac{1}{\T^{1-\delta}}\right) < c_1(\beta) \cdot \T \cdot \exp\left( - c_2(\beta) \frac{\T^{2\delta}}{\log \T} \right).
	\end{equation}
	
	\medskip
	
	We remark that since $D$ and $D^{\text{rand}}$ are block-diagonal, the bound on the maximum matrix element of their difference is equivalent to a similar bound for any other norm, e.g., for the spectral norm, which we used in Proposition \ref{Proposition_map_continuity}.
	
	We now choose another $\T\times \T$ uniformly-random orthogonal (or unitary if $\beta=2$) matrix $O_2$ (independent from the rest), replace $-\tilde O L_V \tilde O^*$ with $O_2 D O_2^*$ and replace $O$ with $O_2 D^{\text{rand}} O_2^*$. The invariance of the uniform measure on the orthogonal group $SO(N)$ (or on the unitary group $U(N)$ if $\beta=2$) with respect to right/left multiplications, implies the distributional identities:
	$$
	-\tilde O L_V \tilde O^* \stackrel{d}{=} O_2 D O_2^*,\qquad O\stackrel{d}{=} O_2 D^{\text{rand}} O_2^*.
	$$
	The right-hand sides of the identities provide the desired coupling and \eqref{eq_MM_bound} implies that these two random matrices are close to each other as $\T\to\infty$.
	
	It now remains to apply Proposition \ref{Proposition_map_continuity} (see also Remark \ref{Remark_map_continuity_complex}) with the first matrix being $O_2 D^{\text{rand}} O_2^*$ and the second matrix being $O_2 D O_2^*$.
\end{proof}


\subsection{Small rank perturbations}

\label{Section_small_rank}

In this section we prove Theorem \ref{Theorem_empirical} by combining Theorem \ref{Theorem_vark_approximation} with Proposition \ref{Theorem_Jacobi_as} and general statements about small rank perturbations. The key step of the proof is the following observation:

\begin{theorem} Let $R_i$, $i=0,k$ be as in Section \ref{Section_setting} for $X_t$ solving Eq.~\eqref{var_k} and let  $\tilde R_{i}$, $i=0,k$ be as in Section \ref{Section_modified_test} under $\widehat H_0$ for  $X_t$ solving Eq.~\eqref{eq_H0}. Suppose that $X_0$ and the noises $\eps_t$ used in the constructions of $R_{i}$ and $\tilde R_{i}$ are the same. Introduce $T\times T$ projection matrices:
	\begin{equation}
	\label{eq_x9}
	P_0= R_0^* (R_0 R_0^*)^{-1} R_0, \quad   P_k= R_k^* (R_k R_k^*)^{-1} R_k, \quad  \tilde P_0=\tilde  R_0^* (\tilde R_0 \tilde R_0^*)^{-1} \tilde R_0, \quad   \tilde P_k= \tilde R_k^* (\tilde R_k \tilde R_k^*)^{-1} \tilde R_k.
	\end{equation}
	Then under the assumptions \eqref{eq_limit_regime}, \eqref{eq_rank_restriction} of Theorem \ref{Theorem_empirical} we have
	\begin{equation}
	\label{eq_rank_diff}
	\lim_{N\to\infty} \frac{1}{N} \rank\left( P_0 P_k P_0-\tilde P_0 \tilde P_k \tilde P_0  \right)=0.
	\end{equation}
\end{theorem}
\begin{proof} Throughout the proof we assume that the matrices $R_i R_i^*$ and $\tilde R_i \tilde R_i^*$ are invertible. In principle, invertibility might fail for some $N$: in such situation we can still use Moore–Penrose inverse in order for the statements to make sense, and we are not going to detail this.
	
	In the following argument we use various properties of ranks:
	\begin{itemize}
		\item If a matrix $A$ differs from a matrix $B$ only in $\mathfrak r$ columns, then $\rank(A-B)\le \mathfrak r$;
		\item $\rank(C A - CB)\le \rank(A-B)$;
		\item $\rank(A+B)\le \rank(A) + \rank(B)$;
		\item If matrices $A$ and $B$ are invertible, then $\rank(A^{-1}-B^{-1})=\rank (A-B)$.
	\end{itemize}

	We refer to the time series defined by \eqref{var_k} as $X_t$ and to the time series defined by \eqref{eq_H0} as $\mathcal X_t$. We form two $N\times T$ matrix $X$ and $\mathcal X$ with columns $X_t$ and $\mathcal  X_t$, $t=1,\dots,T$, respectively.  Our first task is to show that $\tfrac{1}{N}\rank (X-\mathcal X)\to 0$ as $N\to\infty$.
	
	For that we subtract \eqref{eq_H0} from \eqref{var_k} to get:
	\begin{equation}
	\label{eq_x4}
	\Delta X_t - \Delta \mathcal X_t = \Pi X_{t-k}+\sum\limits_{i=1}^{k-1}\Gamma_i\Delta X_{t-i}+\Phi D_t-\mu, \quad t=1,2,\dots,T.
	\end{equation}
	In the matrix form, \eqref{eq_x4} represents the $N\times T$ matrix $\Delta(X-\mathcal X)$ with columns  $\Delta X_t - \Delta \mathcal X_t$ as a sum of $k+2$ low rank matrices, with the total rank (coming from the right-hand side of \eqref{eq_x4}) at most
	\begin{equation}
	\label{eq_x5}
	\rank(\Pi)+\sum_{i=1}^{k-1} \rank(\Gamma_i)+d_D+1.
	\end{equation}
	The matrix $X-\mathcal X$ is obtained from $\Delta(X-\mathcal X)$ by multiplication by the summation matrix
	$$
	\Phi=\begin{pmatrix} 1&0&0&\dots&0\\ 1& 1 &0&\dots &0\\ 1& 1& 1 &\dots & 0 \\ && \ddots \\ 1&1 &\dots & 1 &1.\end{pmatrix}
	$$
	Hence, the rank of $X-\mathcal X$ is at most \eqref{eq_x5} and $\tfrac{1}{N}\rank (X-\mathcal X)\to 0$  by assumption \eqref{eq_rank_restriction} of Theorem \ref{Theorem_empirical}.

	Next, we should take into account that the procedures for constructing $R_0$, $R_k$ from $X$ and $\tilde R_0$, $\tilde R_k$ from $\mathcal X$ are slightly different.  Namely, the latter involves cyclic shifts of indices, rather than usual shifts, involves regressing over only constants, rather than $d_D$ deterministic terms $D_t$, and finally involves detrending \eqref{eq_detrending}. However, cyclic shifts only affect the first $k-1$ indices $t$ and, hence, lead to bounded difference in ranks, and similarly for detrending. Regressing on $d_D$ terms leads to another $O(d_D)$ difference in ranks, which is negligible after division by $N$ in the limit $N\to\infty$ by \eqref{eq_rank_restriction}. The conclusion is that $R_0$, $R_k$ from one side and $\tilde R_0$, $\tilde R_k$ on the the other side are constructed from two finite sets of matrices, which differ by small rank perturbations, by finitely many of operations of addition, multiplication, and inversion. Each of these operations preserves the smallness of the rank of perturbations and, hence,
	\begin{equation}
	\label{eq_x6}
	\lim_{N\to\infty} \frac{1}{N} \rank(R_0-\tilde R_0)=0,\quad \lim_{N\to\infty} \frac{1}{N} \rank(R_k-\tilde R_k)=0.
	\end{equation}
	Since $P_0 P_k P_0$ and $\tilde P_0 \tilde P_k \tilde P_0$ are obtained from $R_0$, $R_k$ and $\tilde R_0$, $\tilde R_k$, respectively, by the same algebraic operations, \eqref{eq_rank_diff} follows from \eqref{eq_x6}.
\end{proof}

\begin{proof}[Proof of Theorem \ref{Theorem_empirical}]  Let $\lambda_1\ge \lambda_2\ge \dots\ge \lambda_N$ denote the eigenvalues of $\mathcal C =S_{kk}^{-1} S_{k0} S_{00}^{-1} S_{0k}$ and let $\tilde \lambda_1\ge \tilde \lambda_2\ge \dots\ge \tilde \lambda_N$ denote the eigenvalues of $\tilde{\mathcal C}=\tilde S_{kk}^{-1} \tilde S_{k0} \tilde S_{00}^{-1} \tilde S_{0k}$.
	Combining Theorem \ref{Theorem_vark_approximation} with Proposition \ref{Theorem_Jacobi_as} we conclude that the empirical measure of $\tilde \lambda_i$ converges:
	\begin{equation}
	\label{eq_empirical_conv}
	\lim_{N\to\infty} \frac{1}{N}\sum_{i=1}^N \delta_{\tilde \lambda_i} = \mu_{2,\tau-k}.
	\end{equation}
	We would like to show that $\tilde \lambda_i$ can be replaced by $\lambda_i$ in \eqref{eq_empirical_conv}. Note that although the spectra of matrices $\mathcal C$ and $\tilde{\mathcal C}$ are real, but these matrices are not symmetric. Similarly to \eqref{eq_rank_diff}, $\frac{1}{N}\rank(\mathcal C-\tilde{\mathcal C})\to 0$, however, in general, for non-symmetric matrices even rank $1$ perturbations can lead to significant changes in the spectrum. Hence, we need to be more careful and symmetrize $\mathcal C$ and $\tilde {\mathcal C}$ by using projectors as in \eqref{eq_rank_diff}.
	
	Note that for any two $K\times M$ matrices $A$ and $B$, the non-zero eigenvalues of $A B^*$ and of $B^* A$ coincide. Recalling that $S_{ij}=R_i R_j^*$ and using the notations \eqref{eq_x9}, we conclude that the eigenvalues of $\mathcal C$ are the same as $N$ largest eigenvalues of $P_k P_0$. Since $P_0^2=P_0$, they are also the same as $N$ largest eigenvalues of $P_k P_0 P_0$ and the same as those of $P_0 P_k P_0$. Similarly, the eigenvalues of $\tilde{\mathcal C}$ are the same as $N$ largest eigenvalues of $\tilde P_0 \tilde P_k \tilde P_0$.
	
	Denote $\mathfrak r=\rank(P_0 P_k P_0-\tilde P_0 \tilde P_k \tilde P_0)$. All the involved matrices are symmetric and we can use classical inequalities between eigenvalues of a Hermitian matrix $A$ and Hermitian matrix $A+B$, where $B$ has rank $\mathfrak r$, see, e.g., \citet[Corollary 4.3.5]{Horn_Johnson}. In our situation the inequalities read
	\begin{equation}
	\label{eq_rank_inequality}
	\lambda_{m-\mathfrak r}\ge \tilde \lambda_m\ge \lambda_{m+\mathfrak r}, \quad 1\le m-\mathfrak r\le m+\mathfrak r\le N.
	\end{equation}
	Therefore, for any points $0<a<b<1$,
	$$
	\left|\#\{1\le i \le N \mid \lambda_i\in [a,b]\}-\#\{1\le i \le N \mid \tilde \lambda_i\in [a,b]\}\right|\le 2\mathfrak r.
	$$
	Hence, \eqref{eq_empirical_conv} implies
	\begin{equation}
	\label{eq_empirical_conv_2}
	\lim_{N\to\infty} \frac{1}{N}\sum_{i=1}^N \delta_{\lambda_i}= \mu_{2,\tau-k}.\qedhere
	\end{equation}
	
\end{proof}

\section{Appendix 2. Discussion of asymptotics under $H_0$ and $H_1$}
\label{Section_Appendix_2}

The goal of this section is to discuss the asymptotics of the test statistic $LR_{N,T}(r)$ of \eqref{LR_NT} under various data generating processes \eqref{var_k_restr} generalizing $\widehat H_0$ of \eqref{eq_strong_vark} and Theorem \ref{Theorem_J_stat}.

\subsection{Beyond $\widehat H_0$}\label{appendix_H0}
We start by working under a slightly more restrictive assumption than \eqref{eq_rank_restriction} of Theorem \ref{Theorem_empirical}. Let $\|A\|_2$ be the spectral norm of a matrix $A$.

\begin{conjecture} \label{Conjecture} Fix some $k\in\mathbb{N},\,C>0$. Suppose that the data generating process is
 \begin{equation}\label{eq_var_k_conjecture}
\Delta X_t=\mu+\sum\limits_{i=1}^{k-1}\Gamma_i\Delta X_{t-i}+\eps_t,\qquad t=1,\ldots,T,\qquad\qquad \text{where}
\end{equation}
\begin{enumerate}
 \item $\eps_t\thicksim\text{i.i.d.}~\mathcal{N}(0,\Lambda)$ and the covariance matrix $\Lambda$ satisfies $\|\Lambda\|_2<C$ and $\|\Lambda^{-1}\|_2<C$;
 \item $\|\Gamma_i\|_2<C$ and $\rank(\Gamma_i)<C$ for all $1\le i \le k-1$;
 \item All roots of the following characteristic equation \eqref{eq_char_eq} satisfy\footnote{This guarantees that $\Delta X_t$ is $I(0)$ process, which is a standard assumption in the cointegration literature.} $|z|>1+C^{-1}$:
 \begin{equation}
 \label{eq_char_eq}
  \det\left(\1_N-\sum_{i=1}^{k-1} \Gamma_i  z^i\right)=0;
 \end{equation}
 \item $\|\Gamma_j \Delta X_{1-i}\|_2\le C$ and $\|\Gamma_j \mu\|\le C$ for all $1\le i,j \le k-1$.
\end{enumerate}
 Then as $T,N\to\infty$ in such a way that $\frac{T}{N}\in[k+1+C^{-1},C]$, the conclusion of Theorem  \ref{Theorem_J_stat} continues to hold with the same $c_1(N,T)$ and $c_2(N,T)$:
 	\begin{equation}
	\label{eq_statistic_limit_repeat}
	 \frac{\sum_{i=1}^{r} \ln(1-\tilde{\lambda}_i)- r \cdot c_1(N,T)}{ N^{-2/3}  c_2(N,T)}  \, \xrightarrow[T,N\to\infty]{d} \sum_{i=1}^r \aa_i.
	\end{equation}
\end{conjecture}
We do not expect the conditions in Conjecture \ref{Conjecture} to be optimal. For instance, the Gaussianity assumption can likely be relaxed, as the simulations of \citet[Section 7.1]{BG}  indicate, and it is plausible that $\rank(\Gamma_i)<C$ condition can be replaced with slow growth of $\rank(\Gamma_i)$, as in Theorem \ref{Theorem_empirical}. Nevertheless, we wanted to record Conjecture \ref{Conjecture} in the present form, as a precise statement to be addressed in the future work. We are not giving a proof of Conjecture \ref{Conjecture} here: the required mathematical apparatus does not exist so far. Instead, we are going to provide a heuristic argument for its validity based on our recent results in \citet{BG_CCA} in a related, yet different setting.

\citet{BG_CCA} studied the following general setting: let $\UU$ and $\VV$ be two random linear subspaces in $S$--dimensional space with $\dim(\UU)=K$, $\dim(\VV)=M$ and all three numbers $K,M,S$ assumed to be growing to infinity. In addition, suppose that there are $\mathbbm q$ special vectors $\u_1,\dots,\u_{\mathbbm q}$ inside $\UU$ and other $\mathbbm q$ special vectors $\v_1,\dots,\v_{\mathbbm q}$ inside $\VV$, where $\mathbbm q$ is assumed to stay finite as other parameters grow. We directly observe $\UU$ and $\VV$, but not $\u_1,\dots,\u_{\mathbbm q}$ or $\v_1,\dots,\v_{\mathbbm q}$.  Can we reconstruct $\u_1,\dots,\u_{\mathbbm q}$, $\v_1,\dots,\v_{\mathbbm q}$, or at least identify their presence by looking at the squared sample canonical correlations between $\UU$ and $\VV$ and corresponding vectors?

The connection to our cointegration tests comes from taking as $\UU$ the space spanned by the $N$ rows of $\tilde{R}_0$, as defined after \eqref{res_BG}, and as $\VV$ the space spanned by the rows of $\tilde R_k$. The value ${\mathbbm q}$ corresponds to the cointegration rank and $\v_1,\dots,\v_{\mathbbm q}$ correspond to the cointegrating relationships.

While any finite $\mathbbm q$ can be analyzed in a similar fashion, let us stick to $\mathbbm q=1$ case for simplicity, so that we have a single vector $\u\in\UU$ and another vector $\v\in\VV$. The most important quantity is the sample squared correlation coefficient $r^2$ between vectors $\u$ and $\v$. It turns out that if $r^2$ is large (i.e., close to $1$, because $0\le r^2 \le 1$), then the largest canonical correlation between $\UU$ and $\VV$ is clearly separated from the rest (reminiscent of Figure \ref{small_rk_pic}) and the corresponding eigenvectors can be used to extract information on $\u$ and $\v$. On the other hand, if $r^2$ is small, then the histogram of the canonical correlations does not have such a spiked eigenvalue and all the information about $\u$ and $\v$ is washed out.
\citet[Theorem 2.5, Theorem 3.2, Theorem 3.3, Theorem 3.4]{BG_CCA} proved the existence of $r^2_{\rm{critical}}\in(0,1)$ separating the above two regimes for a variety of settings for the data generating process for $\UU$, $\VV$, $\u$, and $\v$, see also \citet{bao2019canonical,yang2022limiting}. However, the results of \cite{BG_CCA} do not address the setting relevant to cointegration and further new ideas would be necessary to find the value of $r^2_{\rm{critical}}$ for cointegration or rigorously prove its existence. Nevertheless, because the cointegration testing is also based on canonical correlations, one expects that the same phenomenology is true for it and, therefore, there should be the following dichotomy:
\begin{enumerate}
 \item If the linear subspace (in $T$--dimensional space) spanned by the $N$ rows of $\tilde{R}_0$ (as defined after \eqref{res_BG}) has a special vector $\u$ and the linear subspace spanned by rows of $\tilde{R}_k$ has a special vector $\v$, such that the sample squared correlation coefficient between $\u$ and $\v$ is atypically large compared to correlation coefficients of other vectors (e.g., if it is close to $1$), then the histogram of all squared canonical correlations would have a spike as in Figure \ref{small_rk_pic} and we should be able to reject the null of no cointegration.
 \item Otherwise, there would be no spikes (e.g., as in Figure \ref{SP100pic}) and we expect validity of asymptotics as in \eqref{eq_statistic_limit} and \eqref{eq_statistic_limit_repeat} consistent with the hypothesis of no cointegration.
\end{enumerate}

We now present heuristics in favor of Conjecture \ref{Conjecture} based on this dichotomy. Some of the technical details are omitted as we try to express the key ideas instead.

\begin{proof}[Heuristics for Conjecture \ref{Conjecture}]
 For simplicity of the presentation we stick to the case $k=2$, take the covariance matrix $\Lambda$ to be identical, set $\mu=0$, and let $\Gamma_1$ to be a matrix, which has $\theta$ in the upper-left corner and $0$ everywhere else. Clearly, $\rank(\Gamma_1)=1$ and the only root of \eqref{eq_char_eq} is $1/\theta$, hence, the third condition in the statement of Conjecture \ref{Conjecture} turns into $|\theta|<1$.

 Note that if we look only at the last $(N-1)$ out of $N$ coordinates of $X_t$, then we are in the setting of Theorem \ref{Theorem_J_stat} and asymptotics \eqref{eq_statistic_limit} holds. In particular, the largest canonical correlation is not separated from the rest. Hence, we only need to investigate how the addition of the special first row changes the situation. We will rely on the above dichotomy for our assessment. There are two ways how the addition of the first coordinate changes the setting compared to the situation when it did not exist (and $N$ was smaller by $1$):
 \begin{enumerate}
  \item The matrices $\tilde{R}_0$ and $\tilde R_k$ have a new first row each. Hence, we should check how large is the correlation between these first rows, if they are viewed as the special vectors $\u$ and $\v$.
  \item We projected the data orthogonally to $\tilde{Z}_{1t}$ in Step 3 of the procedure, see \eqref{res_BG}. The $\tilde{Z}_{1t}$ matrix also has a new first row, hence, we are now decreasing the dimension by $1$ via projecting orthogonally to an additional vector.
 \end{enumerate}
 Let $y_t$, $t=1,2,\dots,T$ denote the first coordinate of $X_t$. It solves the scalar recurrence
 \begin{equation}
 \label{eq_x10}
  \Delta y_t = \theta \Delta y_{t-1} + \xi_t,
 \end{equation}
 where $\xi_t$ is the first coordinate of $\eps_t$, and therefore a Gaussian $\mathcal N(0,1)$ random variable, i.i.d.\ in time $t$. Iterating \eqref{eq_x10}, we get
 \begin{equation}
  \Delta y_t= \theta^{t}(y_0-y_{-1})+\sum_{\tau=1}^t \theta^{t-\tau} \xi_\tau , \qquad y_t=y_0 + \frac{\theta-\theta^{t+1}}{1-\theta}(y_0-y_{-1})+\sum_{\tau=1}^t  \frac{1-\theta^{t+1-\tau}}{1-\theta}\xi_\tau.
 \end{equation}
 Next, we make the detrending of Step 1 in Procedure \ref{sscc_BG} of Section \ref{Section_modified_test}. As in \eqref{eq_detrending}, we define
 \begin{multline}
 \tilde y_t = y_{t-1} - \frac{t-1}{T} (y_T-y_0)\\= y_0 + \frac{\theta-\theta^{t}}{1-\theta}(y_0-y_{-1})+\sum_{\tau=1}^{t-1} \frac{1-\theta^{t-\tau}}{1-\theta}\xi_\tau-  \frac{t-1}{T} \left(\frac{\theta-\theta^{T+1}}{1-\theta}(y_0-y_{-1})+\sum_{\tau=1}^T  \frac{1-\theta^{T+1-\tau}}{1-\theta}\xi_\tau\right).
\end{multline}
Recalling cyclic shifts of Step 2 in Procedure \ref{sscc_BG}, we set
$$
 \tilde z_{t}=\tilde y_{t-1},\quad 2\le t \le T, \qquad \tilde z_{1}=\tilde y_T.
$$
The vector $\tilde z_{t}$, $1\le t\le T$, is the first row of $\tilde Z_{2t}$, $1\le t\le T$, viewed as an $N\times T$ matrix. Simultaneously, the vector $\Delta y_t$, $1\le t\le T$, is the first row of $\tilde Z_{0t}$. Recalling Step 3 in Procedure \ref{sscc_BG} and its restatement in terms of projectors at the end of Section \ref{Section_modified_test_statistics}, we analyze the sample correlation coefficients between the vectors $\tilde z_t$ and $\Delta y_t$ projected orthogonally to the constant vector and $\tilde Z_{1t}$. The first row of $\tilde Z_{1t}$ is $\Delta y_{t-1}$ (with cyclic shift of index, so that the $t=1$ coordinate is actually $\Delta y_T$). Hence, we are allowed to subtract multiples of the constant vector and multiples of $\Delta y_{t-1}$ from either of $\tilde z_t$ or $\Delta y_t$ without changing the desired sample correlation coefficient. Therefore, subtracting multiples of constants, we replace $\tilde z_t$ with a vector whose $t$-th coordinate for $2\le t \le T$ is
\begin{equation}
\label{eq_x11}
 \frac{-\theta^{t-1}-\frac{t-T/2}{T}(1-\theta^{T+1})}{1-\theta}(y_0-y_{-1})+\sum_{\tau=1}^{t-2} \frac{1-\theta^{t-\tau-1}}{1-\theta}\xi_\tau-  \frac{t-2}{T} \left(\sum_{\tau=1}^T  \frac{1-\theta^{T+1-\tau}}{1-\theta}\xi_\tau\right),
\end{equation}
and the first coordinate is given by a similar expression, which we omit. By subtracting $\theta \Delta y_{t-1}$, we replace $\Delta y_t$ with a vector whose $t$-th coordinate for $2\le t \le T$ is simply $\xi_t$

\smallskip

\noindent {\bf Claim.} The squared sample correlation coefficient between the vectors \eqref{eq_x11} and $\xi_t$ tends to $0$ as $T\to\infty$.

\smallskip

The claim follows from three computations:
\begin{enumerate}
\item[A)] The scalar product between these two vectors grows as $O(T)$.
\item[B)] The scalar product of \eqref{eq_x11} with itself is of order $T^2$.
\item[C)] The scalar product of the vector $\xi_t$, $1\le t\le T$, with itself is of order $T$.
\end{enumerate}

Each computation is a straightforward application of the Law of Large Numbers and Central Limit Theorem for i.i.d.\ random variables and we leave the details to the reader. We only note that the condition $|\theta|<1$ and boundness of $y_0-y_{-1}$ are both used here. Together, these computations imply that the sample correlation coefficient is of order $O(T^{-1})$, thus proving the claim.

\bigskip

In order to further pass from the two vectors in the claim to the first rows of the two matrices $\tilde{R}_0$ and $\tilde R_k$, we need to project orthogonally to the constant vector, to the vector $\Delta y_{t-1}$ (which is the first row of $\tilde Z_{1t}$), and to the remaining $(N-1)$ rows of $N\times T$ matrix $\tilde Z_{1t}$, $1\le t \le T$. One can check that projecting orthogonally to the first two vectors does not change the conclusion of the claim --- this is simply because these two vectors are very close to being orthogonal to the vectors of the claim. Showing that projecting orthogonally to the last $N-1$ rows of $\tilde Z_{1t}$ preserves the conclusion of the claim is a more challenging computation, which we record in the following abstract lemma, which is proven later.

\begin{lemma} \label{Lemma_projecting_two_vectors}
 Suppose that as $T\to\infty$ we are given a $T$--dimensional space and the following random data inside it: two vectors $\ba$ and $\bb$, such that the angle\footnote{Note that the cosine of the angle between $\ba$ and $\bb$ matches the sample correlation coefficient $\frac{\langle \ba,\bb\rangle}{\sqrt{\langle \ba,\ba\rangle \langle \bb,\bb\rangle}}$.} between them tends to $\pi/2$ as $T\to\infty$, and a linear subspace $\mathcal V$ of dimension $M$. We assume that the ratio $M/T$ tends to a number $\alpha$ such that $0<\alpha<1$ and that $\mathcal V$ is uniformly distributed among all subspaces of dimension $M$ and is independent of $\ba$ and $\bb$. Then the angle between orthogonal projections of $\ba$ and $\bb$ onto $\mathcal V$ tends to $\pi/2$ as $T\to\infty$.
\end{lemma}

The lemma is applicable in our situation, because the last $N-1$ rows of $\tilde Z_{1t}$, $1\le t \le T$, are formed by $(N-1)T$ i.i.d.~$\mathcal N(0,1)$ random variables, independent from $y_t$. Because of the invariance of the Gaussian law with identical covariance matrix under orthogonal transformations, the distribution of the space spanned by these $N-1$ rows is invariant under orthogonal transformations, which is the same as being uniformly distributed.

\medskip

The overall conclusion from the discussion is that the sample correlation coefficient between the new first rows of the matrices $\tilde{R}_0$ and $\tilde R_k$ tends to $0$ as $N,T\to\infty$. Hence, by the dichotomy, these rows can not be special vectors which cause the appearance of a spike in the histogram of eigenvalues. Therefore, we expect that \eqref{eq_statistic_limit_repeat} holds.
\end{proof}
\begin{remark}
 One additional effect which we have not examined in the above heuristics is that the last $N-1$ rows of $\tilde Z_{0}$ and $\tilde Z_2$ matrices should also be projected orthogonally to $\Delta y_{t-1}$ (in addition to projecting orthogonally to $\tilde Z_1$ covered by the setting of Theorem \ref{Theorem_J_stat}). Because $\Delta y_{t-1}$ is independent from the rest and is close to being orthogonal to every other vector entering into the procedure, we do not expect this effect to significantly change the asymptotics of the canonical correlations.
\end{remark}
We now come back to Lemma \ref{Lemma_projecting_two_vectors}.
\begin{proof}[Proof of Lemma \ref{Lemma_projecting_two_vectors}]
 Let $\mathcal W$ denote the two-dimensional space spanned by the vectors $\ba$ and $\bb$. Let us introduce canonical bases of spaces $\mathcal W$ and $\mathcal V$, see \citet[Chapter 12]{anderson1958introduction} or \citet[Section 11.3]{Muirhead_book} for the general introduction to canonical correlations and corresponding variables. Thus, we choose an orthonormal basis of $\mathcal W$, $\w_1,\w_2\in\mathcal W$ and an orthonormal basis of $\mathcal V$, $\v_1,\dots,\v_M\in \mathcal V$, such that $\langle \w_1, \v_1\rangle=c_1$, $\langle \w_2, \v_2 \rangle=c_2$ and all other scalar products $\langle \w_i,\v_j\rangle$ are zeros. We can assume without loss of generality that $1\ge c_1\ge c_2\ge 0$. These numbers are canonical correlations between spaces $\mathcal W$ and $\mathcal V$. Because the space $\mathcal W$ is uniformly distributed along all $M$--dimensional subspaces, the distribution of the squared correlations $(c_1^2,c_2^2)$ is explicit, it equals the distribution of eigenvalues of the Jacobi ensemble $\J(2; \frac{M-1}{2},\frac{T-M-1}{2})$, see  \citet[Corollary 11.3.3]{Muirhead_book}, \citet[Sections 2.1.1, 2.1.2]{Johnstone_Jacobi}, and references therein. This means that the joint density of $(c_1^2,c_2^2)$ denoted $\rho(x,y)$ is proportional to:
 \begin{equation}
  \rho(x,y)\sim (x-y)\, x^{ \frac{M-3}{2}}\, (1-x)^{\frac{T-M-3}{2}}\, y^{ \frac{M-3}{2}}\, (1-y)^{\frac{T-M-3}{2}}.
 \end{equation}
 As $T,M\to\infty$, the density $\rho(x,y)$ is sharply concentrated around its maximum. Hence, directly computing the asymptotics of $\rho(x,y)$ we find that
 \begin{equation}
  \label{eq_x12}
  \lim_{T\to\infty} (c_1^2, c_2^2)= \lim_{T\to\infty} \left(\frac{M}{T}, \frac{M}{T} \right)=(\alpha,\alpha).
 \end{equation}
 Let us expand $\ba$ and $\bb$ in $(\w_1,\w_2)$ basis:
 $$
  \ba= a_1 \w_1 + a_2 \w_2, \qquad \bb=b_1\w_1+ b_2 \w_2.
 $$
 The squared cosine of the angle between $\ba$ and $\bb$ is then computed as
 \begin{equation}
 \label{eq_x13}
  \frac{ (a_1 b_1 + a_2 b_2)^2}{( a_1^2+a_2^2)(b_1^2+b_2^2)}.
 \end{equation}
 The orthogonal projections of $\ba$ and $\bb$ onto $\mathcal V$ are
 $$
  \mathrm{proj}(\ba)= c_1 a_1 \v_1 + c_2 a_2 \v_2, \qquad \mathrm{proj}(\bb)= c_1 b_1 \v_1 + c_2 b_2 \v_2.
 $$
 Hence, the squared cosine of the angle between two projections is
 \begin{equation}
 \label{eq_x14}
  \frac{ (c_1^2 a_1 b_1 + c_2^2 a_2 b_2)^2}{( c_1^2 a_1^2+ c_2^2 a_2^2)(c_1^2 b_1^2+ c_2^2 b_2^2)}.
 \end{equation}
 Using \eqref{eq_x12}, it becomes clear that \eqref{eq_x14} tends to $0$ as $T\to\infty$ whenever \eqref{eq_x13} does.
\end{proof}

\subsection{Power}\label{appendix_power}

We proceed to our next computation, supplementing Conjecture \ref{Conjecture}. This time we would like to explain what changes in the asymptotics, if the data generating process satisfies the alternative $H_1$, rather than the null-hypothesis $H_0$. For the clarity of the exposition, we only concentrate on one particular instance of $k=1$ case here and deal with the $N$-dimensional data generating process
\begin{equation}
\label{eq_x15}
 \Delta X_t=\theta E_{11} X_{t-1}+\eps_t,\qquad t=1,\ldots,T,\qquad\qquad \text{where}
\end{equation}
$\eps_t \thicksim\text{i.i.d.}~\mathcal{N}(0,\Lambda)$, $E_{11}$ is the matrix with $1$ in top-left corner and $0$s everywhere else. $\theta$ is a real parameter, we set $\beta=1+\theta$, which implies that the first coordinate of $X_t$ is a scalar process $y_t$ solving
\begin{equation}\label{eq_ytH1}
 y_t=\beta y_{t-1} + \xi_t, \qquad \xi_t \text{ is the first coordinate of }\eps_t.
\end{equation}
Note that $\xi_t$ are i.i.d.\ $\mathcal N(0,\sigma^2)$ random variables for some constant $\sigma^2$.
Because in the notations of \eqref{var_k_restr} the matrix $\Pi$ now has rank $1$, one hopes that our test statistic of Section \ref{Section_modified_test_statistics} under \eqref{eq_x15} behaves significantly differently than under $H_0$ of Conjecture \ref{Corollary_test_stat}. This would imply that the no-cointegration test based on Theorem \ref{Theorem_J_stat} has high power against the rank one alternative \eqref{eq_x15}. Let us prove that this is indeed true for large values of the ratio $T/N$.
\begin{proposition} \label{Proposition_H1_inequality}
  In the notations of \eqref{eq_x15}--\eqref{eq_ytH1} assume that $|\beta|<1$ and let $\sigma^2$ be the variance of $\xi_t$. Let $\tilde \lambda_1\ge \dots \ge \tilde \lambda_N$ be eigenvalues of the matrix $\tilde{\mathcal C}$ from Section \ref{Section_modified_test_statistics} constructed using the $k=1$ procedure. For each $\epsilon>0$, we have
 \begin{equation}
 \label{eq_x16}
  \lim_{T\to\infty} \mathrm{Prob}\left(\tilde \lambda_1 > \frac{1}{\frac{2}{1-\beta}+\frac{1+\beta}{6\sigma^2}\left(y_T-y_0\right)^2}-\epsilon \right) =1,
 \end{equation}
 where $N$ can depend on $T$ in \eqref{eq_x16} in an arbitrary way.
\end{proposition}

As a corollary, we deduce that our cointegration test has a significant power against rank one stationary alternative, and this power tends to $1$ as $T/N\to\infty$. Here is a precise statement:

\begin{corollary}\label{Corollary_power}
 Suppose that $T,N\to\infty$ in such a way that  $\lim_{T,N\to\infty} \frac{T}{N}=\tau$. Fix a confidence level $0<\alpha<1$, $y_0$, $\sigma^2$, and $\beta=1+\theta$ such that $|\beta|<1$. Let $H_1$ be the data generating process \eqref{eq_x15}--\eqref{eq_ytH1}. Then the $k=1$ cointegration test based on Theorem \ref{Theorem_J_stat} has asymptotic power against $H_1$ at least $p(\alpha,\tau)$  as $T,N\to\infty$. Here $p(\alpha,\tau)$ is a non-negative function, such that for each $\alpha$ we have $\lim_{\tau\to\infty} p(\alpha,\tau)=1$.
\end{corollary}

Our approach to the proof of Corollary \ref{Corollary_power} gives a lower bound on $p(\alpha,\tau)$. Finding exact formulas for the power under $H_1$ of this corollary and under other alternatives remains an important open problem for the future research.

\begin{proof}[Proof of Proposition \ref{Proposition_H1_inequality}] By definition, $\hat \lambda_1$ is the largest  sample canonical correlations between matrices $\tilde{R}_0$ and $\tilde{R}_k$. The variational interpretation for $\hat \lambda_1$ (see, e.g.,  \citet[Chapter 12]{anderson1958introduction}) as the maximal sample correlation coefficient between vectors in linear spans of $\tilde R_0$ and $\tilde R_1$, implies that $\hat \lambda_1$ is larger or equal than the correlation coefficient between the first rows of $\tilde R_0$ and $\tilde R_k$. In the rest of the proof we estimate this correlation coefficient and show it is satisfies the asymptotic inequality \eqref{eq_x16}.

Let us compute these first rows by following the procedure of Section \ref{Section_modified_test_statistics}. From Eq.~\eqref{eq_ytH1} we obtain
$$
y_t=\beta^t y_0+\sum\limits_{i=1}^{t}\beta^{t-i}\xi_i,\qquad \Delta y_t=(\beta-1)\beta^{t-1} y_0 +(\beta-1)\sum\limits_{i=1}^{t-1}\beta^{t-1-i}\xi_i +\xi_t.
$$
Then
$$
\tilde{y}_t=y_{t-1}-\frac{t-1}{T}(y_T-y_0)=\beta^{t-1} y_0+\sum\limits_{i=1}^{t-1}\beta^{t-1-i}\xi_i-\frac{t-1}{T}(y_T-y_0).
$$
After regressing on a constant we get residuals (here $\tilde R_{0t,1}$ is the first element of the column $\tilde R_{0t}$ and similarly for $\tilde R_{kt,1}$)
\begin{align}
\tilde R_{0t,1}&=\Delta y_t-\frac1{T}\sum\limits_{\tau=1}^T \Delta y_{\tau}\label{eq_R0H1} \\ \notag
&=y_0\left((\beta-1)\beta^{t-1}+\frac{1-\beta^T}{T}\right) +\xi_t-\frac1{T}\sum\limits_{i=1}^{T}\xi_i
+(\beta-1)\sum\limits_{i=1}^{t-1}\beta^{t-1-i}\xi_i +\frac1{T}\sum\limits_{i=1}^{T}(1-\beta^{T-i})\xi_i,
\end{align}
\begin{align}
\label{eq_RkH1}
\tilde R_{kt,1}&=\tilde{y}_t-\frac1{T}\sum\limits_{\tau=1}^T \tilde{y}_{\tau}
=y_0\left(\beta^{t-1}-\frac{1-\beta^T}{T(1-\beta)}\right)
+\sum\limits_{i=1}^{t-1}\beta^{t-1-i}\xi_i-\frac1{T}\sum\limits_{i=1}^{T}\frac{1-\beta^{T-i}}{1-\beta}\xi_i
\\& -\left(\frac{2t-1}{2T}-\frac1{2}\right)(y_T-y_0). \notag
\end{align}
In order to compute the sample correlation coefficient, we analyze three sums representing sample variances and covariance: $\frac1{T}\sum\limits_{t=1}^{T}\tilde R_{0t,1}^2,\,\frac1{T}\sum\limits_{t=1}^{T}\tilde R_{kt,1}^2,\,\frac1{T}\sum\limits_{t=1}^{T}\tilde R_{0t,1}\tilde R_{kt,1}$. Let us analyze the sums sequentially. Summing the geometric series and using the law of large numbers, we get
\begin{equation}\begin{split}\label{eq_R0^2_terms_1}
&\frac{y_0^2}{T}\sum\limits_{t=1}^{T}\left((\beta-1)\beta^{t-1}+\frac{1-\beta^T}{T}\right)^2\xrightarrow[T\to\infty]{}0, \qquad
\frac1{T}\sum\limits_{t=1}^{T}\xi_t^2\xrightarrow[T\to\infty]{P}\sigma^2, \qquad
\left(\frac1{T}\sum\limits_{i=1}^{T}\xi_i\right)^2\xrightarrow[T\to\infty]{P}0,\\
&\frac1{T}\sum\limits_{t=1}^{T}\left(\sum\limits_{i=1}^{t-1}\beta^{t-1-i}\xi_i \right)^2
=\frac1{T}\sum\limits_{i=1}^{T-1}\xi_i^2\frac{1-\beta^{2(T-i)}}{1-\beta^2}+\frac2{T}\sum\limits_{t=1}^{T}\sum\limits_{i=1}^{t-2}\sum\limits_{j=i+1}^{t-1}\beta^{2(t-1)-i-j}\xi_i \xi_j
\xrightarrow[T\to\infty]{P}\frac{\sigma^2}{1-\beta^2},\\
&\left(\frac1{T}\sum\limits_{i=1}^{T}(1-\beta^{T-i})\xi_i\right)^2
=\frac1{T^2}\sum\limits_{i=1}^{T}(1-\beta^{T-i})^2\xi_i^2+\frac1{T^2}\sum\limits_{i\neq j}(1-\beta^{T-i})(1-\beta^{T-j})\xi_i \xi_j\xrightarrow[T\to\infty]{P}0.
\end{split}\end{equation}
We also have
\begin{equation}\label{eq_R0^2_terms_2}
 \mathbb E \left[\frac1{T}\sum\limits_{t=1}^{T}\xi_t \left(\sum\limits_{i=1}^{t-1}\beta^{t-1-i}\xi_i\right)\right]^2=\frac1{T^2}\sum\limits_{t=1}^{T}\sigma^2\sum\limits_{i=1}^{t-1}\beta^{2(t-1-i)}\sigma^2\xrightarrow[T\to\infty]{}0,
\end{equation}
which implies that the expression under expectation tends to $0$. Using formulas \eqref{eq_R0^2_terms_1},\eqref{eq_R0^2_terms_2} and Cauchy-Schwarz inequality to show that the remaining averages of cross-products of terms in Eq.~\eqref{eq_R0H1} converge to $0$, we get
\begin{equation}\label{eq_R0^2_lim}
\frac1{T}\sum\limits_{t=1}^{T}\tilde R_{0t,1}^2 \xrightarrow[T\to\infty]{P}\sigma^2+(\beta-1)^2\frac{\sigma^2}{1-\beta^2}= \frac{2\sigma^2}{1+\beta}.
\end{equation}
To analyze $\frac1{T}\sum\limits_{t=1}^{T}\tilde R_{kt,1}^2$, we again sum geometric series and use the law of large numbers:
\begin{equation}\label{eq_Rk^2_terms_1}
\frac{y_0^2}{T}\sum\limits_{t=1}^{T}\left(\beta^{t-1}-\frac{1-\beta^T}{T(1-\beta)}\right)^2 \xrightarrow[T\to\infty]{}0,\qquad
\frac{(y_T-y_0)^2}{T}\sum\limits_{t=1}^{T}\left(\frac{2t-1}{2T}-\frac1{2}\right)^2 \approx_{T\to\infty} \frac{(y_T-y_0)^2}{12},
\end{equation}
where the $\approx_{T\to\infty}$ sign means that the ratio of the left-hand side and the right-hand side tends to $1$ in probability.
It is also straightforward to show that
\begin{equation}
 \label{eq_Rk^2_terms_2}
 \frac{y_T-y_0}{T}\sum\limits_{t=1}^{T}\left[\left(\frac{2t-1}{2T}-\frac1{2}\right)\sum\limits_{i=1}^{t-1}\beta^{t-1-i}\xi_i \right] \xrightarrow[T\to\infty]{P}0.
\end{equation}
Using formulas \eqref{eq_Rk^2_terms_1}, \eqref{eq_Rk^2_terms_2}, second and third lines of formulas \eqref{eq_R0^2_terms_1}, and Cauchy-Schwarz inequality to show that the remaining averages of cross-products of terms in Eq.~\eqref{eq_RkH1} converge to $0$, we get
\begin{equation}\label{eq_Rk^2_lim}
\frac1{T}\sum\limits_{t=1}^{T}\tilde R_{kt,1}^2 \approx_{T\to\infty} \frac{\sigma^2}{1-\beta^2}+\frac1{12}\left(y_T-y_0\right)^2.
\end{equation}
We are left with the analysis of the covariance $\frac1{T}\sum\limits_{t=1}^{T}\tilde R_{0t,1}\tilde R_{kt,1}$, which relies on similar computations as for the two variances. The only asymptotically non-vanishing term is given by the computation of the second line in \eqref{eq_R0^2_terms_1}: we multiply $(\beta-1)\sum\limits_{i=1}^{t-1}\beta^{t-1-i}\xi_i$ from \eqref{eq_R0H1} by $\sum\limits_{i=1}^{t-1}\beta^{t-1-i}\xi_i$ and sum over $t$. Thus,
\begin{equation}\label{eq_R0Rk_lim}
\frac1{T}\sum\limits_{t=1}^{T}\tilde R_{0t,1} \tilde R_{kt,1} \xrightarrow[T\to\infty]{} (\beta-1)\frac{\sigma^2}{1-\beta^2}=-\frac{\sigma^2}{1+\beta}.
\end{equation}
Combining \eqref{eq_R0^2_lim}, \eqref{eq_Rk^2_lim}, \eqref{eq_R0Rk_lim} together we get
\begin{align}\label{eq_corr2_limit}
\bigl(\widehat{corr}(\tilde R_{0,1}, \tilde R_{k,1})\bigr)^2=\frac{\left(\frac1{T}\sum\limits_{t=1}^{T}\tilde R_{0t,1}\tilde R_{kt,1}\right)^2}{\left(\frac1{T}\sum\limits_{t=1}^{T}\tilde R_{0t,1}^2\right)\left(\frac1{T}\sum\limits_{t=1}^{T}R_{kt,1}^2\right)}
&\approx_{T\to\infty}\frac{\frac{\sigma^4}{(1+\beta)^2}}{\frac{2\sigma^2}{1+\beta}\left[\frac{\sigma^2}{1-\beta^2}+\frac1{12}\left(y_T-y_0\right)^2\right]}
\\&=\frac{1}{\frac{2}{1-\beta}+\frac{1+\beta}{6\sigma^2}\left(y_T-y_0\right)^2}.\qedhere
\end{align}
\end{proof}

\begin{proof}[Proof of Corollary \ref{Corollary_power}]
 Cointegration test based on Theorem \ref{Theorem_J_stat} has the form: reject $H_0$, if
 \begin{equation}
 \label{eq_rejection_event}
	 \frac{\sum_{i=1}^{r} \ln(1-\tilde{\lambda}_i)- r \cdot c_1(N,T)}{ N^{-2/3}  c_2(N,T)} \ge \kappa,
 \end{equation}
 where $\kappa$ is a constant depending on $r$ and the confidence level $\alpha$ ($\kappa$ is found from the equation $\mathrm{Prob}(\sum_{i=1}^r \aa_i\le\kappa)=\alpha$). In order to prove Corollary \ref{Corollary_power}, we need to find the probability of the event \eqref{eq_rejection_event} under $H_1$ given by \eqref{eq_x15}--\eqref{eq_ytH1}, and show that this probability tends to $1$ in the double limit in which we first send $T,N\to\infty$ with $\lim \frac{T}{N}=\tau$ and then send $\tau$ to infinity.

 Recall that $c_2(N,T)$ is negative, as given in \eqref{eq_constants_c1_c2}. Hence, using deterministic inequalities $\ln(1-\tilde{\lambda}_i)\le 0$ and $\ln(1-\tilde{\lambda}_1)\le-\tilde{\lambda}_1$, we conclude that the probability of the event \eqref{eq_rejection_event} is larger than the probability of a simpler event
  \begin{equation}
 \label{eq_rejection_event_2}
	 \tilde{\lambda}_1\ge - r \cdot c_1(N,T)-  \kappa N^{-2/3}  c_2(N,T).
 \end{equation}
 Note that both terms in the right-hand side of \eqref{eq_rejection_event_2} are positive and the second one vanishes as $N,T\to\infty$. As for the first one, using \eqref{eq_constants_c1_c2}, we see that it converges to a positive constant as $N,T\to\infty$ with $\lim \frac{T}{N}= \tau$, and this constant further tends to $0$ as $\tau\to\infty$. The conclusion is that the right-hand side of \eqref{eq_rejection_event_2} tends to $0$ in our double limit.

 On the other hand, under $H_1$ by Proposition \ref{Proposition_H1_inequality}, for any $\epsilon>0$, with probability tending to $1$ as $T\to\infty$, we have
 \begin{equation}
    \label{eq_lower_bound}
   \tilde \lambda_1 > \frac{1}{\frac{2}{1-\beta}+\frac{1+\beta}{6\sigma^2}\left(y_T-y_0\right)^2}-\epsilon.
 \end{equation}
 Note that $y_0$ is assumed to be bounded. Simultenously, we assumed $|\beta|<1$, and therefore, $y_T$, which due to \eqref{eq_ytH1} can be expressed as
 $$
 y_T=\beta^T y_0 + \sum_{t=1}^T \beta^{T-t} \xi_t,
 $$
 has uniformly bounded second moment. Therefore, the denominator in \eqref{eq_lower_bound} does not explode. Hence, \eqref{eq_lower_bound} implies that \eqref{eq_rejection_event_2} holds with probability tending to $1$ in the double limit.
\end{proof}

\bibliographystyle{abbrvnat}
\bibliography{coint_var_k_bib}

\begin{thebibliography}{48}
\providecommand{\natexlab}[1]{#1}
\providecommand{\url}[1]{\texttt{#1}}
\expandafter\ifx\csname urlstyle\endcsname\relax
  \providecommand{\doi}[1]{doi: #1}\else
  \providecommand{\doi}{doi: \begingroup \urlstyle{rm}\Url}\fi

\bibitem[Anderson et~al.(2010)Anderson, Guionnet, and
  Zeitouni]{anderson2010introduction}
G.~Anderson, A.~Guionnet, and O.~Zeitouni.
\newblock \emph{An introduction to random matrices}.
\newblock Cambridge university press, 2010.

\bibitem[Anderson(1951)]{Anderson}
T.~W. Anderson.
\newblock Estimating linear restrictions on regression coefficients for
  multivariate normal distributions.
\newblock \emph{Annals of Mathematical Statistics}, 22\penalty0 (3):\penalty0
  327--351, 1951.

\bibitem[Anderson(2003)]{anderson1958introduction}
T.~W. Anderson.
\newblock \emph{Introduction to multivariate statistical analysis, 3rd
  edition}.
\newblock John Wiley \& Sons, 2003.

\bibitem[Bai and Ng(2008)]{Bai_Ng_2008}
J.~Bai and S.~Ng.
\newblock Large dimensional factor analysis.
\newblock \emph{Foundations and Trends in Econometrics}, 3\penalty0
  (2):\penalty0 89--163, 2008.

\bibitem[Baik et~al.(2005)Baik, {Ben Arous}, and P\'{e}ch\'{e}]{BBP}
J.~Baik, G.~{Ben Arous}, and S.~P\'{e}ch\'{e}.
\newblock Phase transition of the largest eigenvalue for nonnull complex sample
  covariance matrices.
\newblock \emph{The Annals of Probability}, 33\penalty0 (5):\penalty0
  1643--1697, 2005.

\bibitem[Bao et~al.(2019)Bao, Hu, Pan, and Zhou]{bao2019canonical}
Z.~Bao, J.~Hu, G.~Pan, and W.~Zhou.
\newblock Canonical correlation coefficients of high-dimensional gaussian
  vectors: Finite rank case.
\newblock \emph{Annals of Statistics}, 47\penalty0 (1):\penalty0 612--640,
  2019.

\bibitem[Breitung and Pesaran(2008)]{Breitung_Pesarann_2008}
J.~Breitung and M.~H. Pesaran.
\newblock Unit roots and cointegration in panels.
\newblock In \emph{M\'{a}ty\'{a} L., Sevestre P. (eds) The Econometrics of
  Panel Data. Advanced Studies in Theoretical and Applied Econometrics, vol.\
  46}, pages 279--322. Springer, Berlin, Heidelberg, 2008.

\bibitem[Bykhovskaya and Gorin(2022)]{BG}
A.~Bykhovskaya and V.~Gorin.
\newblock Cointegration in large vars.
\newblock \emph{Annals of Statistics}, 2022.

\bibitem[Bykhovskaya and Gorin(2023)]{BG_CCA}
A.~Bykhovskaya and V.~Gorin.
\newblock High-dimensional canonical correlation analysis.
\newblock \emph{arXiv preprint arXiv:2306.16393}, 2023.

\bibitem[Bykhovskaya et~al.(2023)Bykhovskaya, Gorin, and
  Kiss]{vignette_largevars}
A.~Bykhovskaya, V.~Gorin, and E.~Kiss.
\newblock Largevars: an {R} package for testing large {VAR}s for the presence
  of cointegration.
\newblock 2023.
\newblock \url{https://github.com/eszter-kiss/Largevars}.

\bibitem[Cavaliere et~al.(2012)Cavaliere, Rahbek, and Taylor]{cav_et_all}
G.~Cavaliere, A.~Rahbek, and A.~R. Taylor.
\newblock Bootstrap determination of the co-integration rank in vector
  autoregressive models.
\newblock \emph{Econometrica}, 80\penalty0 (4):\penalty0 1721--1740, 2012.

\bibitem[Choi(2015)]{choi15}
I.~Choi.
\newblock Panel cointegration.
\newblock In \emph{Baltagi B.H. (eds) The Oxford handbook of panel data}.
  Oxford University Press, 2015.

\bibitem[Dumitriu and Edelman(2002)]{dumitriu_edelman}
I.~Dumitriu and A.~Edelman.
\newblock Matrix models for beta ensembles.
\newblock \emph{Journal of Mathematical Physics}, 43\penalty0 (11):\penalty0
  5830--5847, 2002.

\bibitem[Erdos and Yau(2012)]{ErdosYau}
L.~Erdos and H.~T. Yau.
\newblock Universality of local spectral statistics of random matrices.
\newblock \emph{Bulletin of the American Mathematical Society}, 49\penalty0
  (3):\penalty0 377–--414, 2012.

\bibitem[Forrester(1993)]{Forrest_spectr}
P.~J. Forrester.
\newblock The spectrum edge of random matrix ensembles.
\newblock \emph{Nuclear Physics B}, 402\penalty0 (3):\penalty0 709--728, 1993.

\bibitem[Forrester(2010)]{forrest}
P.~J. Forrester.
\newblock \emph{Log-gases and random matrices}.
\newblock Princeton University Press, 2010.

\bibitem[Gonzalo and Pitarakis(1995)]{gonzalo_pitarakis1995}
J.~Gonzalo and J.~Y. Pitarakis.
\newblock Comovements in large systems.
\newblock \emph{Statistics and Econometrics Series, Vol.\ 10. Working Paper
  95-38, Universidad Carlos III de Madrid}, 1995.

\bibitem[Gonzalo and Pitarakis(2002)]{gonzalo_pitarakis2002}
J.~Gonzalo and J.~Y. Pitarakis.
\newblock Lag length estimation in large dimensional systems.
\newblock \emph{Journal of Time Series Analysis}, 23\penalty0 (4):\penalty0
  401--423, 2002.

\bibitem[Han et~al.(2018)Han, Pan, and Yang]{HanPanYang}
C.~Han, G.~Pan, and Q.~Yang.
\newblock A unified matrix model including both cca and f matrices in
  multivariate analysis: The largest eigenvalue and its applications.
\newblock \emph{Bernoulli}, 24\penalty0 (4B):\penalty0 3447--3468, 2018.

\bibitem[Han et~al.(2016)Han, Pan, and Zhang]{HanPanZhang_2016}
X.~Han, G.~M. Pan, and B.~Zhang.
\newblock The tracy-widom law for the largest eigenvalue of f type matrix.
\newblock \emph{The Annals of Statistics}, 44\penalty0 (4):\penalty0
  1564--1592, 2016.

\bibitem[Horn and Johnson(2013)]{Horn_Johnson}
R.~A. Horn and C.~R. Johnson.
\newblock \emph{Matrix Analysis}.
\newblock Cambridge University Press, 2nd edition, 2013.

\bibitem[Hua(1963)]{Hua}
L.~Hua.
\newblock \emph{Harmonic analysis of functions of several complex variables in
  the classical domains}, volume~6.
\newblock American Mathematical Soc., 1963.

\bibitem[Johansen(1988)]{johansen1988}
S.~Johansen.
\newblock Statistical analysis of cointegrating vectors.
\newblock \emph{Journal of Economic Dynamics and Control}, 12\penalty0
  (2--3):\penalty0 231--254, 1988.

\bibitem[Johansen(1991)]{johansen1991}
S.~Johansen.
\newblock Estimation and hypothesis testing of cointegration vectors in
  gaussian vector autoregressive models.
\newblock \emph{Econometrica}, 59:\penalty0 1551--1580, 1991.

\bibitem[Johansen(1995)]{johansen_book}
S.~Johansen.
\newblock \emph{Likelihood-based inference in cointegrated vector
  autoregressive models}.
\newblock Oxford University Press, 1995.

\bibitem[Johansen(2002)]{johansen_correction}
S.~Johansen.
\newblock A small sample correction for the test of cointegrating rank in the
  vector autoregressive model.
\newblock \emph{Econometrics}, 70\penalty0 (5):\penalty0 1929--1961, 2002.

\bibitem[Johnstone(2008)]{Johnstone_Jacobi}
I.~Johnstone.
\newblock Multivariate analysis and jacobi ensembles: largest eigenvalue,
  tracy-widom limits and rates of convergence.
\newblock \emph{Annals of statistics}, 36\penalty0 (6):\penalty0 2638--2716,
  2008.

\bibitem[Juselius(2006)]{juselius}
K.~Juselius.
\newblock \emph{The Cointegrated VAR Model: Methodology and Applications}.
\newblock Oxford University Press, 2006.

\bibitem[Keilbar and Zhang(2021)]{crypto_coint_paper}
G.~Keilbar and Y.~Zhang.
\newblock On cointegration and cryptocurrency dynamics.
\newblock \emph{Digital Finance}, 3\penalty0 (1):\penalty0 1--23, 2021.
\newblock
  \url{https://github.com/QuantLet/CryptoDynamics/blob/master/CryptoDynamics_Series/price.csv}.

\bibitem[Maddala and Kim(1998)]{maddala}
G.~S. Maddala and I.-M. Kim.
\newblock \emph{Unit Roots, Cointegration, and Structural Change}.
\newblock Cambridge University Press, 1998.

\bibitem[Meckes and Meckes(2013)]{MeckesMeckes}
E.~Meckes and M.~Meckes.
\newblock Spectral measures of powers of random matrices.
\newblock \emph{Electronic communications in probability}, 18, 2013.

\bibitem[Muirhead(2009)]{Muirhead_book}
R.~J. Muirhead.
\newblock \emph{Aspects of multivariate statistical theory}.
\newblock John Wiley $\&$ Sons, 2009.

\bibitem[Neretin(2002)]{Neretin_Hua}
Y.~Neretin.
\newblock Hua-type integrals over unitary groups and over projective limits of
  unitary groups.
\newblock \emph{Duke Mathematical Journal}, 114\penalty0 (2):\penalty0
  239--266, 2002.

\bibitem[Olshanski(2003)]{Olshanski_Harmonic}
G.~Olshanski.
\newblock The problem of harmonic analysis on the infinite-dimensional unitary
  group.
\newblock \emph{Journal of Functional Analysis}, 205\penalty0 (2):\penalty0
  464--524, 2003.

\bibitem[Onatski and Wang(2018)]{onatski_ecta}
A.~Onatski and C.~Wang.
\newblock Alternative asymptotics for cointegration tests in large vars.
\newblock \emph{Econometrica}, 86\penalty0 (4):\penalty0 1465--1478, 2018.

\bibitem[Onatski and Wang(2019)]{onatski_joe}
A.~Onatski and C.~Wang.
\newblock Extreme canonical correlations and high-dimensional cointegration
  analysis.
\newblock \emph{Journal of Econometrics}, 2019.

\bibitem[Pagan(1987)]{pagan}
A.~Pagan.
\newblock Three econometric methodologies: A critical appraisal.
\newblock \emph{Journal of Economic surveys}, 1\penalty0 (1):\penalty0 3--23,
  1987.

\bibitem[Reinsel and Ahn(1992)]{Reinsel_Ahn}
G.~C. Reinsel and S.~K. Ahn.
\newblock Vector autoregressive models with unit roots and reduced rank
  structure: Estimation. likelihood ratio test, and forecasting.
\newblock \emph{Journal of time series analysis}, 13\penalty0 (4):\penalty0
  353--375, 1992.

\bibitem[Sims(1980)]{sims}
C.~A. Sims.
\newblock Macroeconomics and reality.
\newblock \emph{Econometrica}, 48\penalty0 (1):\penalty0 1--48, 1980.

\bibitem[Swensen(2006)]{Swensen}
A.~R. Swensen.
\newblock Bootstrap algorithms for testing and determining the cointegration
  rank in var models.
\newblock \emph{Econometrica}, 74\penalty0 (6):\penalty0 1699--1714, 2006.

\bibitem[Tao and Vu(2012)]{Tao_Vu}
T.~Tao and V.~Vu.
\newblock Random matrices: the universality phenomenon for wigner ensembles.
\newblock \emph{Modern aspects of random matrix theory}, 72:\penalty0
  121–--172, 2012.

\bibitem[Tracy and Widom(1996)]{Tracy_Widom}
C.~A. Tracy and H.~Widom.
\newblock On orthogonal and symplectic matrix ensembles.
\newblock \emph{Communications in Mathematical Physics}, 177\penalty0
  (3):\penalty0 727--754, 1996.

\bibitem[Udell and Townsend(2019)]{lowrk_highdim}
M.~Udell and A.~Townsend.
\newblock Why are big data matrices approximately low rank?
\newblock \emph{SIAM Journal on Mathematics of Data Science}, 1\penalty0
  (1):\penalty0 144--160, 2019.

\bibitem[Wang and Tsay(2022)]{Wang_Tsay_low_rank_VAR}
D.~Wang and R.~S. Tsay.
\newblock Rate-optimal robust estimation of high-dimensional vector
  autoregressive models.
\newblock \emph{arXiv preprint arXiv:2107.11002}, 2022.

\bibitem[Wang et~al.(2022)Wang, Zheng, Lian, and Li]{JASA_low_rank_VAR}
D.~Wang, Y.~Zheng, H.~Lian, and G.~Li.
\newblock High-dimensional vector autoregressive time series modeling via
  tensor decomposition.
\newblock \emph{Journal of the American Statistical Association}, 117\penalty0
  (539):\penalty0 1338--1356, 2022.

\bibitem[Yang(2022{\natexlab{a}})]{FanYang}
F.~Yang.
\newblock Sample canonical correlation coefficients of high-dimensional random
  vectors: Local law and tracy–widom limit.
\newblock \emph{Random Matrices: Theory and Applications}, 11\penalty0
  (1):\penalty0 2250007, 2022{\natexlab{a}}.

\bibitem[Yang(2022{\natexlab{b}})]{yang2022limiting}
F.~Yang.
\newblock Limiting distribution of the sample canonical correlation
  coefficients of high-dimensional random vectors.
\newblock \emph{Electronic Journal of Probability}, 27:\penalty0 1--71,
  2022{\natexlab{b}}.

\bibitem[Zhang et~al.(2018)Zhang, Pan, and Gao]{ZhangPanGao_2018}
B.~Zhang, G.~M. Pan, and J.~T. Gao.
\newblock Clt for largest eigenvalues and unit root testing for
  high-dimensional nonstationary time series.
\newblock \emph{The Annals of Statistics}, 46\penalty0 (5):\penalty0
  2186--2215, 2018.

\end{thebibliography}

\end{document}